\documentclass[twoside,11pt]{article}

\usepackage{amsthm}
\usepackage{jmlr2e} 

\usepackage[ruled,vlined]{algorithm2e}

\usepackage{enumitem} 

\usepackage[english]{babel}
\usepackage{mathtools}
\usepackage{amsmath,amssymb,dsfont, natbib}
\usepackage{epsfig}
\usepackage{color}

\usepackage[a4paper,margin=1.3in,footskip=0.25in]{geometry}
\usepackage{makeidx} 
\usepackage{bbm} 
\newcommand{\one}{\mathbbm{1}}

\definecolor{purple}{rgb}{0.55,0.2,0.90}
\definecolor{burntorange}{rgb}{0.8, 0.33, 0.0}
   
\newcommand{\mn}[2]{\widehat{m}_{#1,(#2)}}

\newcommand{\phie}{f_\epsilon}

\newcommand{\bs}{\boldsymbol}
\newcommand{\mc}{\mathcal}
\newcommand{\wh}{\widehat}
\newcommand{\wt}{\widetilde}
\newcommand{\iid}{\stackrel{iid}{\sim}}

\newcommand{\R}{\mathbb{R}}

\renewcommand{\hat}{\widehat}

\newcommand{\nx}{n_x}
\newcommand{\ny}{n_y}
\newcommand{\Ne}{N} 

\DeclareMathOperator*{\argmin}{argmin} 

\newtheorem*{assumption*}{\assumptionnumber}
\providecommand{\assumptionnumber}{}
\makeatletter 
\newenvironment{assumption}[1]
 {%
  \renewcommand{\assumptionnumber}{Assumption #1}%
  \begin{assumption*}%
  \protected@edef\@currentlabel{#1}%
 }
 {%
  \end{assumption*}
 }
\makeatother 

\newtheorem{theo}{Theorem}
 \newtheorem{prop}{Proposition}

\newcommand{\lp}{\left(} 
  \newcommand{\rp}{\right)}
\newcommand{\lb}{\left\{} 
  \newcommand{\rb}{\right\}}

\newcommand{\RR}{\mathbb{R}}
\newcommand{\FF}{{\mathbb F}}
\newcommand{\HH}{{\mathbb H}}

\newcommand{\NN}{\mathbb{N}}
\newcommand{\MM}{\mathbb{M}}

\global\long\def\inv#1{\frac{1}{#1}}

\DeclareMathOperator*{\Var}{Var}

\usepackage{lastpage}
\jmlrheading{22}{2021}{1-\pageref{LastPage}}{7/20; Revised
4/21}{7/21}{20-689}{Fadoua Balabdaoui, Charles R.\ Doss, C\'ecile Durot}
\ShortHeadings{Unlinked Monotone Regression}{Balabdaoui, Doss, and Durot}

\firstpageno{1}

\begin{document}

\title{Unlinked Monotone  Regression}
\author{\name Fadoua Balabdaoui \email fadoua.balabdaoui@stat.math.ethz.ch \\
  \addr Seminar f\"ur Statistik\\
          ETH, Zurich \\
          R\"amistrasse 101 \\
          8092, Zurich, Switzerland 
\AND
  \name Charles R.\ Doss \email cdoss@umn.edu \\
  \addr School of Statistics\\
  University of Minnesota\\
  Minneapolis, MN 55455, USA 
  \AND
 \name C\'ecile Durot \email cecile.durot@gmail.com \\
  \addr Modal'x \\
          Universit\'e Paris Nanterre \\
          F92000, Nanterre, France
  }

\editor{Gabor Lugosi}








  



\maketitle

\begin{abstract}%
  We consider so-called univariate unlinked (sometimes ``decoupled,'' or ``shuffled'') regression when the unknown regression curve is monotone.  In standard monotone regression, one observes a pair $(X,Y)$ where a response $Y$ is linked to a  covariate $X$ through the model $Y= m_0(X) + \epsilon$, with $m_0$ the (unknown) monotone regression function and $\epsilon$ the unobserved error (assumed to be independent of $X$).  In the unlinked regression setting one gets only to observe a vector of realizations from both the response $Y$ and from the covariate $X$ where now $Y \stackrel{d}{=} m_0(X) + \epsilon$.  There is no (observed) pairing of $X$ and $Y$.  Despite this, it is actually still possible to derive a consistent non-parametric estimator of $m_0$ under the assumption of monotonicity of $m_0$ and knowledge of the distribution of the noise $\epsilon$.  In this paper, we establish an upper bound on the rate of convergence of such an estimator under minimal assumption on the distribution of the covariate $X$.  We discuss extensions to the case in which the distribution of the noise is unknown.  We develop a second order algorithm for its computation, and we demonstrate its use on synthetic data.  Finally, we apply our method (in a fully data driven way, without knowledge of the error distribution) on longitudinal data from the US Consumer Expenditure Survey.
\end{abstract}


\tableofcontents

\section{Introduction}
\label{sec:introduction}

An important part of data science is the construction of a data set;
nowadays,
because there are so many different entities collecting increasing amounts of
data, data sets are often constructed by combining separate sub-data sets or
datastreams.  Also, data sets sometimes undergo some form of anonymization:
this can be due to the increasing prevalence of privacy concerns, or in some cases due to
concerns about having limited data-transmission bandwidth  when many separate sensors are streaming data to a central server  \citep{Pananjady:eo}.  Thus,
it is increasingly common for data scientists/analysts to want to relate variables in one
data set to variables in another data set when the two data sets are {\it unlinked}.  In this paper, we consider the problem of unlinked regression, specifically when the regression function is assumed to satisfy a monotonicity constraint. 

In the standard regression setting, we have
\begin{equation}
  \label{eq:standard regression}
   Y_i = m_0(X_i) + \epsilon_i,  \qquad E(\epsilon_i) = 0,  \qquad i=1,\ldots, n,
\end{equation}
for a random noise variable $\epsilon_i$ that is independent of $X_i$. The most basic assumption of this model is that for each index $i=1, \ldots, n$, the {\it pair} $(X_i, Y_i)$ is observed.  For now, we assume that the covariates $X_i, \ i=1,\ldots, n $ are univariate random variables.  A more general model than the above standard regression model is the {\it shuffled} regression, in which we do not get to see the pairs $(X_i, Y_i)$, $i=1, \ldots, n$; rather, we only observe $(X_1,\ldots, X_n)$ and $(Y_1, \ldots, Y_n)$, without knowing which $X_i$ is paired with which $Y_i$.  Thus, we have the same model as \eqref{eq:1} except that the first equality is now only an equality in distribution and there is an unknown permutation $\pi$ on $\{1,\dots,n\}$ such that $Y_i=m_0(X_{\pi(i)})+\epsilon_i$ for all $i$. This happens for instance in case of anonymized data. An even more general model is the {\it unlinked} regression model that we consider in this paper, where again, the first equality only holds in distribution but where in addition, the $X_i$'s could be observed on different individuals from the $Y_i$'s so that the two samples are not necessarily connected through a permutation $\pi$. This happens for instance if the two samples have been collected independently (by separate entities). The number of observed $X_i$'s may even differ from the number of observed $Y_i$'s so we observe independent and identically distributed (i.i.d.) variables $Y_1,\dots,Y_{\ny}$ and i.i.d.\ variables $X_1,\dots,X_{\nx}$ such that 
\begin{equation}
  \label{eq:2}
  Y \stackrel{d}{=}  m_0(X)  + \epsilon
\end{equation}
where $Y \sim Y_i$, $X \sim X_i$,
and $\epsilon$ is independent of $X$ with distribution function $\Phi_{\epsilon}$.  In shuffled or unlinked regression models, it may seem hopeless to even try to learn the regression function $m_0$ since $m_0$ is in general not even identifiable (see Section \ref{sec:setup-term-notat} below); however, it turns out that if $m_0$ is assumed to be monotonically increasing, as we will do in this paper, and if the error distribution is known a priori, then in fact one can construct consistent estimators of $m_0$. (We will return in Section~\ref{sec: ext} to discuss  the case where the error distribution is unknown.)

Unlinked estimation can be considered in a variety of settings (e.g., \cite{DeGroot:1980ko}).   It appears that unlinked monotone regression was recently introduced by \cite{carp16}. 
One motivating example discussed in \cite{carp16} is about expenditure on goods or services, such as housing: the price an individual is willing to pay for housing is expected to be monotonically increasing (at least on average, if not individually) as a function of the individual's salary.
 However, estimating the monotonic relationship is hindered by the simple fact that the data on wages and  housing transactions are often gathered by different agents.
There are many other motivating examples for unlinked or shuffled regression, besides the ones already discussed.
 (In some examples, there may be information allowing partial matching; see our discussion point \ref{item:2} in Section~\ref{sec: concl}.)
 In flow cytometry, cells suspended in a fluid flow past a laser, and the response (scattering of light) reveals information about the cell, which may be explained by its features (e.g., gene expression).  However, the order of the cells as they pass the laser is unknown, so we are in a shuffled regression setting \citep{Abid:2017ws}.
 In ``image stitching'' (related to the so-called pose-correspondence problem) one wants to find the unknown correspondence between point clouds constructed from multiple camera angles of the same image
 \citep{Pananjady:2018hd}.
 Several other motivating examples are discussed by
 \cite{Pananjady:2018hd} (in the context of permuted/shuffled linear regression).
 
The method of \cite{carp16} is based on the fact that monotone unlinked
regression can be rephrased as the so-called deconvolution problem, as is
apparent from (\ref{eq:2}).  Below, we will discuss links to this problem in
more detail.
Note that, as is also true in the deconvolution setting,  $m_0$ is not identifiable if we do not know the distribution of the noise $\epsilon$ (or have at least an estimate thereof). The following simple example explains why. Suppose that $m_0(x) = 2 x$  and $X \sim \mathcal{N}(0,1)$ is independent of  $ \epsilon  \sim \mathcal{N}(0, 1)$. Then, this model is the same as $m_0(x) = x$ and $X \sim \mathcal{N}(0,1)$ independent of $\epsilon   \sim \mathcal{N}(0,4)$.
Let $F_0$ be the cumulative distribution function (CDF) of $X$ and $L_0$ the CDF of $m_0(X)$.  If both $m_0$  and $F_0$ are assumed to be one-to-one, then $L_0$ satisfies
\begin{equation}
  \label{eq:3}
  L_0(w)    = P(m_0(X)  \le w)  =  F_0 \circ m^{-1}_0(w)
\end{equation}
for all $w$ in the range of $m_0$, and we have
\begin{equation}\label{eq: mLF}
  m _0 = L^{-1}_0 \circ F_0.
\end{equation}
Thus, the estimator constructed by \cite{carp16} takes the form
\begin{eqnarray}\label{Estimf}
  \widetilde m(x)  =  \widetilde{L}^{-1}_{\ny} \circ \mathbb{F}_{\nx}(x)
\end{eqnarray}
where $\widetilde{L}_{\ny}$ is an estimator of $L_0$ obtained by deconvolution methods and is based on the sample $(Y_1, \ldots, Y_{\ny})$ and knowledge of the distribution of $\epsilon$, and
where $\mathbb{F}_{\nx}$ can be taken for example to be the empirical distribution function of $F_0$ based on the sample $X_1,\dots,X_{\nx}$.
Thus, the estimator in (\ref{Estimf}) at a point $x$ is equal to the deconvolution estimator of the quantile of $m_0(X)$ corresponding to the  random level $ \mathbb{F}_{\nx}(x)$.  In the presence of ``contextual variables'' (i.e., covariates
 that are paired with both $X$ and $Y$), \cite{carp16}  gave some consistency and rate of convergence results in their Theorem 3.2 under the assumption that $m_0$  and $F_0^{-1}$ belong to H\"older classes.  \cite{carp16} do not discuss how to choose  the optimal bandwidth, as the main focus in
that paper is to show how their estimation approach can  be easily implemented for real data sets.

Beside \cite{carp16}, the only other work of which we are aware on a similar problem is the very recent article of \cite{Rigollet:2018wia}.  In their setting, the authors consider a shuffled monotone regression model  in a fixed design setting but use the term \lq\lq uncoupled\rq\rq \ to describe the problem.  In fact, the authors do not make any attempt to recover the unknown permutation, and focus entirely, as we do in this paper,  on estimating the unknown regression function.  \cite{Rigollet:2018wia} assume that the known distribution of the noise is sub-exponential and the true monotone function is bounded by some known constant, but they do not make any smoothness assumption on that function. Using the Wasserstein's distance and arguments from optimal transport, they showed that their estimator converges to the truth at a rate no slower than $\log \log n / \log n$ and that this rate is minimax when the distribution of the noise is Gaussian.  Although \cite{Rigollet:2018wia} describe in their Section 2.2 an algorithm for computing their monotone estimator, the authors do not present simulation results.

It is worth mentioning that more research seems to have been done in the shuffled (permuted) linear regression  model than in the monotone regression model. We refer here to the work of \cite{Abid:2017ws},  \cite{Pananjady:2017kj}, \cite{Pananjady:2018hd},  and \cite{Unnikrishnan:2018gp}. The main focus in the former three papers is to find conditions on the signal-to-noise ratio that guarantee recovery of the unknown permutation. \cite{Abid:2017ws} show that the least-squares estimator in this model is inconsistent in general, but they construct a method-of-moments type estimator and prove that this estimator is consistent assuming that $E(X_i) \ne 0$. A common feature of these works is to restrict attention to the case of Gaussian noise.   As we show in the present paper, the noise distribution may be of fundamental importance in these problems.

As announced above, we give now some more detail to the existing link between unlinked monotone regression and estimation in deconvolution problems. As noted earlier,  \eqref{eq: mLF} shows clearly that there is a tight connection to quantile estimation in a deconvolution setting.  In fact, consider the deconvolution setting in which one observes $n$ i.i.d.\ copies of $Y$ where  $Y = X + \epsilon$ for independent random variables $X$ and $\epsilon$. In this problem, the goal is learn the distribution of $X$ under the assumption that the distribution of $\epsilon$ is known (such an assumption can be relaxed if this distribution can be estimated).  Estimation of the distribution and quantile functions of $X$  has been considered in \cite{hall2008estimation}, and revisited in \cite{dattner2011deconvolution,dattner2016adaptive} under slightly different assumptions. In particular, contrary to \cite{hall2008estimation}, no moment assumptions are made about the covariate $X$ or $\epsilon$ in
  \cite{dattner2011deconvolution,dattner2016adaptive}. There, the smoothness of the density of $X$ is measured in terms of belonging to  Sobolev or H\"older balls.
In the case where the error is ordinary smooth of order larger than $1/2$, \cite{dattner2011deconvolution} recover the rates of convergence that \cite{hall2008estimation} obtained for the integrated risk when estimating the distribution function, and moreover  provide new rates of convergence for the case of smoother error distributions; the square-root rate is shown to be achieved for smooth enough distribution of $X$.  The convergence results obtained in these  previous papers do not apply directly in this present paper, as we do not
assume that the covariate $X$ (from \eqref{eq:2}) admits even a Lebesgue density, which also means that $m_0(X)$ is not assumed to have a density.

While  \cite{carp16} and \cite{Rigollet:2018wia} are the only other articles of which we are aware on unlinked  monotone regression besides the present paper,  the classical isotonic  regression model given in \eqref{eq:standard regression} is a very well-known estimation problem with a vast literature.  The most known estimator in this problem is certainly the Grenander-type estimator, obtained by taking the right derivative of the greatest convex minorant of the cumulative sum diagram associated with the data $(X_i, Y_i), i =1, \ldots, n$  \citep{BBBB,Robertson:1988vf,Groeneboom:2014hk}.  The pointwise non-standard rate of convergence of the Grenander estimator is $n^{-1/3}$ if $m_0$ is continuously differentiable with a non-vanishing derivative and $\sqrt n$ in case $m_0$ is locally flat \citep{Groeneboom:1983ud,carolan99,Zhang:2002ed,cator2011,Chatterjee:2015eza}. 
Asymptotic properties, including the pointwise limit distribution and convergence in the $L_p$-norms for $p \in [1, 5/2)\cup\{\infty\}$  have been fully described in \cite{brunk70},  \cite{durot2002}, \cite{durot2007} and \cite{durot2011}; see also  \cite{groe85,groe89}.  One can also combine kernel estimation with the monotonicity constraint to improve rates of convergence if $m_0$ has higher orders of smoothness \citep{mammen91,durot2014}.

In this paper, we investigate  the unlinked  monotone regression following  the method introduced by \cite{edelman} for estimating the mixing distribution in a mixture problem with Normal noise. Let $(X_1, \theta_1), \ldots, (X_n, \theta_n)$ be independent random vectors such that $X_i | \theta_i \sim \mathcal{N}(\theta_i, 1)$, that is, conditionally on $\theta_1, \ldots, \theta_n$ the random variables  $X_1, X_2, \ldots, X_n$  are generated from the unknown distributions $\Phi(. - \theta_1),  \Phi(. - \theta_2), \ldots, \Phi(\cdot - \theta_n)$, where $\Phi$ denotes the cumulative distribution function of $\mathcal{N}(0,1)$. The approach of \cite{edelman} consists of finding the vector $(\tilde{\theta}_1, \ldots, \tilde{\theta}_n)$ which minimizes the integrated difference between $n^{-1} \sum_{i=1}^n \Phi(\cdot - \theta_i)$ and the empirical distribution $\mathbb F_n$ based on the observations $X_i, i =1, \ldots, n$ among all vectors $(\theta_1, \ldots, \theta_n)$.  As \cite{edelman} already noted, the normal distribution can be replaced by other distributions, which is exactly what we do here. The merits of this approach are the facts that it does not depend on some bandwidth, and that it  is easily implementable.  The link between our problem and the work \cite{edelman} is made clear in Section \ref{sec: estim}. There,  we introduce our estimator, which we call the minimum contrast estimator, and we establish its existence and some of its important features in the case of equal sample sizes for the responses and covariates.  In Section \ref{sec: rate},  we establish rates of convergence under some fairly general conditions on the distribution function  of the covariates. In that section, we assume that the noise distribution is known and distinguish between three cases for its smoothness: (1) ordinary smooth, (2) supersmooth and (3) discrete with finite number of support points.  The convergence rates in cases (1) and (2) are derived using some classical Fourier transform techniques, while in case (3) a very different approach is employed, which makes use of a recent result in \cite{Meis}. In the proofs, we use a conversion device which allows us to link the convergence rate of the estimator to that of its generalized inverse, an interesting result in its own right. Since our method allows for different sample sizes for the responses and covariates, we extend the construction of the estimator to this case and derive the corresponding convergence rate in the three aforementioned smoothness cases. In Section \ref{sec:morerates}, we show that the estimator achieves the parametric rate in estimating the moments of $m_0(X)$ in case (3) and also when the noise is known to be uniformly distributed on a compact.    Although our estimator cannot be shown to be unique,  we prove in Section \ref{sec: fenchel} that any solution has to satisfy a necessary optimizing condition. This condition is used to develop a gradient-descent algorithm to compute the estimator, see Section \ref{sec: comp}.   In Section \ref{sec: ext}, we discuss how our method can be  extended to 
the more realistic situation where the noise distribution is unknown.  In Section \ref{sec:demonstr-synth-real}, we illustrate our approach through synthetic and real data. We finish this manuscript with some concluding remarks and future research directions; see Section \ref{sec: concl}. Technical proofs are deferred to an Appendix.



\section{The Minimum Contrast Estimator: Existence}\label{sec: estim}

\subsection{Setup, Terminology, and Notation}
\label{sec:setup-term-notat}

Let $m_0$ be the monotone function appearing in the model in (\ref{eq:2}), and in which we are interested. In the model that we consider, the response $Y$ has the same distribution as the convolution of $m_0(X)$ and the noise $\epsilon$ with $m_0$ monotone non-decreasing. Note for the case where  $m_0$ is non-increasing it is enough to consider $-Y$ instead of $Y$ and all our results will still apply.
Denote by $\mc M$ the set of all bounded non-decreasing and right continuous functions defined on $[0,1]$. This class accomodates for the assumption made in the sequel that the covariate $X\in[0,1]$ almost surely.

It is worth mentioning that without the monotonicity assumption, the function $m_0$ is not identifiable in general, even in the simple case where $m_0(X)$ is observed without noise, that is in the case where $Y=m_0(X)$ is observed. Precisely, the distribution of $X$  together with the distribution of $m_0(X)$ does not pin down the function $m_0$. A simple counterexample (that was suggested by a referee)
is where $ X$ has the uniform distribution on $[0,1]$:
\begin{eqnarray*}
m_1(x) = x \ \textrm{and} \ \ m_2(x) =  1- 2 \big \vert x - \frac{1}{2} \big \vert. 
\end{eqnarray*}
Then,  it is easy to see that both $m_1(X)$ and $m_2(X)$ have the uniform distribution on
$[0,1]$. On the other hand,  if $m_0$ is monotone, then it can be determined on the support of $X$ using knowledge about the distribution of $X$ and $m_0(X)$. This  identifiability property is proved in the following proposition.

\begin{prop}\label{prop: identifiability} Let $X$ be a real valued random variable with a continuous distribution on $[0,1]$, and let $m_1$ and $m_2$ be non-decreasing and right-continuous functions on the support of $X$.  If $m_1(X)$ has the same distribution as $m_2(X)$, then $m_1=m_2$ on the support of $X$. 
\end{prop}

Here, we describe the working framework of our estimation approach. We observe two independent samples $(X_1, \ldots, X_n)$ and $(Y_1, \ldots, Y_n)$ such that the following holds.
  \begin{assumption}{A0}
    \label{assm:A0:basic:m0-Xi}
    $X_1,\ldots, X_n \iid F_0$, $0  \le X_i \le 1$ almost surely, and $Y_1, \ldots, Y_n \iid H_0$. 
    The unobserved error $\epsilon$ is independent of $X$, satisfies $E(\epsilon)=0$, and has CDF $\Phi_\epsilon$. Moreover, $m_0\in\mathcal M$.
  \end{assumption}

Note that there are no restrictions on the relationships between the $X_i$'s and the $Y_j$'s, other than the equality in distribution in (\ref{eq:2}). Then,  $F_0$, $H_0$, and $\Phi_\epsilon$ are the true CDF's of $X_i$, $Y_i$, and $\epsilon_i$, respectively.  Let $\FF_n(x) = n^{-1} \sum_{i=1}^n \one_{[X_i,\infty)}(x)$ and 
  $\HH_n(x) =  n^{-1} \sum_{i=1}^n \one_{[Y_i, \infty)}(x)$ where $x \mapsto \one_A(x)$ is the indicator function for the set $A$. Also, let $\| \cdot \|_\infty$ be the supremum norm, i.e.\ $\| m \|_\infty = \sup_{t \in [0,1]} |m(t)|$.   For $K > 0$, let $\mathcal{M}_K$ be the set of functions $m\in{\cal M}$ such that
  $\Vert m \Vert_\infty \le K$.  For $m \in \mc M$, $m^{-1}$ denotes the {\it generalized inverse} of $m$, i.e.\ 
\begin{equation}\label{eq: inverse}
m^{-1}(y) := \inf \lb x \in[0,1] : m(x) \ge y \rb
\end{equation}
  where the infimum of an empty set is defined to be $1$. Hence, we have $m^{-1}(y)=1$ for all $y> m(1)$.

  In this section, we assume that $\Phi_\epsilon$ is known.  This assumption will be relaxed
  in Section~\ref{sec: ext}.
  Also, although we take the respective sizes of the samples of covariates and responses to be equal, our method can be easily adapted to the case where these sizes are different. In that case, the convergence rate of our estimator is driven by the minimum of the sample sizes;
  see Subsection~\ref{sec:conv-rate-case-diff}.



 Let $\mathcal{F}$ be the set of all distribution functions on $\RR$. A contrast function $\mathcal{C}$ defined  on the Cartesian product $\mathcal{F} \times \mathcal{F}$ is any non-negative function such that $\mathcal{C}(F_1, F_2) = 0$ if and only if $F_1  = F_2$.  Consider the estimator
\begin{eqnarray}\label{Minim}
\widehat m_n= \textrm{argmin}_{m \in \mathcal{M}}   \ \mathcal{C}\Big(\mathbb{H}_n, n^{-1}  \sum_{i=1}^n \Phi_\epsilon(\cdot  -  m(X_i)) \Big)
\end{eqnarray}
provided that a minimizer exists (in which case it is not necessarily unique).   Since the criterion depends on $m$ only through its values at the observations $X_i, i=1, \ldots, n $,  it follows that  any candidate for the minimization problem in (\ref{Minim}), $m$ in $\mathcal{M}$, can be well replaced by the non-decreasing function $\widetilde{m}$ such that
\begin{eqnarray*}
\widetilde{m}(t) = \left \{
\begin{array}{ll}
 m(X_{(1)}), \  \  \textrm{for $t \in [0, X_{(1)}]$}
\\
m(X_{(n)}) , \ \   \textrm{for $t \in [X_{(n)},1]$}
\end{array}
\right.
\end{eqnarray*}
and $\widetilde{m}$ is constant between the remaining order statistics such that $\tilde{m} $  is right continuous (by definition of $\mathcal{M}$) and coincides with $m$ at every $X_{(i)}$, $i=1,\dots,n$.   Here as is customary, $X_{(1)}\leq\dots \leq X_{(n)}$ denote the order statistics corresponding to $X_1,\dots,X_n$.


  In addition,  note that $\widehat m_n$  does not have to be unique at the data points and thus $\wh m_n$ denotes any solution of the minimization problem.

\subsection{Existence}\label{sec: exists}
In the sequel, we consider the following contrast function
\begin{eqnarray*}
\mathcal{C}(F_1, F_2)  =  \int_{\mathbb{R}}  \Big \{F_1(y)  - F_2(y)\Big \}^2 dy
\end{eqnarray*}
whenever this integral is finite, which is the case when $F_1$ and $F_2$ are distribution functions of random variables admitting finite expectations.  The choice of such contrast function is mainly motivated by application of the Parseval-Plancherel's Theorem.  The estimator we consider here is reminiscent of the one studied in \cite{edelman} for deconvoluting a distribution function from a Gaussian noise.  However, our goal here is  different since the main target in our problem is the monotone transformation $m_0$.

Before starting the analysis of the estimator, we establish first its existence.  Denote
\begin{eqnarray*}
\mathbb{M}_n(m)  = \int_{\mathbb{R}}  \Big \{ \mathbb{H}_n(y)  -  n^{-1}  \sum_{i=1}^n \Phi_\epsilon(y  -  m(X_i)) \Big \}^2   dy
\end{eqnarray*}
and  let
\begin{eqnarray*}
\mathbb{M}(m)  = \int_{\mathbb{R}}  \left \{ H_0(y)  -  \int_{\R}  \Phi_\epsilon(y  -  m(x)) dF_0(x)  \right \}^2  dy
\end{eqnarray*}
be its deterministic counterpart. Then the minimizer $\widehat m_n$ in \eqref{Minim} (if it exists) can also be written as
\begin{equation}
  \label{Minim2}
  \widehat m_n
  \in
  \argmin_{m \in \mathcal{M}}  \mathbb{M}_n(m). 
\end{equation}
The following assumptions will be needed.

\begin{assumption}{A1}
  \label{assm:A1:m0-sup}
  For some $K_0 \in [0, \infty)$, we have $ \Vert m_0 \Vert_\infty = K_0$.
\end{assumption}

\begin{assumption}{A2}
  \label{assm:A2:Phi-cts}
  The distribution function $\Phi_\epsilon$ is continuous on $\RR$. 
\end{assumption}


\begin{prop}\label{prop: exist}
  Let Assumption~\ref{assm:A0:basic:m0-Xi} hold. Then, 
  \begin{enumerate}
  \item $\mathbb{M}_n(m)$ is finite for any  $m \in \mathcal{M}$  for all $n \ge 1$;
  \item If  Assumption~\ref{assm:A1:m0-sup} also holds then $\mathbb{M}(m)$ is finite
    for any $m \in \mathcal{M}$;
  \item \label{item:1} If  Assumptions~\ref{assm:A1:m0-sup} and \ref{assm:A2:Phi-cts} also hold, then there exists at least a solution to \eqref{Minim2} that is piecewise constant and right-continuous, for all $n \ge 1$.

  \item \label{item:12}If Assumptions~\ref{assm:A1:m0-sup} and \ref{assm:A2:Phi-cts}  also hold, and $\epsilon$ has a
    bounded support, then with probability 1, there exists at least a solution to \eqref{Minim2} that is piecewise constant, right-continuous and bounded in the sup-norm by a deterministic constant for all $n \ge 1$.  This deterministic constant can be taken to be equal to $K_0+2$ in case  $\Vert \epsilon \Vert_\infty \le 1$ with probability $1$. 

  \end{enumerate}
\end{prop}

Note that parts~\ref{item:1} and ~\ref{item:12} of the proposition give existence but not necessarily uniqueness of $\wh m_n$. In fact, it can be seen from the proof of Proposition \ref{prop: exist} that if $\wh m_n$ is a solution to \eqref{Minim2}, then any monotone function that coincides with $\wh m_n$ at  the observed covariates $X_1,\dots,X_n$ gives another solution to \eqref{Minim2}.
In what follows, we will consider a solution $\widehat m_n$ that takes constant values between successive covariates and that is right continuous. This choice is
consistent with the way the Grenander-type estimator is defined, that is, the estimator in the classical monotone regression estimation problem.

\section{Convergence Rates of the Minimum Contrast Estimator}\label{sec: rate}

In this section, we will give upper bounds on the convergence rate of the minimum contrast estimator defined above.  Not surprisingly, this rate depends on the smoothness of the noise distribution. More specifically, we will consider the following cases for the smoothness: (1)  ordinary smooth, (2) supersmooth  and (3) discrete with a finite support.    In the whole section, we assume that the Assumptions~\ref{assm:A0:basic:m0-Xi},~\ref{assm:A1:m0-sup} and~\ref{assm:A2:Phi-cts} hold.  By Proposition \ref{prop: exist}, this guarantees that the minimization problem in  (\ref{Minim2}) admits a piecewise constant and right-continuous solution.  We will denote $\wh m_n$ any such solution.  Also, we will require the following assumption about the distribution of the design points.
\begin{assumption}{A3}
  \label{assm:A3:f0-bounded}
  The common distribution function $F_0$ of the covariates $X_1, \ldots, X_n$
is continuous. 
\end{assumption}

 For the smoothness cases (1) and (2), we will need the following notation:
\begin{eqnarray}\label{PsiF}
\psi_F(x)  =    \int_{\R}  e^{it x}  dF(t)
\end{eqnarray}
for any distribution function $F$ on $\RR$, and 
\begin{eqnarray}\label{phig}
\phi_g(x) =  \int_{\R}  e^{it x} g(t) dt
\end{eqnarray}
denotes the Fourier transform of $g$, whenever $g$ is integrable.   Note that when $F$ is absolutely continuous with density $f$, then it follows immediately that $\psi_F =  \phi_f$.  . 

\subsection{Convergence Rate Under Ordinary Smooth Noise}

We start with the ordinary smooth case for the noise distribution described in the next assumption.

\begin{assumption}{A4}
  \label{assm:A4:Phi-smoothness}
  The distribution function $\Phi_\epsilon$ is absolutely continuous with a $0$-mean ordinary smooth density $f_\epsilon$ in the sense that 
  \begin{eqnarray}\label{OrdSmooth}
    \frac{d_0}{\vert t \vert^{\beta}}  \le \vert \phi_{f_\epsilon}(t) \vert   \le \frac{d_1}{\vert t \vert^{\beta}} 
  \end{eqnarray}
  as $ \vert t \vert \to \infty$,
  for some $\beta > 0$ and constants $d_0 > 0, d_1 > 0$. 
\end{assumption}

Some comments are in order. Assumption \ref{assm:A4:Phi-smoothness}  is common in deconvolution problems, see for instance \cite{dattner2011deconvolution}. 
The positive real number $\beta$  is usually referred to as the order of smoothness.   Known examples include the Exponential distribution with any scale parameter for which $\beta=1$,  Gamma distribution with shape parameter $\alpha > 0$ and scale $\gamma > 0$  for which $\beta=\alpha$, the Laplace distribution with $\beta =2$,  and more generally symmetric Gamma distributions (the distribution of $X - X'$ where $X$ and $X'$ are i.i.d. $\sim$ Gamma$(\alpha, \gamma)$) in which case we have $\beta = 2 \alpha$.  
See for example the examples given in \cite{Fan91} after (1.4).
We plot in Figure~\ref{fig:gamma-density} in the appendix several gamma densities with a variety of shape parameters, to show the ``decreasing smoothness:'' as the shape parameter goes to $0$, the density has an increasing spike at $0$.


Our main theorem below provides the rate of convergence of the $L_1$-error with respect to the distribution of the design points, $F_0$,  over a given interval $[a,b]\subset  [0,1]$  provided that the estimator is stochastically bounded on that interval. Sufficient conditions for this boundedness are given below.

The following assumption is crucial for deriving the convergence rate of $\wh  m_n$. It will be also needed below in the supersmooth case. 

\begin{assumption}{A5}
  \label{assm:A5}
Assume that there exists $T^*$ such that $\vert \phi_{f_\epsilon}(t) \vert   \ge\vert \phi_{f_\epsilon}(T) \vert>0$ for all $T>T^*$ and $\vert t\vert \le T$.
\end{assumption}

Now, we we are able to state the main result of this subsection.

\medskip

\begin{theo}
\label{theo: rate} \medskip

\begin{enumerate}

\item Suppose that Assumptions \ref{assm:A0:basic:m0-Xi} to 
\ref{assm:A5}  hold.   Let $[a,b]\subset[0,1]$ be a fixed interval such that
\begin{equation}\label{eq: AnBn}
\widehat m_n(a)=O_P(1)\mbox{ and }\widehat m_n(b)=O_P(1).
\end{equation}
Then, it holds that
\begin{eqnarray}\label{eq:rateab}
\int_a^b \left \vert \widehat m_n(x)  -  m_0(x) \right \vert dF_0(x) = O_P(n^{-1/(2(2\beta +1))} ).
\end{eqnarray}

\item If the Assumptions \ref{assm:A0:basic:m0-Xi} to 
\ref{assm:A5}  hold, then the claims in \eqref{eq: AnBn} hold for all $a$ and $b$ such that $0<F_0(a)\leq F_0(b)<1$. 

\end{enumerate}
\end{theo}

\bigskip

We give below the main steps of the proof and conclude this subsection with some comments about the rates in Theorem \ref{theo: rate}. Details of the proof are postponed to Section \ref{sec: proofs}. Let us define  
\begin{eqnarray}\label{DefLnhatLn}
  \widehat{\mathbb{L}}_n(w) : =  \frac{1}{n} \sum_{i=1}^n
   \one_{ \{\widehat m_n(X_i)  \le w \}}
  \ \textrm{and} \ \
  \mathbb L_n(w)  :=   \frac{1}{n} \sum_{i=1}^n
  \one_{\{m_0(X_i)  \le w\}}  
\end{eqnarray}
for all $w \in \R$.  Using deconvolution arguments we show closeness of the two processes in the following proposition.

\medskip

\begin{prop}\label{prop: DistanceG}
Assume that Assumptions \ref{assm:A4:Phi-smoothness} and \ref{assm:A5} hold.   Then,
\begin{eqnarray*}
\ E \left[\int_{\R}  \left(\widehat{\mathbb{L}}_n(w)   -  \mathbb L_n(w)  \right)^2  dw \right]  = O(n^{-1/(2\beta +1)} ).
\end{eqnarray*}
\end{prop}

Next, using entropy arguments from empirical process theory, we show in the following proposition that on intervals $[A,B]$  that are possibly random, the empirical processes $\widehat{\mathbb{L}}_n$ and $\mathbb L_n$ are close to their population counterparts $\widehat L^0_n$ and $L_0$ respectively, where  for all $w \in \RR$,
\begin{eqnarray}\label{eq: Ln}
\widehat L^0_n(w)  = \int \one_{\{\widehat m_n(x) \le w\}}dF_0(x) 
\mbox{ and } L_0(w)  =  \int\one_{\{m_0(x)\leq w\}} dF_0(x).
\end{eqnarray}  
More precisely, we  derive the convergence rate of the associated $L_2$-error integrated on such interval $[A, B]$. Note that the first two claims in the proposition hold under the only assumptions that $X_1,\dots,X_n$ are i.i.d. and $\widehat m_n$ is taken from
\eqref{Minim2}.
In fact, it can be seen from the proof that these two claims continue to hold with $\widehat m_n$ replaced by any monotone estimator.

\medskip

\begin{prop}\label{prop: E1andE3andE2}
If $X_1,\dots,X_n$ are i.i.d., then for all random variables $A < B $ (that may depend on $n$) it holds that 
\begin{eqnarray*}
\int_{A}^B  \left(\widehat{\mathbb{L}}_n(w)    -\widehat{L}^0_n(w)   \right)^2  dw   \leq (B-A)  O_P( 1/n),
\end{eqnarray*}
\begin{eqnarray*}
 \int_{A}^B  \Big(\mathbb{L}_n(w) - L_0(w)\Big)^2  dw  \leq (B-A) O_P(1/ n),
\end{eqnarray*}
where $O_P( 1/n)$ is uniform in $A$ and $B$. Moreover, if $B-A=O_P(n^{2\beta/(2\beta+1)})$ and Assumptions \ref{assm:A4:Phi-smoothness} and  \ref{assm:A5} holds, then
\begin{eqnarray}\label{E2}
\int_{A}^B  \left(\widehat L^0_n(w)  -  L_0(w)  \right)^2  dw=  O_P(n^{-1/(2\beta+1)}).
\end{eqnarray}
\end{prop}

\medskip

The following proposition makes the connection between the above error and a squared distance between the inverse functions of $\widehat m_n$ and $m_0$ composed with the distribution function $F_0$ of the $X_i$'s;  a rate of convergence of that squared distance is derived. We recall that for $m \in \mc M$, the inverse of $m$ is defined by \eqref{eq: inverse},   where the infimum of an empty set is defined to be $1$. 

\medskip

\begin{prop}\label{prop: DistanceG2}
 Under Assumptions \ref{assm:A3:f0-bounded}, \ref{assm:A4:Phi-smoothness}, and \ref{assm:A5}
for all random variables $A < B $ such that $B-A=O_P(1)$ it holds that
\begin{eqnarray*}
\int_{A}^B  \left(  F_0 \circ \widehat m^{-1}_n(w)  -   F_0 \circ m^{-1}_0(w)   \right)^2  dw & =&  
\int_{A}^B   \left(\widehat{L}^0_n(w)  -   L_0(w) \right)^2  dw\\
&= &  O_P(n^{-1/(2\beta +1)} ).
\end{eqnarray*}

\end{prop}

The last step in the proof makes the connection between the above squared distance and the $L_1$-error of $\widehat m_n$. 

\medskip

\begin{prop}\label{prop: Conversion}
Suppose that Assumption  \ref{assm:A3:f0-bounded} holds, and let $[a,b]\subset[0,1]$ be a fixed interval. Then  it holds that
\begin{eqnarray*}
\int_a^b \left \vert \widehat m_n(x)  -  m_0(x) \right \vert dF_0(x)  \le 
\left((B_n-A_n)\int_{A_n}^{B_n}( F_0\circ \widehat m_n^{-1}(x)-F_0\circ m_0^{-1}(x)) ^2dx\right)^{1/2}
\end{eqnarray*}
where 
$A_n=m_0(a)\wedge \widehat m_n(a)\mbox{ and }B_n=m_0(b)\vee \widehat m_n(b).$ 
\end{prop}

We conclude the subsection with some comments about the convergence rate obtained in Theorem \ref{theo: rate}.

\begin{itemize}
\item When $\beta = 1/4$, then the $L_1$-rate of convergence obtained for our estimator matches the well-known $n^{1/3}$-rate of the Grenander estimator in the classical isotonic regression (where the link between the responses and covariates is known).  If  $\beta < 1/4$, then the rate is strictly better than the cubic rate,  which may appear as a contradiction. However,  there is  one major difference between the regular isotonic regression model and the one we consider here: we assume the error distribution to be known, which is not the case in regular isotonic regression. Thus,  it seems that a transitional regime occurs in the rate of convergence in case the noise distribution is known. In the unlinked regression setting which we study in this paper, if the error distribution is known to  be a centered (or symmetric)  Gamma with shape and scale parameters $\alpha \in (0,1/4)$  (or $\alpha \in (0,1/8)$)  and $\lambda > 0$  then the smoothness parameter is $\beta = \alpha$ (or $\beta = 2\alpha$)  belongs to $(0,1/4)$ and the rate of convergence of our estimator will be faster than $n^{1/3}$.
 
\item   Note that if we let $\beta \to 0$, then the centered or symmetric Gamma will converge to a Dirac at $0$ and the rate of convergence will approach the parametric rate $\sqrt n$. In fact, it can be shown that the rate of convergence is precisely the parametric rate if the error distribution is a Dirac at $0$, i.e. if $\epsilon=0$ with probablity one. Indeed, when $\epsilon =0$ with probability $1$, we have $\Phi_\epsilon(y- \widehat m_n(X_i)) =  \mathds{1}_{ \widehat m_n(X_i) \le y}$.  If $\widehat m_n$ is taken to be a minimum contrast estimator that is bounded in the sup-norm by $K_0+2$ (such a solution is known to exist with probability $1$ by Proposition  \ref{prop: exist}), then we have  by definition of $\wh m_n$
\begin{eqnarray*}
\int_{\mathbb R} \big\{ \mathbb H_n(y)  -   n^{-1} \sum_{i=1}^n  \mathds{1}_{ \widehat m_n(X_i) \le y} \big \}^2 dy  \le \int_{\mathbb R} \big \{ \mathbb H_n(y)  -   n^{-1} \sum_{i=1}^n  \mathds{1}_{m_0(X_i) \le y} \big \}^2 dy
\end{eqnarray*}
implying that 
\begin{eqnarray*}
&& \int_{\mathbb R} \bigg\{ n^{-1} \sum_{i=1}^n  \mathds{1}_{ \widehat m_n(X_i) \le y}  -   n^{-1} \sum_{i=1}^n  \mathds{1}_{ m_0(X_i) \le y} \bigg \}^2 dy  \nonumber \\
&& \le  4  \int_{\mathbb R} \big \{ \mathbb H_n(y)  -   n^{-1} \sum_{i=1}^n  \mathds{1}_{ m_0(X_i) \le y} \big \}^2 dy  \nonumber \\
&& \le  8 \int_{\mathbb R} \big \{ \mathbb H_n(y)  -H_0(y) \big\}^2 dy  +  8   \int_{\mathbb R} \big \{  H_0(y)  -  n^{-1} \sum_{i=1}^n  \mathds{1}_{ m_0(X_i) \le y} \big \}^2 dy.  \nonumber
\end{eqnarray*}
The  first term was already shown to be $O_P(n^{-1})$; see the proof of Proposition \ref{prop: DistanceG}. Moreover, the true distribution $H_0$ of the $Y_i$'s is that of $m_0(X)$, which implies that the second term is also $O_P(n^{-1})$ since its expectation is equal to $n^{-1}  \int Var(\mathds{1}_{m_0(X) \le y }) dy = n^{-1}  \int H_0(y) (1-H_0(y)) dy $ where the integral $\int H_0(y) (1- H_0(y)) dy$ is known to be finite; see Appendix A.  Hence,
\begin{eqnarray}\label{compdirac}
 \int_{\mathbb R} \bigg\{ n^{-1} \sum_{i=1}^n  \mathds{1}_{ \widehat m_n(X_i) \le y}  -   n^{-1} \sum_{i=1}^n  \mathds{1}_{ m_0(X_i) \le y} \bigg \}^2 dy 
 =   O_P(n^{-1}).
\end{eqnarray}
If follows from the rate obtained in (\ref{compdirac}) and the empirical process arguments as in the proof of Proposition \ref{prop: E1andE3andE2}  that 
\begin{eqnarray*}
\int_{-K_0-2}^{K_0+2} \bigg (\widehat L^0_n(w)  -  L_0(w) \bigg )^2 dw  & = &   \int_{-K_0-2}^{K_0+2} \bigg (F_0 \circ \wh m^{-1}_n(w)  -  F_0 \circ m^{-1}_0(w) \bigg )^2 dw   \\
& = & O_P(n^{-1}).
\end{eqnarray*}
This in turn implies by Proposition \ref{prop: Conversion} (with $a=0$, $b=1$,  and $[A_n, B_n] \subseteq [-K_0-2, K_0+2]$)   that 
\begin{eqnarray*}
\int_0^1 \vert \widehat m_n(x) - m_0(x)  \vert dF_0(x)  =  O_P(n^{-1/2}).  
\end{eqnarray*}

\item The parametric rate obtained for a noise that is Dirac at $0$ can be generalized to the case where $\epsilon$ is is supported on a finite set of points.  This is proved in  Theorem \ref{theo: nonoise}.

\item It may be considered as unsatisfactory that the conclusions were stated with $O_P$ notation, so we provide now more precise bounds. The bounds can be obtained by closely inspecting the proof of Theorem \ref{theo: rate}; details are omited. It can be seen from the proofs that, with $A_n$ and $B_n$ the stochastically bounded random variables taken from   Proposition \ref{prop: Conversion}, we have
\begin{eqnarray}\label{eq:biasvar} \notag
&&\int_a^b \left \vert \widehat m_n(x)  -  m_0(x) \right \vert dF_0(x) \\ \notag
&&\qquad \le 
\left((B_n-A_n)\left[3\int_{A_n}^{B_n}   \left(\widehat{L}_n(w)  -   L_n(w) \right)^2  dw+ \frac{6(B_n-A_n)}{n}    \Vert \mathbb{G}_n \Vert_{\mathcal{I}}^2\right]\right)^{1/2}\\ \notag
&&\qquad \le 
\left(3(B_n-A_n)\int_{\RR}   \left(\widehat{L}_n(w)  -   L_n(w) \right)^2  dw\right)^{1/2}\\
&&\qquad\qquad+(B_n-A_n)\left( \frac{6\Vert \mathbb{G}_n \Vert_{\mathcal{I}}^2}{n}\right)^{1/2}
\end{eqnarray}
where $   \Vert \mathbb{G}_n \Vert_{\mathcal{I}}^2$ is a random variable with finite expectation. The expectation is bounded above by an unknown absolute constant that is connected to the entropy measure of the set of monotone functions on $\RR\to[0,1]$, as well as  to the absolute constants that emerge from the empirical process theory. Hence, the second summand in the right-hand term is of the parametric order $n^{-1/2}$, and can be seen as a systematic error that is typically negligable as compared to the first summand. This means that the rate of convergence of the estimator is driven by  the integral in the first summand. For this integral, we have
$$E\left[\int_{\RR}   \left(\widehat{L}_n(w)  -   L_n(w) \right)^2  dw\right]\leq Kn^{-1/(2\beta+1)}, $$
where $K$ depends on the parameters of the model. One can choose for instance
$$K:= (2\beta+1)\left(\frac{24}{d_0}(2/(\beta\pi))^{2\beta}(2K_0+E|\epsilon|)\right)^{1/(2\beta+1)}.$$
It is worth mentioning that the decomposition in \eqref{eq:biasvar} still holds with $\widehat m_n$ replaced (in the definitions of $\widehat L_n$, $A_n$ and $B_n$) by any estimator in $\mathcal M$.
\end{itemize}

\subsection{Convergence Rate Under Supersmooth Noise}
The main arguments used above for an ordinary smooth noise continue to apply to the supersmooth case. In this setting,  Assumption \ref{assm:A4:Phi-smoothness} should be replaced by the following one.
\begin{assumption}{A4'}
  \label{assm:A4'}
  The distribution function $\Phi_\epsilon$ is absolutely continuous with a $0$-mean supersmooth density $f_\epsilon$ in the sense
  that
\begin{eqnarray*}
d_0 \vert t \vert^{\alpha} \exp(- \vert t\vert^{\beta}/\gamma)  \le \vert  \phi_{f_\epsilon}(t)  \vert \le d_1 \vert t \vert^{\alpha} \exp(- \vert t\vert^{\beta}/\gamma)
\end{eqnarray*}
as $\vert t \vert \to \infty$, for some $\alpha > 0, \beta  > 0$ and constants $d_0 > 0, d_1 > 0$.  
\end{assumption}

\noindent
In the above definition we provide for supersmoothness we deviate from the one given by \cite{Fan91} in (1.3) by taking the same exponent $\alpha$ in the lower and upper bound, for simplicity. 
\begin{theo}
  \label{rem:1}
  Under the  Assumptions  \ref{assm:A0:basic:m0-Xi}--\ref{assm:A3:f0-bounded}, \ref{assm:A4'}, and \ref{assm:A5}, we  have for any $[a,b] \subset (0,1)$ such that $0 < F_0(a) \le F_0(b) < 1$ that
  \begin{equation*}
    \int_a^b \left \vert \widehat m_n(x)  -  m_0(x) \right \vert dF_0(x)
    = O_P( (\log n)^{-1/(2\beta)} ).
  \end{equation*}
\end{theo}

\bigskip
\noindent 
For instance,
in the case of a Gaussian noise, in which case $\beta=2$, the rate is $1/(\log n)^{1/4}$, which matches the conclusion of \cite{edelman} right after the proof of his Theorem 1. On the other hand, this rate is slower than the minimax rate, $\log \log n/\log n$, obtained by \cite{Rigollet:2018wia} 
in the shuffled regression problem. Note however that  the setting studied by \cite{Rigollet:2018wia} is the shuffled regression problem, whereas we consider the unlinked regression problem (recall we explain the difference in  our introduction).  It is not known whether the two problems share the same minimax rates; or, rather, it is not known if the minimax lower bound  derived in \cite{Rigollet:2018wia} applies to the unlinked problem, since the unlinked problem is statistically more difficult than the shuffled problem.  Therefore, it is currently unknown  if the rate we derived is in fact minimax suboptimal or not.  

\subsection{Convergence Rate Under a Discrete Noise Distribution}

In this section, we consider the case where  $\epsilon$ is supported on a finite set of points. We prove that in that case, the minimum contrast estimator achieves the parametric rate in the $L_1$-loss.

\begin{theo}\label{theo: nonoise}
Suppose that  Assumptions \ref{assm:A0:basic:m0-Xi} to 
\ref{assm:A2:Phi-cts}  hold, and that $\epsilon$ is supported on a finite set of points.  If $\widehat{m}_n$ is a solution to \eqref{Minim2} that is bounded in the sup-norm by a deterministic constant (which exists with probability $1$ in view of Proposition \ref{prop: exist}), then 
\begin{eqnarray*}
    \int_0^1 \left \vert \widehat m_n(x)  -  m_0(x) \right \vert dF_0(x)
    = O_P( n^{-1/2} ).
\end{eqnarray*}
\end{theo}

\subsection{Convergence Rate in the Case of Different Sample Sizes}
\label{sec:conv-rate-case-diff}

In this section we briefly consider the case where one observes $Y_1, \ldots, Y_{\ny}$ and  $X_1, \ldots, X_{\nx}$
with possibly different sample sizes $\nx \neq \ny$. In that case, the estimator is defined as
\begin{equation*}
  \label{eq:9}
  {\widehat m_{n_x, n_y}}= \textrm{argmin}_{m \in \mathcal{M}} \int_{\mathbb{R}}  \Big \{ \mathbb{H}_{\ny}(y)  -  \nx^{-1}  \sum_{i=1}^{\nx} \Phi_\epsilon(y  -  m(X_i)) \Big \}^2   dy   
\end{equation*}
where $\mathbb{H}_{\ny}$ denotes the empirical distribution function corresponding to the sample $Y_1, \ldots, Y_{\ny}$. The asymptotic here has to be understood in the sense that both sample sizes $\nx$ and $\ny$ go to infinity. This means that $\nx\wedge\ny\to\infty$, where $\nx\wedge\ny$ denotes the infimum between $\nx$ and $\ny$.   In the following theorem, we give the upper bound on the convergence rate of our minimum contrast estimator in the three regimes (1) ordinary smooth, (2) supersmooth and (3) discrete with finite number of support points.

\begin{theo}\label{theo: rateunequalsamplesizes}
\begin{enumerate}
\item
Suppose that the Assumptions~\ref{assm:A0:basic:m0-Xi},~\ref{assm:A1:m0-sup},~\ref{assm:A2:Phi-cts}, \ref{assm:A3:f0-bounded},    and \ref{assm:A5}  hold true.   Let $[a,b]\subset (0,1)$ be any fixed interval such that $0<F_0(a)\leq F_0(b)<1$.
\begin{itemize}
\item If  Assumption \ref{assm:A4:Phi-smoothness} holds true, then  as $n_x, n_y \to \infty$, we have
\begin{equation}\label{eq: nxy}
\widehat m_{n_x, n_y}(a)=O_P(1)\mbox{ and }\widehat m_{n_x, n_y}(b)=O_P(1)
\end{equation}
 and
\begin{eqnarray*}
\int_a^b \left \vert \widehat m_{n_x, n_y}(x)  -  m_0(x) \right \vert dF_0(x) = O_P((\nx\wedge\ny)^{-1/(2(2\beta +1))} ).
\end{eqnarray*}
\item If Assumption \ref{assm:A4'}  holds true, then as $n_x, n_y \to \infty$, \eqref{eq: nxy} holds, and
\begin{eqnarray*}
\int_a^b \left \vert \widehat m_{n_x, n_y}(x)  -  m_0(x) \right \vert dF_0(x) = O_P( \log (\nx\wedge\ny)^{-1/(2\beta)} ).
\end{eqnarray*}
\end{itemize}
\item Suppose that the Assumptions~\ref{assm:A0:basic:m0-Xi},~\ref{assm:A1:m0-sup},~\ref{assm:A2:Phi-cts} hold true. If $\epsilon $ is supported on a finite number of points, then with probability $1$ there exists a solution $\widehat m_{n_x, n_y}$ which is  bounded  by a deterministic constant. Furthermore,  
\begin{eqnarray*}
\int_0^ 1 \left \vert \widehat m_{n_x, n_y}(x)  -  m_0(x) \right \vert dF_0(x) = O_P((\nx\wedge\ny)^{-1/2} )
\end{eqnarray*}
as $n_x,  n_y \to \infty$. 
\end{enumerate}

\end{theo}

\subsection{Uniform Consistency}
The convergence rates for the minimum contrast estimator we derived above are obtained for the $L_1$-norm. Thus, a natural question is whether these global rates also hold pointwise. Although the exact answer to this question is still unknown  we can provide at least an intermediate result which shows that the estimator is pointwise consistent, even uniformly provided that the true monotone regression function, $m_0$, is continuous.   For simplicity, we assume that $n_x = n_y = n$. 

\medskip

\begin{theo}\label{theo: unifcons}
  Suppose that the assumptions of Theorem \ref{theo: rate}, or \ref{rem:1} or \ref{theo: nonoise} hold, and let $\mathcal{S}_0$ denote the support of $F_0$.  Let $a,b \in [0,1]$ be any points such that $\widehat m_n(a)$ and $\wh m_n(b)$ are each $O_P(1)$.  Let $C$ be any compact set in the interior of $[a,b] \cap \mc S_0$.  Assume $m_0$ is continuous on $(0,1)$.
    Then it holds that
\begin{eqnarray*}
\sup_{x \in C}  \vert \widehat m_n(x)  - m_0(x) \vert =  o_P(1).
\end{eqnarray*}

\end{theo}

\noindent Note that in the case where  $\epsilon$ is compactly supported, the result above implies that if $m_0$ is continuous on $[0, 1]$ 
\begin{eqnarray*}
\sup_{x \in [0, 1] \cap \mathcal{S}_0}  \vert \widehat m_n(x)  - m_0(x) \vert =  o_P(1)
\end{eqnarray*}
for any estimator $\widehat m_n$ that is bounded in the sup-norm by a deterministic constant.  Recall such an estimator exists with probability $1$  (see Claim 4 of  Proposition \ref{prop: exist}).   Finally, note that we cannot hope to extend the convergence outside $\mathcal{S}_0$ since $m_0$ is only identifiable on this set; see Proposition \ref{prop: identifiability}.

\section{Estimation of Moments of  $m_0(X)$} \label{sec:morerates}
In this section, we  showcase that our minimum contrast estimator can achieve the $\sqrt n$-rate for estimating certain smooth functionals of $m_0$.  In doing so, we restrict attention to estimating moments of $m_0(X)$ and the cases where either $\epsilon$ is discrete with a finite number of points in the support or when it is uniformly distributed over a compact. Modulo some scaling, we can assume without loss of generality that the support is a subset of $[-1,1]$ in both cases.  For simplicity, we assume that the sample sizes of the covariates and responses are equal.

Note that boundedness of $m_0$ implies that $m_0(X)$ admits finite moments of any order. Furthermore, in case the noise distribution is compactly supported, as assumed here in this section, all moments of  the response $Y$ are finite and we have
\begin{eqnarray*}
E(Y^k) = \sum_{j=0}^k \binom{k}{j} E[m^j_0(X)] E(\epsilon^{k-j}).
\end{eqnarray*}
Since the distribution of $\epsilon$ is known, the moments of the error can be exactly computed. Then, replacing the moments of $Y$ by the corresponding empirical estimators in the above formula yields a $\sqrt n$-consistent estimator of $E[m_0^k(X)].$  
The minimum contrast estimator offers an alternative (and also a direct) way for estimating the moments since one can simply take the natural choice $ \int  \widehat m^k_n d\mathbb F_n$, that is
$$\frac{1}{n} \sum_{i=1}^n \widehat m^k_n(X_i).$$
 Our result below shows that the latter estimator is converging at the parametric rate. 


\medskip

\begin{theo}\label{theo: parametricrate}
 Assume that Assumptions~\ref{assm:A0:basic:m0-Xi} to 
\ref{assm:A2:Phi-cts}  hold, and that either $\epsilon$ is supported on a finite set of points  or uniformly distributed on some compact.  Let  $\wh m_n$ denote a solution to \eqref{Minim2} which is piecewise constant, right-continuous and bounded in the sup-norm by a deterministic constant, for all $n \ge 1$ (which exists with probability $1$), and $\mathbb F_n$ the empirical CDF based on $X_1,..., X_n$. The following holds true.
\begin{itemize}
\item If $\epsilon$ is supported on a finite set of points,  then for all integers $k \ge 1$
\begin{eqnarray}\label{eq:moments}
\left \vert \int_0^1  \widehat m^k_n(x) d\mathbb F_n(x) - \int_0^1 m^k_0(x) dF_0(x)  \right \vert= O_P(n^{-1/2}), 
\end{eqnarray}
\item  If $\epsilon$ is uniformly distributed on a compact,  then 
\begin{eqnarray*}
\left \vert \int_0^1  \widehat m_n(x) d\mathbb F_n(x) - \int_0^1 m_0(x) dF_0(x)  \right \vert= O_P(n^{-1/2}).
\end{eqnarray*}

\end{itemize}

\end{theo} 

\medskip

Note that when $\epsilon$ is a uniformly distributed noise (the second case in the above theorem),  the convergence rate of $\wh m_n$ cannot be obtained from the results obtained in Section \ref{sec: rate}. Indeed, considering without loss of generality the case where the support is $[-1, 1]$, a uniform distribution does not belong to the ordinary smooth nor to supersmooth categories since in this case
\begin{eqnarray*}
\vert \phi_\epsilon(t)  \vert =  \left \vert \frac{\sin(t)}{ t} \right \vert,  \  \ \textrm{for $t \ne 0$}
\end{eqnarray*}
and hence   $\vert \phi_\epsilon(t)  \vert $ cannot be bounded below by $d_0/\vert t \vert^\beta$ for some $d_0, \beta > 0$  nor by  $d_0 \vert t \vert^\alpha \exp(- \vert t \vert^\beta/\gamma)$ for some $d_0, \alpha, \beta, \gamma > 0$.  Thus, the convergence rate of $\wh m_n$ for such a noise distribution is still an open problem. Nevertheless, Theorem  \ref{theo: parametricrate} shows clearly that our estimator behaves reasonably well when the goal is estimation of the first moment of $m_0(X)$.

\section{Fenchel Optimality Conditions}\label{sec: fenchel}

In view of the computational section below, we derive in this section the optimality conditions related to the optimization problem defining the estimator $\widehat{m}_n$.  Recall that by Assumption  A0, the covariates $X_1, \ldots, X_n$ are assumed to belong to $[0,1]$.  In the following,  we denote by $\widehat{m}_n$ a piecewise constant and right-continuous solution to \eqref{Minim2}, see Proposition \ref{prop: exist}. For some $1 \le p \le n$, we write
$ \mn{n}{1}< \ldots  < \mn{n}{p}$  for the distinct values taken by $\widehat{m}_n$  on $[0,1]$. 
We will use the following assumption on the density $\phie$. 

\medskip
\medskip

\begin{assumption}{A6}
  \label{assm:a6}
  The density $f_\epsilon$ is continuously differentiable such that 
  $$\sup_{t \in \RR} \vert f'_\epsilon(t) \vert \le D,$$
  for some constant $D > 0$, and
  $\int_{\RR} \vert f'_\epsilon(t) \vert dt < \infty$.
\end{assumption}

\medskip
\medskip

\begin{prop}\label{Fenchel}
  Let
  $\mn{n}{1} < \ldots < \mn{n}{p}$
  be the distinct values of the estimator $\widehat m_n$.
  Let Assumption~\ref{assm:a6} hold.
  Then,  for any  $k \in \{1, \ldots, p \}$  
\begin{eqnarray}\label{Cond1}
  \int_{\RR} \Big ( \mathbb{H}_n(y)  -  n^{-1}  \sum_{i=1}^n \Phi_\epsilon(y  -  \widehat m_n(X_i))  \Big )  f_\epsilon(y- 
  \mn{n}{k}
  )   dy = 0. 
\end{eqnarray}
Furthermore, this condition can be equivalently re-written as 
\begin{eqnarray}\label{AltCond2}
  1-\frac{1}{n} \sum_{i=1} ^n \Phi_\epsilon(Y_i-
  \mn{n}{k})   = \frac{1}{n} \sum_{i=1}^n \int_{\RR} \Phi_\epsilon(y- \widehat m_n(X_i)) f_\epsilon(y- \widehat m_k)   dy 
\end{eqnarray}
for $k = 1, \ldots, p$.
\end{prop}  






\begin{remark}
The alternative form in (\ref{AltCond2}) gives a useful way of verifying numerically the second equality condition via numerical integration. Indeed, with  $m=\wh m_n(X_i)$ and $m'=\hat m_k$, the integral on the right side of (\ref{AltCond2}) takes the form
\begin{eqnarray*}
   \int_{\RR} \Phi_\epsilon(y-m) f_\epsilon(y-m') \; dy 
&=&B(m'-m),
\end{eqnarray*}
where for all $m\in\RR$, $B(m)=E\Phi_\epsilon(\epsilon+m)=\int \Phi_\epsilon(y) f_\epsilon(y-m)dy.$ Explicit formulas of $B(m)$, for $m\in\RR$, can be even found for some distributions such as Laplace or Gaussian; see Subsections~\ref{sec:grad-comp-lapl} and \ref{sec: Gaussian}.
\end{remark}

\begin{remark}
  Consider the case where $\wh m_n$ takes one unique value, denoted
   $\mn{1}{1}$.
  Then
  the right side of   \eqref{AltCond2} equals $\int_{\RR} \Phi_{\epsilon}( y -
 \mn{1}{1}
  ) \phie(y-
   \mn{1}{1} %
  )
  dy = \int_{\RR} \Phi_\epsilon (y) \phie(y) dy$ which equals $E( \Phi_\epsilon(\epsilon)) = E(U) = 1/2$ where $U$ is Uniform$(0,1)$.
\end{remark}

\section{Computation}\label{sec: comp}
 Recall that our goal is to minimize $\MM_n$.  Write $\bs m := (m_1, \ldots, m_n)$, so that the objective function can be written as 
\begin{equation}\label{eq: Mm}
\MM_n(\bs m) := \int_{\RR} \lb \HH_n - n^{-1} \sum_{i=1}^n \Phi_\epsilon(\cdot - m_i) \rb^2
\end{equation}
 (by a slight abuse of notation).  To minimize $\MM_n$,  we can compute an unconstrained minimizer ${\bs{\tilde m}}$ of $\MM_n$ (i.e., we do not force $\tilde m_i \le \tilde m_{i+1}$ for $i=1, \ldots, n-1$), and then the overall solution would be given by reordering the entries of ${\bs {\tilde m}}$ so that it is nondecreasing; i.e., $\wh m_n(X_{(i)}) := \tilde m_{(i)}$ (where $\tilde m_{(1)} \le \cdots \le \tilde m_{(n)}$). The gradients $\frac{\partial}{\partial m_i} \MM_n(\bs m)$ can be computed using that
\begin{equation}\label{eq: gradient}
\frac{\partial}{\partial m_i} \MM_n(\bs m)=2 n^{-1}-2n^{-2}
\sum_{\alpha=1}^n\left\{ \Phi_\epsilon(Y_\alpha - m_i)  + B(m_i-m_\alpha)\right\}
\end{equation}
where for all $m\in\RR$, $B(m)=E\Phi_{\epsilon}(\epsilon+m)$, see  Appendix \ref{appendixgradient}.
In Appendix~\ref{sec:gradient-computation}
we show how to derive the gradient of $\MM_n$ when $\Phi_\epsilon$ is either a Gaussian or a Laplace distribution.  Thus, we can consider using a (first order) gradient descent algorithm for computation.

However, because the objective function is symmetric in all its $m_i$ components, at any stationary point at which some of the $m_i$'s are equal, no gradient-based method can tell if we could improve the objective function by allowing the equal $m_i$'s to take separate values.  Thus, we will consider a second order method to solve this problem: we will compute a second derivative in a direction related to separating $m_i$ into two distinct values.

Let $\MM_{n,p}$ denote the objective function parameterized such that it takes $2p$ arguments and the regression function $m$ is represented by  $p$ values
$\bs m := (m_1, \ldots, m_p)$ and $p$ weights $\bs \pi := (\pi_1, \ldots, \pi_p)$
($\sum_{\alpha=1}^p \pi_\alpha   =1 $), so that
 $ \MM_{n,p}(\bs \pi, \bs m) =   
  \int_{\RR}
  \big( \HH_n(y) -  \sum_{\alpha=1}^p \pi_\alpha \Phi_{\epsilon}(y- m_\alpha) \big)^2
  dy.$

The following algorithm is an active set type of algorithm.  
At a point $(\bs {\wt \pi}_p , \bs {\wt m}_p) \in \RR^{2p}$,
define
\begin{equation*}
  \bs{\wt \pi}_{p+1} :=
  (\wt \pi_1, \ldots, \lambda \wt \pi_i, (1-\lambda) \wt \pi_{i},
  \wt \pi_{i+1}, \ldots, \wt \pi_p),
  \text{ and }
  \bs {\wt m}_{p+1} := (\wt m_1, \ldots, \wt m_i, \wt m_i, \ldots, \wt m_p),
\end{equation*}
and consider
$  \MM_{n,p+1}(\bs {\wt \pi}_{p+1}, \bs {\wt m}_{p+1}),$
where $\lambda \in [0,1]$.
Let $\theta := (m_i - m_{i+1})/2$ and,
for $1 \le i \le p$,
let
\begin{equation}
  \label{eq:defn-Cip}
  \mathfrak{C}_{i,p} := \partial^2 / \partial \theta^2 \MM_{n,p+1}(\bs{\wt \pi}_{p+1}, \bs{\wt m}_{p+1}). 
\end{equation}
This is the curvature in the direction in which we separate out the $i$th component into two separate components.
Note that $\mathfrak{C}_{i,p} \equiv \mathfrak{C}_{i,p}(\bs{\wt \pi}_p, \bs{\wt m}_p)$ depends on $(\bs {\wt \pi}_p , \bs {\wt m}_p)$ but we will suppress that dependence below for succinctness.
Also, a priori, $\mathfrak{C}_{i,p}$ depends on $\lambda$.  However we can see from
\eqref{eq:10002} that in fact $\mathfrak{C}_{i,p}$ is minimized by taking $\lambda = 1/2$ always, so we do this from now on. 
We will use $\mathfrak{C}_{i,p}$ when
$(\bs {\wt \pi}_p , \bs {\wt m}_p)$ is a stationary point.

With the above setup, we describe our active set type of algorithm in Algorithm~\ref{alg:active-set} below.  Algorithm~\ref{alg:active-set} includes a call to a generic subroutine to find the optimum value of $\bs m$ given a fixed length $p$ and a fixed value of $\bs{\pi}$ above (corresponding to ``counts'' in the algorithm).  This is referred to as the ``fixed-$p$-subroutine'' in the algorithm. 
Our suggested implementation is to use a ``trust-region'' second-order method for this generic subroutine, see \cite[Section 5.1]{Fletcher:1987wy} or \cite[Section 4.2]{Nocedal:1999iy} for details of a trust region method.
Algorithm~\ref{alg:active-set}  also includes a call to a subroutine, for collapsing (approximately) unique entries, described in Algorithm~\ref{alg:activate-constraints}.
To optimize a function, the method requires to be able to compute the function, its gradient, and its Hessian. The latter two we have derived closed forms for, for error distributions that are either Gaussian, Laplace, or mixtures of Gaussians.  We use numerical integration to compute the objective function itself at this point.
In the algorithms, we use the notation $\bs m_{i:j}$ to denote the subvector of $\bs m$ given by the indices $\{ i, \ldots, j \}$.

\begin{algorithm} 
  \label{alg:active-set}
  \SetAlgoLined
  \SetKwInOut{Input}{input}\SetKwInOut{Output}{output}
  \Input{ $p^{(0)} \in \RR$,
    $\bs m^{(0)} := ( m^{(0)}_1, \ldots , m^{(0)}_p )$,
    counts$^{(0)} := ($counts$^{(0)}_1, \ldots, $counts$^{(0)}_p)$
    where $p = p^{(0)}$, and
    Tolerance parameter $eps$,
    Stepsize $\eta$}

  \Output{$\wh m_n(X_{(1)}),\dots,\wh m_n(X_{(n)})$}

  \While{end criterion not met}{
    Do fixed-$p$-subroutine($\bs m^{(i-1)}$, counts$^{(i-1)}$):
    Find 
    fixed-$p^{(i)}$ and fixed-counts optimal %
    $\bs m$, and  assign to $\bs m^{(i)}$\;

    Do activate-constraints-subroutine($\bs m^{(i)}$, counts$^{(i-1)}$, $eps$):
    run Algorithm~\ref{alg:activate-constraints} which collapses the
    non-unique entries in $\bs m^{(i)}$, and store the output in
    $\bs m^{(i)}$ and counts$^{(i)}$, and let $p^{(i)}$ be the new (smaller)
    number of unique entries\;

    Compute $\mathfrak{C}_{j,p}$ (see \eqref{eq:defn-Cip}) for each $j=1,\ldots,p$\;
    \uIf{$\min_j \mathfrak{C}_{j,p} \ge 0$}{End algorithm\;}
    \uElse{
      $k \leftarrow \argmin_i \mathfrak{C}_{i,p}$\;
      $p^{(i)} \leftarrow p^{(i)} + 1$\;
      counts$^{(i)} \leftarrow ($counts$^{(i)}_1, \ldots, $ counts$^{(i)}_k/2, $
      counts$^{(i)}_k/2$, $\ldots$, counts$^{(i)}_p)$\;
      $\bs m^{(i)} \leftarrow (m^{(i)}_1, \ldots, m^{(i)}_k - \eta, m^{(i)}_k + \eta, \ldots, m^{(i)}_p)$\;
    }
    $i \leftarrow i+1$\;
  }

  \tcc{Reconstruct full length solution}
  The solution vector is given by the (unique, sorted) elements  $m^{(K)}_i$, $i=1, \ldots, p_K$,  each repeated counts$^{(K)}_i$ times, respectively, where $K$ is the number of iterations run.
  \caption{Active set algorithm}
\end{algorithm}
\begin{algorithm}
  \label{alg:activate-constraints}
  \SetAlgoLined
  \SetKwInOut{Input}{input}\SetKwInOut{Output}{output}
  \Input{$\bs m := ( m_1, \ldots , m_p )$,
    counts$ := ($counts$_1, \ldots, $counts$_p)$,
    $eps$ (tolerance parameter)}

  \Output{$\bs{m}^{\text{new}} \in \RR^{\wt p}$,
  counts$^{\text{new}} \in \RR^{\wt p}$,  where $1 \le \wt p \le p$}
  
  newidx $\leftarrow 1$, begidx $\leftarrow 1$\;
  \For{$j \leftarrow 2$ \KwTo $p+1$}{
    \If{$(j == p+1)$ OR $(m_j - m_{\text{begidx}} > $ eps)}{
      $m^{\text{new}}_{\text{newidx}} \leftarrow  $  mean($m_{\text{begidx}:(j-1)})$\;
      counts$^{\text{new}}_{\text{newidx}} \leftarrow $  sum of counts$_{\text{begidx}:(j-1)}$\;
      begidx   $ \leftarrow j$\;
      newidx $\leftarrow$ newidx $+1$\;
    }
    $\wt p \leftarrow $ length of  $\bs m^{\text{new}}$\;
  }
  \caption{Activate constraints: group (approximately) non-unique entries in
    $\bs m$}
\end{algorithm}

In Appendix~\ref{sec:algorithm-notes} we provide a few comments about practical implementation of the algorithm.


\section{Extension to the Case of Unknown Noise Distribution}\label{sec: ext}

\subsection{Estimation of the Noise Distribution in the Semi-supervised Setting}

In general, full knowledge of the distribution of $\epsilon$  might not be available which means that one needs to estimate it.  In this case, it may be possible to collect a sample of $\epsilon$'s, $\epsilon_1^*, \ldots, \epsilon_{\Ne}^*$, from a separate data source.  These can be used to construct an estimate of $\Phi_{\epsilon}$ which can be then plugged into the objective function.  The sample of $\epsilon$'s does not necessarily need to be independent of the $Y$ or $X$ samples (note, for instance, \citealt{dattner2016adaptive}).
There are a variety of ways one may arrive at the sample of $\epsilon$'s.  
In some cases, the main data set may consist of unlinked covariates and responses, but there may be a smaller (or sub) data set of linked/paired covariates and responses, $(X^*_1, Y^*_1), \ldots, (X_{\Ne}^*, Y_{\Ne}^*)$.  In this case, one may run  the traditional monotone regression on this subset to obtain a monotone estimator $\wh m_{\Ne}^*$  from which one can compute the estimated residuals by $\wh \epsilon^*_i := Y^*_i - \wh m_{\Ne}^*(X^*_i)$, $i=1,\ldots, \Ne$.
In general, the previously-described setting might be considered to be one of {\it semi-supervised learning}, where only a  part of the data is unlinked.  It would be useful with such data to learn from all of it simultaneously.  This may be possible using the M-estimation framework we have proposed in this paper, but we leave an investigation of that question for future research.

\subsection{Estimation of the Noise Distribution with Longitudinal Responses}
\label{sec:estim-noise-distr}

Another common framework in which we may want to estimate $\Phi_{\epsilon}$ from data is the one where we have repeated (or longitudinal) observations on the response $Y$.  Assume we observe $X_1, \ldots, X_{\nx}$ as before and also $Y_{1,j}, \ldots, Y_{\ny, j}$, where  we take $j \in \{ 1, 2 \}$ (for simplicity).  We will impose the assumption that the distribution of $\epsilon$ is  symmetric around $0$.  We also assume that $Y_{i,j} = m_0(\tilde{X}_i) + \epsilon_{i,j}$ (for some $\tilde{X}_{i}$ which does not belong to our data set and need not  be observed), where $\epsilon_{i,1}$ is independent of $\epsilon_{i,2}$, and both are independent of all other error terms and all $X$ variables.
Then, as
in \cite{Carroll:2006cg}  and \cite{dattner2016adaptive}, 
we can let $Y^*_{i} := (Y_{i,1} + Y_{i,2}) / 2 = m_0(\tilde{X}_i) + \epsilon'_i$ where $\epsilon'_i := (\epsilon_{i,1} + \epsilon_{i,2})/2$.   If we let $\epsilon^*_{i} := (Y_{i,1} - Y_{i,2}) / 2 = (\epsilon_{i,1} - \epsilon_{i,2}) / 2$, then it follows from the assumption of $0$-symmetry that $\epsilon^*_i \sim \epsilon'_i$.  Then,  we can use $X_{1}, \ldots, X_{\nx}$ and $Y^*_1, \ldots Y^*_{\ny}$ as our unlinked data and $\epsilon^*_1, \ldots \epsilon_{\ny}^*$ to estimate $\Phi_{\epsilon'}$. 

Note that computing the estimator of $m$ in practice as described in Algorithm \ref{alg:active-set} generally requires computation of
derivatives of $\MM_n$, e.g., 
the gradients
$\frac{\partial}{\partial m_i} \MM_n(\bs m)$, where we use the same slight abuse of notation as in \eqref{eq: Mm}. The gradient depends on $\Phi_\epsilon$ and is given by \eqref{eq: gradient}. Hence, in the case where $\Phi_\epsilon$ is unknown, the gradient cannot be computed directly and has to be replaced by an appropriate estimator. In the setting of longitudinal responses, we have
\begin{equation*}
  \frac{\partial}{\partial m_i} \MM_n(\bs m)
  =2 n^{-1}-2n^{-2}
  \sum_{\alpha=1}^n\left\{ \Phi_{\epsilon'}(Y_\alpha^* - m_i)  + B(m_i-m_\alpha)\right\}
\end{equation*}
where $\Phi_{\epsilon'}$ denotes the common distribution function of $\epsilon_1',\dots,\epsilon_n'$ and where for all $m\in\RR$, $B(m)=E(\Phi_{\epsilon'}(\epsilon'+m))$.
 With $\hat\Phi$ the empirical distribution function based on the sample  $\epsilon^*_1, \ldots \epsilon_{\ny}^*$, the gradient can be estimated by 
$$2 n^{-1}-2n^{-2}
    \sum_{\alpha=1}^n\left\{ \hat\Phi(Y_\alpha^* - m_i)  + \hat B(m_i-m_\alpha)\right\}$$
where $\hat B(m)=\ny^{-1}\sum_{i=1}^{\ny}\hat\Phi(\epsilon_{i}^*+m)$. The advantage of this estimator is that it does not require any choice of tuning parameters.

\section{Demonstrations on Synthetic and Real Data} 
\label{sec:demonstr-synth-real}

\subsection{Computations on Synthetic Data}

In this subsection we present simulation studies for our method and compare our minimum contrast estimator to the deconvolution method of \cite{carp16}.  We also  compare to classical/linked (oracle) isotonic regression (which uses matching information that the other estimators do not use).
We will use mean-squared errors (MSE's) for comparison: for an estimator $\wh m$, we report $n^{-1} \sum_{i=1}^n (\wh m(X_{(i)}) - m_0(X_{(i)}))^2.$
There are 2 output tables,
Tables~\ref{tab:sims_n100} and \ref{tab:sims_n1000}, containing Monte Carlo estimates of MSE's.
The sample sizes are   taken to be $n=100$ and $n=1000$ for the first and second tables respectively.
In both tables, we used 
1000 Monte Carlo replications.
We used 5 different true mean functions $m_0$.  They are (up to translation and additive constants), together with the abbreviations that denote them, gathered in
Table~\ref{m0abbr}.
\begin{table}[!h]
\begin{center}
\begin{tabular}{c|c}
\hline
$m_0$ & Abbreviation   \\ 
\hline
 $x$   &  \lq\lq lin\rq\rq   \\
$0$  &  \lq\lq const\rq\rq \\
$2 \one_{[0,5)}(x) + 8 \one_{[5,10]}(x)$   &   ``step2'' \\
$ 5 \one_{[10/3,20/3)}(x) + 10 \one_{[20/3,10]}(x)$   &   ``step3'' \\
$(x^4 \one_{(0,5]}(x) - x^4 \one_{[-5,0)}(x)) / 120$   &  ``power''\\
 \hline
 \end{tabular}
\end{center}
\caption{The true monotone regression function used in the simulations and their abbreviations. } 
\label{m0abbr}
\end{table}
In the tables, our estimator is ``UL BDD'' (where ``UL'' stands for ``unlinked''), \cite{carp16}'s is ``UL CS'', and isotonic regression is ``L mono''  (for linked regression, that is, based on the classical case where all covariates and responses are perfectly linked).  Our simulations were performed taking both Laplace and Gaussian errors with standard deviation 1 (and both unlinked methods are well-specified).
In Figure~\ref{fig:sim_one-run} we present the output from a single Monte Carlo run, so that the true functions along with sample data and estimates can be visualized.
The 10 plots in the figure are all based on $n=100$ samples for each of 10 settings: the 5 monotone regression functions $m_0$  of Table \ref{m0abbr} with Laplace distributed errors (left column) and Gaussian errors (right column).  

To implement the deconvolution method of \cite{carp16}, we needed to monotonize the estimated CDF so that we could compute its generalized inverse.  \cite{carp16} do not mention a specific method for doing this; we chose to replace
$\hat F(x)$ with $\max( \hat F(y) : y \le x )$.  
 The deconvolution estimator at $X_{(1)}$  or $X_{(n)}$ was occasionally unstable because of the steepness of the deconvolution-estimated quantile function. In that case, we have dropped those values out without a noticeable effect.  Furthermore, the bandwidth for the  deconvolution estimator of the CDF was chosen using the bootstrap method of  \cite{Delaigle:2004fc}.

From the output one can see that it is not always the case that the Gaussian noise is harder for our method than Laplace is. This is an interesting finding because it is known that deconvolution is harder with Gaussian noise than it is with Laplace; see e.g. \cite{Fan91}.  This suggests that, although unlinked regression is tightly connected to deconvolution, considering the problem from this point of view may not be the most efficient approach.  The output also shows that the CS estimator performs poorly especially when $m_0$ has discontinuities.  In general our estimator is  competitive with the CS deconvolution estimator and in some cases is significantly better.

\begin{figure}[p] 
  \centering
  \includegraphics[scale=.9]{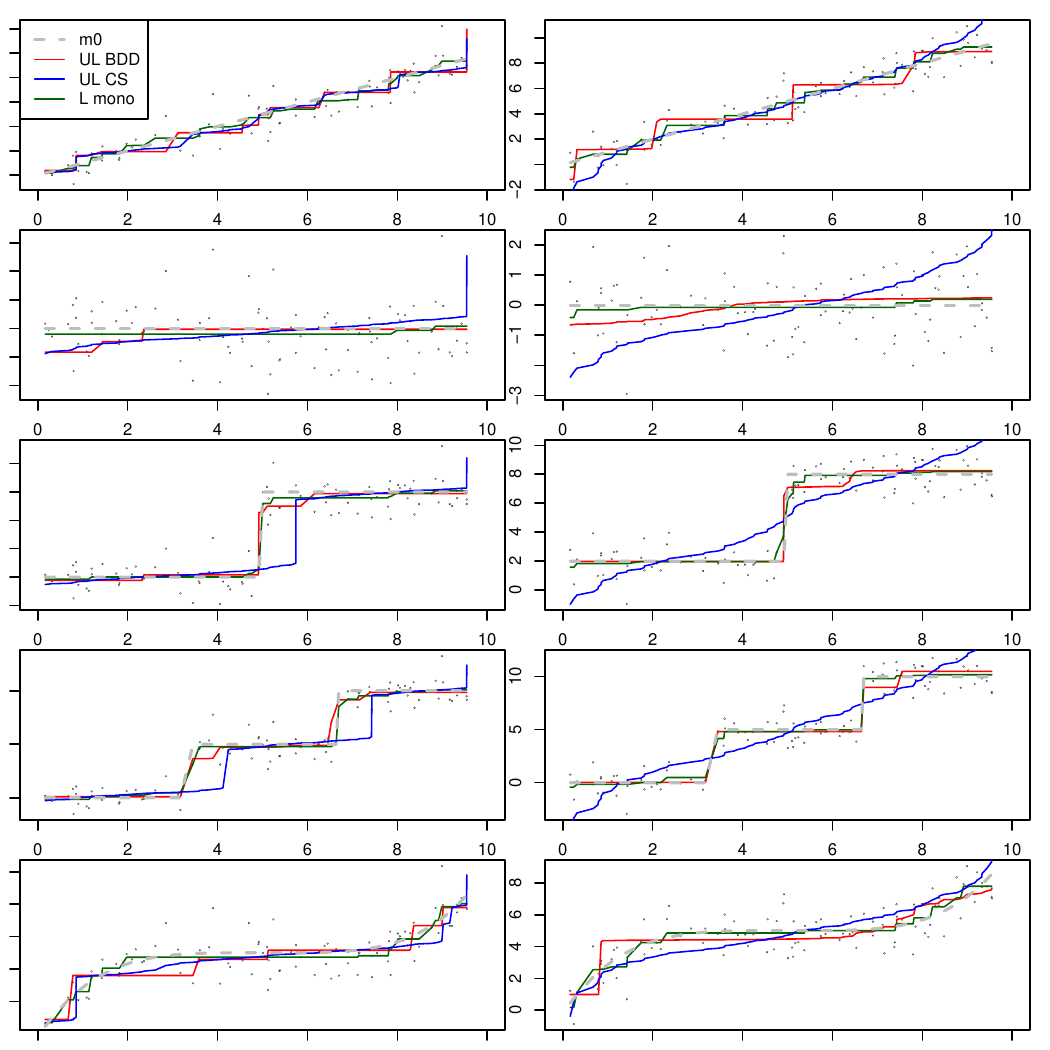}   
  \caption[Single Monte Carlo output.]{Output from a single Monte Carlo simulation, with $n=100$, and $X_i$ i.i.d.\ uniform on $[0,10]$.  The left column has Laplace errors and the right has Gaussian errors, both with standard deviation $1$. The dotted gray line is the true $m_0$, the red line is our minimum contrast estimator, the blue line is the deconvolution estimator of
    \cite{carp16},
    and the green line is a classical/linked isotonic regression.  }
  \label{fig:sim_one-run}
\end{figure}

\begin{table}
  \centering
  \begin{tabular}{l|r|r|r}
    \hline
    & UL BDD & UL CS &  L mono \\
    \hline
    lin, Laplace & 0.31 & 0.27 & 0.16\\
    \hline
    const, Laplace & 0.18 & 0.25 & 0.05\\
    \hline
    step2, Laplace & 0.33 & 2.84 & 0.09\\
    \hline
    step3, Laplace & 0.43 & 2.74 & 0.12\\
    \hline
    power, Laplace & 0.29 & 0.47 & 0.13\\
    \hline
    lin, Gauss & 0.48 & 0.78 & 0.16\\
    \hline
    const, Gauss & 0.10 & 1.23 & 0.05\\
    \hline
    step2, Gauss & 0.19 & 2.48 & 0.09\\
    \hline
    step3, Gauss & 0.32 & 2.59 & 0.12\\
    \hline
    power, Gauss & 0.43 & 0.59 & 0.14\\
    \hline
  \end{tabular}
  \caption{Monte Carlo'd MSE's, $n=100$.}
  \label{tab:sims_n100}
\end{table}

\begin{table}
  \centering  
  \begin{tabular}{l|r|r|r}
    \hline
    MSE's & UL BDD & UL CS & L mono \\
    \hline
    lin, Laplace & 0.10 & 0.10 & 0.03\\
    \hline
    const, Laplace & 0.06 & 0.13 & 0.01\\
    \hline
    step2, Laplace & 0.12 & 3.58 & 0.01\\
    \hline
    step3, Laplace & 0.15 & 3.34 & 0.02\\
    \hline
    power, Laplace & 0.09 & 0.36 & 0.03\\
    \hline
    lin, Gauss & 0.29 & 0.18 & 0.03\\
    \hline
    const, Gauss & 0.03 & 0.70 & 0.01\\
    \hline
    step2, Gauss & 0.07 & 1.04 & 0.01\\
    \hline
    step3, Gauss & 0.13 & 1.40 & 0.02\\
    \hline
    power, Gauss & 0.25 & 0.25 & 0.03\\
    \hline
  \end{tabular}
  \caption{Monte Carlo'd MSE's, $n=1000$.}
  \label{tab:sims_n1000}
\end{table}

\subsection{Computations on CEX Data}

Figure~\ref{fig:CEX-plots} shows plots based on the United States' Consumer Expenditure Survey (CEX).  The CEX survey has detailed data on both income and the expenditure patterns of so-called ``U.S.\ consumer units'' (roughly,  households), see \cite{Rugglesetal}. The CEX consists of two surveys, the ``Interview'' and the ``Diary'';  the data we use here come from the former. Since the CEX survey has data on both income and expenditures, we can use regular (matched / standard) regression techniques.
We compare a U.S.\ consumer unit's food expenditure to its income by regressing the former on the latter.
Moreover, by simply ignoring the $X$-$Y$ pairing information  we can also use our approach detailed above for the unlinked setting, and then compare the results obtained with both approaches.
Note that we prefer here to use  data that are matched, for the sake of being able to ``validate'' our results, but 
  there are many settings where matching naturally lacks; for instance a firm may be able to gather information on an individual consumer's expenditures on the firm's products, but the firm would not be able to know the individual's income information.  They would be able to access that expenditure information (at least nationally) through the CEX, which creates a data set with unlinked covariates and responses.

We consider the interview data only from the second quarter of 2018, for which there are approximately $6000$ respondents.  We narrow this down to $2164$ respondents who provided the relevant information, had income no larger than \$$250,000$, and reported a non-negative response for both income and food expenditure.
The survey actually follows each individual for four consecutive quarters, but we only included those who were surveyed in both quarter 2 and quarter 3. The ``residuals" were computed as described in the previous section: for each individual $i$, we computed $\tilde \epsilon_i := (Y_{i,1} - Y_{i,2})/2$ where $Y_{i,1}$ is the quarter 2 response and $Y_{i,2}$ is the quarter 3 response.  The error distribution is assumed to be Laplace distributed with $\lambda = \sqrt{ \hat \sigma^2 / 2}$ where $\hat \sigma^2$ is the empirical variance of $\tilde \epsilon_1, \ldots, \tilde \epsilon_{2164}$. 
We chose to assume that the errors follow the Laplace distribution rather than to use the full method detailed in Subsection~\ref{sec:estim-noise-distr} because the former is much faster to run.  Choosing the Laplace distribution was based on observing that it fits much better the distribution of the residuals than the Gaussian one.  
In Figure~\ref{fig:CEX-plots}, 
the ``UL BDD'' line is the unlinked monotone minimum contrast estimator proposed in this paper.
This estimate is {\it fully} data driven, using no oracle (matching) information.
The ``L mono'' line corresponds  to monotone regression estimator based on the matched data; i.e., the Grenander-type estimator.  Similarly the ``L linear'' line is a linear regression estimator based on the matched data.   The ``UL CS'' line corresponds to the deconvolution estimator 
of     \cite{carp16} (using the same choice of $\lambda$ we used in our method). 
We also implemented a type of ``unlinked oracle'' estimator in which we used a Gaussian-mixture as the error distribution, labeled ``UL-oracle BDD'':  For this estimator we used the residuals from the matched monotone regression to estimate $\Phi_\epsilon$ (so this is oracle information which would not be available in a true unmatched problem).  We used a four component Gaussian mixture (with no variance constraint) to approximate the error distribution, which was fit using an EM algorithm to converge to a local maximum (it is well known that the global maximum is infinite).  The estimate has
  weights $( .56, .05, .33, .06)$,
  means/locations $(-345, -416, 326, 1710)$
  and standard deviations $(322, 75, 562, 1286)$.
  The
  mixture distribution is a much better fit to the residual distribution than either a
a single Gaussian or Laplace distribution, since the residual distribution is multimodal and heavy tailed.
  This estimator of $\Phi_\epsilon$ is much more dispersed than the fully data driven one; for instance, the former has standard deviation $743$ instead of $266$ for the latter one.  This leads to differences in these two estimators. 
Finally, the ``UL quantile'' line is based on matching the empirical quantiles of the $Y$  and $X$ samples: it is simply given by (connecting linearly) the points $(X_{(1)}, Y_{(1)}), \ldots, (X_{(n)}, Y_{(n)})$.

Our ``UL BDD'' estimator is somewhat accurate although it does differ noticeably from the oracle isotonic regression as well as the ``UL-oracle BDD'' estimator.  The estimator of \cite{carp16} does very poorly on this data set.  We suspect that inaccuracy in choice of the error distribution causes difficulty for both of the unlinked estimators, which is of course expected.


\begin{figure}
  \centering
  \includegraphics[scale=.57]{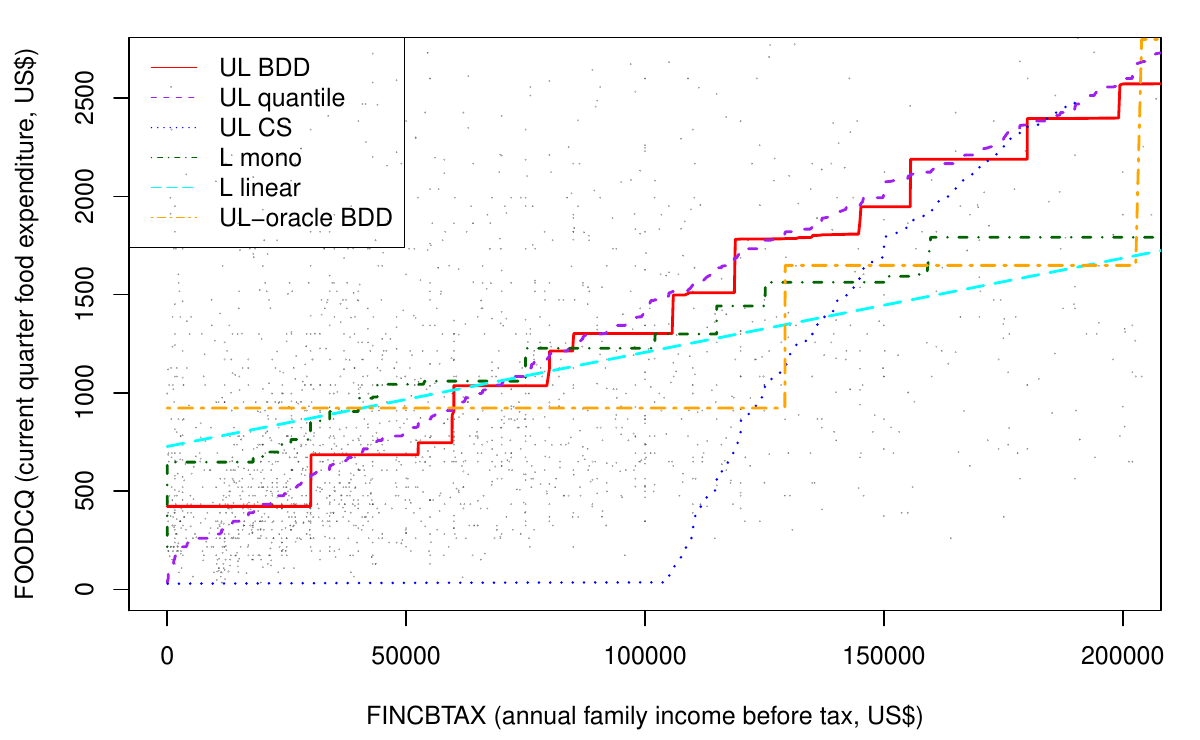} 
  \caption{CEX Interview Survey: Family income vs.\ Food expenditure}
  \label{fig:CEX-plots} 
\end{figure}


\section{Conclusions and Directions for Future Research}\label{sec: concl}

In this paper, we have presented a general method for unlinked regression with a monotonic regression function, and developed basic theory for the resulting estimator.
We believe our approach will generalize to other (identifiable) unlinked regression settings.
We have introduced a variant of an active set
algorithm for computing the estimator, and demonstrated its use on a real data example in a fully data driven way in which we estimated the unknown error distribution.
There are many remaining questions about this problem and about our method that future work could answer.
\begin{enumerate}[leftmargin=*]



\item Our current study is restricted to the case of a univariate predictor.
  Studying both theory and practice when dimension is larger than $1$ will be an important avenue for future work.  In the case of linear regression, several works (\cite{Abid:2017ws,Pananjady:2017kj,Pananjady:2018hd,Unnikrishnan:2018gp}) have already begun this study, although those works focus mostly on the case of Gaussian noise.

\item Finding (minimax) lower bounds for the rate of convergence seems to be hard to obtain in our setting. Such bounds are needed for a more complete theoretical understanding of  the problem setting. In earlier literature, the closest result we are aware of is the one obtained in 
  \cite{Rigollet:2018wia}, who provide a minimax lower bound in the case of a Gaussian noise in the shuffled monotone regression. However, it is not known if the minimax lower bound  derived by these authors applies to our problem, since the unlinked problem is statistically more difficult than the shuffled problem.
  
\item One of the major differences in unlinked regression from linked regression is that in the former the specification of the error distribution is crucially important.  As shown in the introduction, if the error distribution is unknown then the model is not even identifiable. It would be helpful to understand the general properties of unlinked regression models and of our method in particular when one has partial but incomplete knowledge of the error distribution (e.g., some moment parameters can be estimated well but the full distribution is not known precisely).
  
\item \label{item:2}
  \cite{carp16}  allowed for so-called ``contextual variables''; for instance, if the unit $i$ of observation is an individual, both $Y_i$ and $X_i$ may be paired with a contextual variable $Z_{i,Y}$ and $Z_{i,X}$ such as the individual's age.  One may ``match'' $Y$'s and $X$'s with equal (or similar) ages, and then one may consider unmatched regression on these partially matched data sets.  This is what   \cite{carp16} proposed in the case of discrete, perfectly (noiselessly) observed contextual variables.   More broadly,  one may use so-called linkage methods \citep{herzog2007data} to partially link $Y$ and $X$ (effectively reducing the noise level) when the contextual variables are not as idealized, and then perform linked regression.  This methodology could be broadly useful in the linkage literature and warrants further study.

\end{enumerate}


\section*{Acknowledgements} 

The second author is supported in part by NSF Grant DMS-1712664. The  third author is supported in part by MME-DII (ANR11-LBX-0023-01) and by the FP2M federation (CNRS FR 2036).

\appendix

\section{Bounding the Integrals in (\ref{I1}) and (\ref{I2}).}\label{appendix I1}

\par \noindent  Using Assumptions  $A0$ and \ref{assm:A1:m0-sup} we can write that 
\begin{eqnarray*}
\int_{-\infty}^0 H_0(y) dy  & =  &  \int_{-\infty}^0 \int_{\R} \Phi_\epsilon(y - m_0(x) ) dF_0(x)  dy \\
& \le  &  \int_{-\infty}^0  \Phi_\epsilon(y +K_0 )  dy   < \infty,
\end{eqnarray*}
(note  that  this also follows from $E(\vert Y  \vert ) < \infty$ and integration by parts.)  Similarly it can be shown  that $\int_{0}^\infty (1-H_0(y)) dy  < \infty$.
Hence, $I_1 \le \int_{0}^\infty (1-H_0(y)) dy  +  \int_{-\infty}^0 H_0(y) dy < \infty$.  Also,
\begin{eqnarray*}
\int_{\R} \left(\Phi_\epsilon(y- m_0(x)) - H_0(y)  \right)^2 dy &\le  & 2 \int_{-\infty}^0 \left(\Phi_\epsilon(y- m_0(x))^2  + H_0(y)^2  \right) dy  \\
&& +  \int_{0}^\infty (1-H_0(y))^2 dy\\
& \le & 2 \int_{-\infty}^0 \Phi_\epsilon(y +  K_0) dy   + 2 \int_{-\infty}^0 H_0(y) dy \\
&& +   \int_{0}^\infty (1-H_0(y)) dy\\
& < & \infty,
\end{eqnarray*}
as shown above; this implies, by Fubini's Theorem, that $I_2 < \infty$. \hfill $\Box$

\section{Basic Empirical Process Theory Definitions and a Fundamental Preservation Result}
\label{sec:some-rudiments-EPT}

For a (possibly random) signed measure $Q$ on a (measurable) space $\mc X$ and a measurable function $f$ on $\mc X$, we denote $Q f := \int_{\mc X} f dQ$. For some class of functions $\mathcal{G}$, we can define
\begin{itemize}
\item its $\epsilon$-covering number   $N(\epsilon, \mathcal{G}, \Vert \cdot \Vert)$ with respect to some norm $\Vert \cdot \Vert$ is defined as the smallest integer $N  > 0$ such that there exists $g_1, \ldots, g_N  $ satisfying that for any $g \in \mathcal{G}$, there exists $i \in \{1, \ldots, N\}$ such that 
$$ \Vert g - g_i \Vert < \epsilon ,$$

\item  its $\epsilon$-bracketing number   $N_B(\epsilon, \mathcal{G}, \Vert \cdot \Vert)$ with respect to some norm $\Vert \cdot \Vert$ is defined as the smallest integer $N  > 0$ such that there exist pairs $(h_1,k_1),\dots,(h_N,k_N)$
satisfying that for any $g \in \mathcal{G}$, there exist
$i \in \{1, \ldots, N\}$ such that 
$h_i \le g \le k_i$ and 
$$ \Vert k_i - h_i \Vert < \epsilon.$$
\end{itemize}
From the definition of the covering and bracketing numbers it can be easily shown \citep[pp. 83--84]{aadbook} that if $\|\cdot \|$ is an $L_p$ norm, for some $1 \le p \le \infty$, then for any $\delta> 0$,
\begin{eqnarray}\label{RelNumb}
N(\delta, \mathcal{G}, \Vert \cdot \Vert)  \le N_B(2\delta, \mathcal{G}, \Vert \cdot \Vert).
\end{eqnarray}
Also, if the class $\mathcal{G}$ admits an envelope $F$, then define for $\eta > 0$ the number
\begin{eqnarray}\label{Jdef}
J(\eta, \mathcal{G}) = \sup_{Q} \int_0^\eta \sqrt{1 + \log N(\delta \Vert F\Vert_{Q,2}, \mathcal{G}, L_2(Q))}
 d\delta
\end{eqnarray} 
where the supremum is taken over all discrete probability measures $Q$ such that $\Vert F\Vert _{Q, 2}: = \Big(\int \vert F(x) \vert^2 dQ(x)\Big)^{1/2}  < \infty$.

\medskip

We finish this section by the following preservation result.

\begin{prop}
  \label{prop: preserv-2}
Let $\Vert \cdot\Vert$ be some norm, and  $\mathcal G_1$ and $\mathcal G_2$ two classes of functions.
For fixed $\lambda_1, \lambda_2 $ such that $(\lambda_1, \lambda_2) \ne (0, 0)$, define the class
\begin{eqnarray*}
\lambda_1 \mathcal G_{1}  + \lambda_2 \mathcal{G}_2   =  \left\{  h = \lambda_1  g_1  + \lambda_2 g_2:  (g_1, g_2) \in \mathcal G_1 \times \mathcal G_2 \right \}.
\end{eqnarray*}
Then, for any $\epsilon > 0$
\begin{eqnarray*}
N(\epsilon, \lambda_1 \mathcal G_{1}  + \lambda_2 \mathcal{G}_2, \Vert \cdot \Vert) \le N( (\vert \lambda_1 \vert  +\vert \lambda_2 \vert)^{-1} \epsilon, \mathcal G_1,  \Vert \cdot \Vert) \times  N( (\vert \lambda_1\vert  + \vert \lambda_2 \vert)^{-1} \epsilon, \mathcal G_2,  \Vert \cdot \Vert). 
\end{eqnarray*}

\end{prop}

\begin{proof}[Proof of Proposition \ref{prop: preserv-2}.]  \  Fix $\epsilon > 0$. Let $h  = \lambda_1 g_1 +  \lambda_2 g_2 \in \lambda_1 \mathcal G_{1}  + \lambda_2 \mathcal{G}_2$,  $N_1 =  N(\epsilon (\vert \lambda_1\vert  + \vert \lambda_2 \vert)^{-1}, \mathcal G_1,  \Vert \cdot \Vert)$ and $N_2 = N(\epsilon (\vert \lambda_1\vert  + \vert \lambda_2 \vert)^{-1}, \mathcal G_2,  \Vert \cdot \Vert)$.  We assume in the sequel that both $N_1$ and $N_2$ are finite since otherwise, the inequality in Proposition \ref{prop: preserv-2} is trivial.  Then, there exists a pair $(i, j) \in \{1, \ldots, N_1\} \times \{1, \ldots, N_2\}$ and $(g_{1,i}, g_{2, j})$ such that $\Vert g_1 -  g_{1, i} \Vert < \epsilon  (\vert \lambda_1\vert  + \vert \lambda_2 \vert)^{-1}$ and $ \Vert g_2 -  g_{2,j} \Vert < \epsilon (\vert \lambda_1\vert  + \vert \lambda_2 \vert)^{-1}$. Then, by the triangle inequality we have that
\begin{eqnarray*}
\Vert h -  \lambda_1 g_{1, i}  - \lambda_2 g_{2, j} \Vert   &=  &   \Vert \lambda_1 (g_1 - g_{1, i})  + \lambda_2 (g_2 - g_{2, j}) \Vert  \\
 &\le &  \vert \lambda_1 \vert \Vert g_1 -  g_{1, i}  \Vert +  \vert \lambda_2 \vert \Vert g_2 -  g_{2, j}  \Vert  <  \epsilon
\end{eqnarray*}
which completes the proof.
\end{proof}

\section{Wasserstein Distance Lemmas}
\label{sec:wass-dist-lemm}

Recall that $W_1(F,G)$ denotes the first Wasserstein distance between two probability distributions with distribution functions $F$ and $G$.
The following is a well-known representation of the Wasserstein-$1$ distance in one dimension (see, e.g., \cite{bobkov2019one}, or e.g., Proposition 2 of \cite{Meis}).
\begin{prop}
  \label{prop:1}
  Let $F,G$ be distribution functions on $\RR$, each having finite first
  moment.  Then
  \begin{equation*}
    W_1(F,G) =
    \int_\RR | F(x) -G(x) | dx.
  \end{equation*}
\end{prop}

\medskip
\noindent The following  is Proposition 3 of \cite{Meis}.

\begin{prop}
  \label{prop:2}
  Let $F, G$ be two distribution functions supported on $[0,V]$ and suppose $H$ is a distribution function supported on a finite set of points.  Then
  \begin{equation*}
    W_1(F,G) \le C(V,H) W_1( F \star H, G \star H)
  \end{equation*}
  where $C(V,H) > 0$ depends only on $V,H$.  
\end{prop}

\section{Gradient, Curvature, and Other Algorithmic Computations}

In this section, we prove  \eqref{eq: gradient} and  give  an explicit formula of $B(m)$ in the case of Laplace, Gaussian, and Gaussian-mixture distributions.
\label{sec:gradient-computation}

\subsection{Proof of \eqref{eq: gradient}}\label{appendixgradient}




We have
\begin{equation}
  \label{eq:5}
  \begin{split}
    \MoveEqLeft
    \frac{\partial}{\partial m_i} \MM_n(\bs m)=\int_{\RR} 2 \lp \HH_n(y) - n^{-1} \sum_{\alpha=1}^n \Phi_{\epsilon}(y- m_\alpha)  \rp n^{-1} f_{\epsilon}(y - m_i) dy \\
    & = 2n^{-2}
    \sum_{\alpha=1}^n \int_{\RR} \left\{\one_{[Y_{\alpha}, \infty)}(y) f_\epsilon(y-m_i)
    - \Phi_\epsilon(y-m_\alpha) f_\epsilon(y-m_i)\right\} \; dy \\
    & =2 n^{-2}    \sum_{\alpha=1}^n  \left\{\int_{Y_{\alpha}-m_i}^\infty f_\epsilon(y) \; dy 
    - \int_\RR \Phi_\epsilon(y-m_\alpha+m_i) f_\epsilon(y) \; dy \right\},
  \end{split}
\end{equation}
and  \eqref{eq: gradient} follows.


\subsection{Curvature  Derivation}

It is convenient to re-parameterize the objective function as was done in the
algorithm development in the main text.  Let
$\wh{\HH}_{\bs m}(y) := \sum_{j=1}^p \pi_j \Phi_\epsilon(y-m_j)$, and then recall
\begin{equation*}
  \MM_{n,p}(\bs m) := \int ( \HH_n(y) - \wh{\HH}_n(y) )^2 dy
\end{equation*}
where $\pi_j \ge 0$ and $\sum_{j=1}^p \pi_j = 1$.
(Here $n = n_y$, and $n_x$ is defined implicitly in terms of the $\pi_j$ and $p$.)
Then, as derived above but in the new notation, we have
\begin{align*}
  \frac{\partial}{\partial m_i} \MM_{n,p}(\bs m)
  & = 2 \pi_i \int ( \HH_n(y) - \wh{\HH}_n(y) ) f_\epsilon (y-m_i) dy \\
  & = 2 \pi_i \lp n^{-1} \sum_{\alpha= 1}^n 1 - \Phi(Y_\alpha - m_i)
    - \sum_{j=1}^p \pi_j E \Phi_\epsilon(\epsilon + m_i - m_j) \rp \\
  & = 2 \pi_i \lp 1 - n^{-1} \sum_{\alpha=1}^n \Phi_{\epsilon}(Y_\alpha -m_i)
    - \sum_{j=1}^p \pi_j B(m_i - m_j) \rp
\end{align*}
where again $B(m) := E \Phi_{\epsilon}(\epsilon + m)$.
Then (assuming $i \ne j$)
\begin{align*}
  \frac{\partial^2}{\partial m_i^2} \MM_{n,p}(\bs m)
  &= 2 \pi_i ( \sum_{\alpha=1}^n n^{-1} f_\epsilon(Y_\alpha - m_i)
    - \sum_{j \ne i}^p \pi_j B'(m_i - m_j) \\
  \frac{\partial^2}{\partial m_i \partial m_j} \MM_{n,p}(\bs m)
  & = 2 \pi_i \pi_j B'(m_i-m_j).
\end{align*}
Thus, in order to compute the curvature (i.e., $\partial^2  / \partial \theta^2$), it suffices to compute $B'$ (which we do below in the three cases we consider).

One more set of calculations that are useful are the following; we have
\begin{equation*}
  \frac{\partial}{\partial \theta} \MM_n(\bs m) =
  \frac{2}{\sqrt{2}}   \int_{\RR}
  \lp \HH_n(y) -  \sum_{\alpha=1}^p \pi_\alpha \Phi_{\epsilon}(y- m_\alpha) \rp
  ( \pi_i f_{\epsilon}(y - m_i) - \pi_j f_{\epsilon}(y - m_j) ) dy ,
\end{equation*}
\begin{equation}
  \label{eq:10002}
  \begin{split}
    \MoveEqLeft \frac{\partial^2}{\partial \theta^2} \MM_n(\bs m) =
    \int_{\RR}
    \lp \wh{\HH}_{\bs m}(y) - \HH_n(y) \rp \times
    \\
    &
    \qquad \qquad
    ( \pi_i f'_{\epsilon}(y - m_i) + \pi_j f'_{\epsilon}(y - m_j) ) +
    (\pi_i f_\epsilon(y-m_i) - \pi_j f_\epsilon(y-m_j))^2 dy .
  \end{split}
\end{equation}


\subsection{Computations for  Laplace Distribution}
\label{sec:grad-comp-lapl}

Let $\lambda > 0$ and assume that
\begin{equation*}
  \Phi_{\epsilon}(z) :=
  \begin{cases}
    2^{-1} e^{-|z|/\lambda} & \text{ if } z \le 0 \\
    1 - 2^{-1} e^{-z/ \lambda} & \text{ if } z > 0
  \end{cases}
\end{equation*}
and let $f_{\epsilon}(y) := \Phi'_\epsilon(y) = e^{-|z|/\lambda} / 2 \lambda$ for $z \in \RR$.
We compute $B(m)$  for $m \in \RR$.  
We have
\begin{align*}
  B(m) =
  \int_{-\infty}^0 (4\lambda)^{-1} e^{y / \lambda} e^{ -|y- m|/ \lambda } dy
  + \int_0^\infty (1 - 2^{-1} e^{-y/\lambda}) e^{- |y - m|/ \lambda} (2\lambda)^{-1} dy
\end{align*}
which equals
\begin{align*}
  \MoveEqLeft
  \int_{-\infty}^{m \wedge 0} (4\lambda)^{-1} e^{(y - m/2) 2 / \lambda} dy
  + \int_{m \wedge 0}^0 (4 \lambda)^{-1} e^{m / \lambda} dy \\
  & + \int_0^{m \vee 0} (2\lambda)^{-1} ( e^{ (y-m)/\lambda} - 2^{-1} e^{-m /\lambda}) dy \\
  & + \int_{m \vee 0}^\infty \inv{2 \lambda} e^{ - (y-m)/\lambda}
    - \inv{4 \lambda} e^{- (y-m/2) 2 / \lambda} dy.
\end{align*}
If $m \le 0$, then
\begin{equation}
  \label{eq:7}
  B(m) =
  \inv{8} e^{m/\lambda} 
  + \frac{|m|}{4 \lambda} e^{m / \lambda}
  + 0
  + \frac{3}{8} e^{m / \lambda}.
\end{equation}
If $m \ge 0$, then 
\begin{equation}
  \label{eq:8}
  B(m) =
  \inv{8} e^{-m/ \lambda}
  + 0
  + \lp \inv{2} - \inv{4 \lambda} e^{-m / \lambda} (2 \lambda + m) \rp
  + \lp \inv{2} - \inv{8} e^{-m / \lambda} \rp.
\end{equation}
This gives an explicit formula for $B(m)$.

We can then compute that for any $m \in \RR$ (including $m =0$),
$B'(m) = e^{-|m|/\lambda} (\lambda + |m|) / (4 \lambda^2)$.
The calculations for $B'(m)$ are as follows.
For $m \le 0$, we have
\begin{align*}
  B'(m) = \inv{2 \lambda} e^{m/\lambda}
  = ((4 \lambda)^{-1} e^{m / \lambda} + m (4 \lambda^2)^{-1} e^{m/\lambda})
  & = e^{m / \lambda} \lp \frac{2 \lambda - \lambda - m }{4 \lambda^2} \rp \\
  & = e^{-|m|/\lambda} \frac{\lambda + |m|}{4 \lambda^2}.
\end{align*}
And for $m \ge 0$ we have
\begin{align}
  B'(m) = - \lp \inv{4 \lambda} e^{-m/\lambda} -
  \frac{2 \lambda + m}{4 \lambda^2} e^{-m / \lambda} \rp
  = e^{-m / \lambda} \lp \frac{2 \lambda + m}{4 \lambda^2}
  - \frac{ \lambda}{4 \lambda^2} \rp
  = e^{-m/\lambda} \frac{m + \lambda}{4 \lambda^2}.
\end{align}



\subsection{Computations for Gaussian Errors}\label{sec: Gaussian}

Now we consider the case where, for some $\sigma > 0$,  $\Phi_{\epsilon} = \Phi(\cdot / \sigma)$ is the cumulative distribution function of a $N(0, \sigma^2)$  random variable.
It turns out we can write  $B(\cdot)$ 
in terms of  $\Phi$:
by Corollary 1
of \cite{Ellison:1964je},
$B(m) = E \Phi_\epsilon(N(m, \sigma^2)/\sigma) =  \Phi(m / \sigma \sqrt{2})$.
Thus we also have
\begin{align*}
  B'(m)
  & = N(0,\sigma^2) \text{ density}
  = \frac{1}{\sqrt{4 \pi \sigma^2}} e^{-m^2/ 4 \sigma^2} , \\
  B''(m)
     & =  -\frac{2m}{ (2 \sigma)^3 \sqrt{\pi }} e^{-m^2/ 4 \sigma^2} .
\end{align*}


\subsection{Mixtures of Gaussian}

This again relies on Corollary 1 of  \cite{Ellison:1964je}.
Recall $\Phi$ is the CDF of a $N(0,1)$ variable.  Assume  for an integer $L \ge 1$, locations $\mu_i \in \RR$, and weights $\lambda_i$ summing to $1$ that
\begin{equation*}
  \Phi_\epsilon = \sum_{i=1}^L \lambda_i
  \Phi \lp \frac{ \cdot - \mu_i }{\sigma_i} \rp.
\end{equation*}
Then
\begin{equation}
  \label{eq:103}
  B(m) = E \Phi_{\epsilon}(\epsilon + m)
  = \sum_{i,j}^L \lambda_i \lambda_j
  E \Phi \lp \frac{Y - \mu_i + m}{\sigma_i} \rp
\end{equation}
where $Y \sim N(\mu_j, \sigma^2_j)$, so
$\frac{Y - \mu_i + m}{\sigma_i}  \sim N((\mu_j - \mu_i +m)/\sigma_i, \sigma^2_j / \sigma^2_i)$ and thus
\eqref{eq:103} equals
(by Corollary 1 of  \cite{Ellison:1964je})
\begin{equation*}
  \sum_{i,j}^L \lambda_i \lambda_j
  \Phi \lp \frac{(m-\mu_i + \mu_j)/\sigma_i}{\sqrt{1 + \sigma^2_j / \sigma^2_i}} \rp
  =   \sum_{i,j}^L \lambda_i \lambda_j
  \Phi \lp \frac{\mu_j-\mu_i + m}{\sqrt{\sigma_i^2 + \sigma^2_j }} \rp.
\end{equation*}
We then have that
\begin{equation*}
  B'(m) =
  \sum_{i,j}^L \lambda_i \lambda_j
  \phi \lp \frac{\mu_j-\mu_i + m}{\sqrt{\sigma_i^2 + \sigma^2_j }}\rp
  \frac{1}{\sqrt{\sigma_i^2 + \sigma^2_j }}.
\end{equation*}


\subsection{On Implementation of Algorithm~\ref{alg:active-set}}
\label{sec:algorithm-notes}


A few remarks about the implementation of Algorithm~\ref{alg:active-set} are as follows. The entries of the initial counts vector should sum to $n_x$; this relationship is then preserved throughout the algorithm. 
Generally the initializer for $\bs m$ is taken to be a constant (e.g., the median of the responses with $p^{(0)}=1$).  
In the algorithm we did not take care to force the counts to be integers when we divided by $2$, but this can easily be done and should be done for easy interpretability.   As end criterion, one can 
iterate for a fixed number $K$ of steps, or one can iterate until a stopping rule (e.g., the objective function decrease is smaller than a fixed tolerance level) is satisfied.   A heuristic choice for the parameter eps is  
$\text{eps}= (Y_{(n)} - Y_{(1)}) / ( n^{1/3} \sigma )$
where $\sigma^2$ is the variance of $\epsilon$, and 
$n^{1/3}$ is motivated by properties of classical isotonic regression.

\section{Proofs}
\label{sec:proofs}
In this section we provide our proofs. 
\subsection{A  Preparatory Lemma}
We begin with a lemma that will be used several times in the proofs of the main results.
  \begin{lemma}\label{lem: inv}
  Let $m \in \mc M$.  If $F_0$ is continuous, then
 \begin{eqnarray}
    \int_{ \RR}  \left(  \int\one_{\{m(x)\leq w\}}dF_0(x)  -  F_0 \circ m^{-1}(w) \right)^2  dw     
    =     0.  
  \end{eqnarray}
\end{lemma}

\begin{proof}[Proof of Lemma \ref{lem: inv}]   Recall that $m^{-1}$ is defined by \eqref{eq: inverse}
where the infimum of an empty set is defined to be $1$. If the set in \eqref{eq: inverse} is non-empty, then the infimum is achieved by right-continuity of $m$. Hence, we have $m\circ m^{-1}(y)\geq y$ for all $y\leq m(1)$. Now, consider $x\in[0,1]$ and $y\leq m(1)$ such that $m(x)\geq y$. Since the infimum in \eqref{eq: inverse} is achieved this implies that $x\geq m^{-1}(y)$. Conversely, if we have $x\geq m^{-1}(y)$ then monotonicity of $m$ implies that $m(x)\geq m\circ m^{-1}(y)$ where as mentioned above, $m\circ m^{-1}(y)\geq y$.   It follows that
for all $x\in[0,1]$ and $y\leq m(1)$ we have the equivalence
\begin{equation}\label{eq: switch}
m(x)\geq y \Leftrightarrow x\geq m^{-1}(y).
\end{equation}
Now, consider  $y>m(1)$. The set in \eqref{eq: inverse} is empty and therefore,  $m^{-1}(y)=1$ by definition. The left-hand inequality in (\ref{eq: switch}) does not hold if $x\in[0,1]$, and the right-hand inequality does not hold neither if $x<1$ since $m^{-1}(y)=1$. This mean that the above equivalence holds for all $x\in[0,1)$ and $y\in\RR$.
Let $X$ be a random variable with distribution function $F_0$. Since $P(X=1)=0$ by assumption, it follows that for all $w\in\RR$, we have
\begin{equation}\label{eq: inv} 
 P(m(X)  <  w)  = P(X < m^{-1}(w))
\end{equation} and therefore,
$$ P(m(X)  \le  w)  -  P(X \le m^{-1}(w))=P(m(X)=w)-P(X=m^{-1}(w))$$
where the second probability on the right hand side equals zero since $X$ has a continuous distribution function. 
 It follows that

  \begin{eqnarray}\label{Equality}
    \int_{ \RR}  \left(  P(m(X)  \le  w)  -  F_0 \circ m^{-1}(w) \right)^2  dw   & = &\int_{\RR}  P(m(X)  = w)^2  dw \notag \\
                                                                                 & =   &  0.  
  \end{eqnarray}
  To see why the preceding equality holds true, note that since the distribution function of $X$, $F_0$, is assumed to be continuous, then it follows that
  \begin{eqnarray*}
    \int_{\RR}   P(m(X)  = w)^2  dw   =  \int_{\mathcal{W}} P\left (X \in [a(w), b(w)) \right)^2  dw 
  \end{eqnarray*}
 where $\mathcal{W}$ is the set of point $w\in\RR$ such that there exist $x\neq x'$ that satisfy $m(x)=m(x')=w$, and for $w\in\mathcal W$, 
$ a(w)  < b(w)$ are such that $m$ takes the constant value $w$ on $[a(w), b(w))$, and $a(w)  = m^{-1}(w)$.  Using the well-known fact that a monotone function admits at most countably many constant parts, the set $\mathcal W$ is at most countable and therefore,
  \begin{eqnarray*}
    \int_{\RR}   P(m(X)  = w)^2  dw \le   \int_{\mathcal{W}}  dw = \lambda(\mathcal{W}) =0
  \end{eqnarray*}
  where $\lambda$ denotes the Lebesgue measure on $\RR$. Lemma \ref{lem: inv} follows from \eqref{Equality} since
$$P(m(X)  \le  w) =  \int\one_{\{m(x)\leq w\}}dF_0(x).$$
  \end{proof}

\subsection{Proofs for Section \ref{sec: estim}}

\noindent \textbf{Proof of Proposition \ref{prop: identifiability}.}   \ 
Let $\overline F(u)=P(X\geq u)$ for all $u\in\R$. It follows from \eqref{eq: inv}, that holds for all $w\in\R$ an $m\in\cal M$,  that
\begin{eqnarray}
  \label{eq:1}
  P(m_1(X)\geq t)
  &=&P(X\geq m_1^{-1}(t))\\
  &=&\overline F\circ m_1^{-1}(t)
\end{eqnarray}
for all $t\in\R$. Since $m_1(X)$ has the same distribution as $m_2(X)$, this implies that
\begin{equation}
  \label{eq:2'}
  \overline F\circ m_1^{-1}
  =\overline F\circ m_2^{-1}.
\end{equation}
It follows from the definition \eqref{eq: inverse} of the inverse of a function $m\in\cal M$, where we recall that the infimum of an empty set is defined to be one, that $m_j^{-1}(y)=0$ for all $y\leq m_j(0)$ and $m_j^{-1}(y)=1$ for all $y>m_j(1)$. Hence, we define the inverse of 
the non-increasing left-continuous function $\overline F\circ m_j^{-1}$ as
$$(\overline F\circ m_j^{-1})^{-1}(t)=\sup\{y\in[m_j(0);m_j(1)],\ \overline F\circ m_j^{-1}(y)\geq t\}$$
for all $t\in\R$, with the convention that the supremum of an empty set is equal to $m_j(0)$. 
Our aim is to derive from \eqref{eq:2'} that the inverses of $\overline F\circ m_1^{-1}
  $ and $\overline F\circ m_2^{-1}$ are equal. This is not an immediate consequence of the equality in  \eqref{eq:2'} since the definition of the inverse function of $\overline F\circ m_j^{-1}$ involves the function $m_j$ in addition to the function $\overline F\circ m_j^{-1}$. However, we show below that the dependence on $m_j$ can be removed by restricting attention to a restricted support.

We define the generalized inverse of $\overline F$ by 
$$\overline F^{-1}(t)=\sup\{u\in[0,1],\ \overline F(u)\geq t\}$$
for all $t\in\R$, with the convention that the supremum of an empty set is equal to zero. Similar to the proof of Lemma \ref{lem: inv}, it can be proved using that  $\overline F$ is non-increasing and left-continuous that the equivalence
\begin{eqnarray}\label{eq: invc}
\overline F(u)\geq t \Longleftrightarrow \overline F^{-1}(t)\geq u
\end{eqnarray}
holds for all $u\in(0,1]$ and $t\in\R$.  
Combining this  with 
\eqref{eq: switch} (that holds for all $x\in[0,1)$ and $y\in\RR$), we obtain that the equivalence 
\begin{eqnarray}\label{eq: equiv2}
\overline F\circ m_j^{-1}(y)\geq t \Longleftrightarrow m_j\circ \overline F^{-1}(t)\geq y
\end{eqnarray}
holds for all $t>0$ and $y>m_j(0)$. For $y\leq m_j(0)$, the inequalities on both sides of the equivalence in the previous display hold true for all $t\leq 1$ since in that case, $\overline F\circ m_j^{-1}(y)=\overline F(0)=1$. This means that the equivalence in \eqref{eq: equiv2} holds for all $t\in(0,1]$ and $y\in\R$. Now, consider $t\in(0,1]$ such that $t<\overline F\circ m_j^{-1}(m_j(0))$. Since  $m_j^{-1}(m_j(0))=0$, this means that $t<\overline F(0)$ where $\overline F(0)=1$. Otherwise said, we consider $t\in(0,1)$. Because $t<\overline F\circ m_j^{-1}(m_j(0))$, we  have
\begin{eqnarray*}
(\overline F\circ m_j^{-1})^{-1}(t)
&=&\sup\{y\leq m_j(1),\ \overline F\circ m_j^{-1}(y)\geq t\}.
\end{eqnarray*}
Moreover, the inequality $\overline F\circ m_j^{-1}(y)\geq t$ cannot hold for $y> m_j(1)$ since $t>0$ and therefore,
\begin{eqnarray*}
(\overline F\circ m_j^{-1})^{-1}(t)
&=&\sup\{y\in\R,\ \overline F\circ m_j^{-1}(y)\geq t\}.
\end{eqnarray*}
Combining this with \eqref{eq:2'} proves that 
$$(\overline F\circ m_1^{-1})^{-1}(t)=(\overline F\circ m_2^{-1})^{-1}(t)$$
for all $t\in(0,1)$. Using the equivalence in \eqref{eq: equiv2} (that holds for all $t\in(0,1]$ and $y\in\R$), we also have
\begin{eqnarray*}
(\overline F\circ m_j^{-1})^{-1}(t)
&=&\sup\{y\in\R,\ m_j\circ \overline F^{-1}(t)\geq y\}\\
&=&m_j\circ \overline F^{-1}(t).
\end{eqnarray*}
Hence,
$$m_1\circ \overline F^{-1}(t)=m_2\circ\overline F^{-1}(t)$$
for all $t\in(0,1)$.
This in turn implies that $m_1=m_2$ on the support of $X$ since the range of $\overline F^{-1}$ is the support of $X$.   \hfill $\Box$

\medskip
\medskip

\par \noindent \textbf{Proof of Proposition \ref{prop: exist}.}   \  
  {\it Proof of Claim 1.}
Recall that any element $m \in \mathcal{M}$ is bounded, and hence there exists $K > 0$ such that $\Vert m \Vert_\infty \le K$. Denote by $Y_{(1)}\leq\dots\leq Y_{(n)}$ the order statistics corresponding to $Y_1,\dots,Y_n$.  We have for all   $y < Y_{(1)}$  that  $\mathbb{H}_n(y)  = 0$. Moreover, it follows from monotonicity of $\Phi_\epsilon$ and $m$ that $0\leq \Phi_\epsilon(y  -  m(X_i))\leq \Phi_\epsilon(y  -  m(X_{(1)}))\leq 1$ for all $i\in\{1,\dots,n\}$ and therefore, 
\begin{eqnarray*}
\Big \{ \mathbb{H}_n(y)  -  n^{-1}  \sum_{i=1}^n \Phi_\epsilon(y  -  m(X_i)) \Big \}^2  & = & n^{-2}  \left(\sum_{i=1}^n \Phi_\epsilon(y  -  m(X_i)) \right)^2 \\
& \le &  \Phi_\epsilon(y  -  m(X_{(1)}))^2  \\
& \le &  \Phi_\epsilon(y  -  m(X_{(1)})).
\end{eqnarray*}
Now, existence of expectation of $\epsilon$  implies that  
\begin{eqnarray}\label{Integrable}
\int_{-\infty}^c \Phi_\epsilon(t) dt < \infty, \  \ \textrm{and} \  \  \int_{c}^\infty (1-\Phi_\epsilon(t)) dt < \infty
\end{eqnarray}
for arbitrary $c\in\RR$ and therefore,
\begin{eqnarray}\label{intY1} \notag
\int_{-\infty}^{Y_{(1)}}\Big \{ \mathbb{H}_n (y) -  n^{-1}  \sum_{i=1}^n \Phi_\epsilon(y  -  m(X_i)) \Big \}^2 dy&\leq&  \int_{-\infty}^{Y_{(1)}}\Phi_\epsilon(y  -  m(X_{(1)})) dy\\ \notag
&=&  \int_{-\infty}^{Y_{(1)}-m(X_{(1)})}\Phi_\epsilon(y ) dy\\
&<&\infty.
\end{eqnarray}
Similarly, $\mathbb{H}_n(y)  = 1$ for  $y > Y_{(n)}$  and hence 
\begin{eqnarray*}
\Big \{ \mathbb{H}_n(y)  -  n^{-1}  \sum_{i=1}^n \Phi_\epsilon(y  -  m(X_i)) \Big \}^2  & =  &   \Big(1-  n^{-1}\sum_{i=1}^n \Phi_\epsilon(y  -  m(X_{i})\Big)^2  \\
& \le & \Big(1-\Phi_\epsilon(y  -  m(X_{(n)})\Big)^2 \\
& \le & 1-\Phi_\epsilon(y  -  m(X_{(n)}).
\end{eqnarray*}
Combined with \eqref{Integrable}, this proves that 
\begin{eqnarray}\label{intYn}
\int^{\infty}_{Y_{(n)}}\Big \{ \mathbb{H}_n (y) -  n^{-1}  \sum_{i=1}^n \Phi_\epsilon(y  -  m(X_i)) \Big \}^2 dy&<&\infty.
\end{eqnarray}
Since the integrand is bounded, \eqref{intY1} and \eqref{intYn} yield
$$\int_\RR\Big \{ \mathbb{H}_n (y) -  n^{-1}  \sum_{i=1}^n \Phi_\epsilon(y  -  m(X_i)) \Big \}^2 dy<\infty, $$ 
which proves that $\mathbb{M}_n(m)$ is finite.\\

{\it Proof of Claim 2.}  Assume that Assumption $A1$ holds and consider $m \in \mathcal{M}$. Then, we have that
\begin{align}
  \mathbb{M}(m) 
  & =  \int_{\mathbb{R}}  \left \{ \int_{\R} \Big(\Phi_\epsilon(y  -  m_0(x)) - \Phi_\epsilon(y  -  m(x)) \Big) dF_0(x)  \right \}^2  dy  \notag  \\
  & =  \int_{[0, \infty)}  \left \{ \int_{\R} \Big (1-  \Phi_\epsilon(y  -  m(x))  - (1-\Phi_\epsilon(y  -  m_0(x)))  \Big) dF_0(x)  \right \}^2  dy
    \label{Integ1} \\
  & \quad +   \int_{(-\infty, 0]}  \left \{ \int_{\R} \Big (\Phi_\epsilon(y  -  m_0(x)) - \Phi_\epsilon(y  -  m(x)) \Big) dF_0(x)  \right \}^2  dy.  \label{Integ2}
\end{align}

We further bound above  the integral in (\ref{Integ1}) by
\begin{eqnarray}
&&  \le 2  \int_{[0, \infty)}   \left\{\int_{\R} \big(1-  \Phi_\epsilon(y  -  m(x))\big) dF_0(x) \right\}^2 dy  \notag \\
&&  +  \  2  \int_{[0, \infty)}  \left\{ \int_{\R} \big(1-\Phi_\epsilon(y  -  m_0(x))\big)   dF_0(x)  \right \}^2  dy  \notag  \\
&& \le   4  \int_{[0, \infty)}   \Big (1-  \Phi_\epsilon(y  -  \max(K_0, K))\Big)^2 dy      \label{IneqK} 
\end{eqnarray}
where we recall that $K \ge \Vert m \Vert_\infty$.  The latter integral is finite  using again (\ref{Integrable}).   Similarly, we argue that the integral in (\ref{Integ2}) can be also bounded above by 
\begin{eqnarray*}
4  \int_{(- \infty, 0]}  \Big ( \Phi_\epsilon(y  +  \max(K_0, K))\Big)^2 dy  < \infty.
\end{eqnarray*}
This completes the proof that $\mathbb{M}(m)$  is finite.  \\

{\it Proof of Claim 3.}  Using again that $\mathbb{H}_n(y)  = 0$ for all   $y < Y_{(1)}$, together with monotonicity of $\Phi_\epsilon$ and $m$,   we have that
\begin{eqnarray*}
\mathbb{M}_n(m)   & \ge   &  \int_{-\infty}^{Y_{(1)}} n^{-2}  \bigg(\sum_{i=1}^n \Phi_\epsilon(y- m(X_i))\bigg)^2 dy \\
& \ge &    n^{-2}  \int_{-\infty}^{Y_{(1)}} \Phi_\epsilon(y - m(X_{(1)}))^2 dy   \\
 & = &  n^{-2}    \int_{-\infty}^{Y_{(1)}  - m(X_{(1)})} \Phi_\epsilon(t)^2 dt  
\\
& \to &  \infty, \  \  \textrm{if $m(X_{(1)})  \to -\infty$}.
\end{eqnarray*}
Similarly, for $y \ge Y_{(n)}$, it holds that 

\begin{eqnarray*}
\mathbb{M}_n(m)  & \ge  &  n^{-2} \int_{Y_{(n)}}^\infty \big(1 - \Phi_\epsilon(y - m(X_{(n)}))\big)^2  dy  \\
& \ge &   n^{-2} \int_{Y_{(n)} - m(X_{(n)})}^\infty (1 - \Phi_\epsilon(t))^2  dt \\
& \to &  \infty, \  \  \textrm{if $m(X_{(n)})  \to \infty$}.
\end{eqnarray*}

Hence, there exists some $K > 0$ (which  may depend on $n$) such that any candidate $m \in \mathcal{M}$ for the minimization problem in (\ref{Minim2}) should satisfy $ - K \le m(X_{(1)}) \le \ldots  \le m(X_{(n)})  \le K$.  By identifying an element $m \in \mathcal{M}_K$ by the corresponding vector $\theta = (m(X_{(1)}), \ldots, m(X_{(n)}))^T$, it is easy to see that the original minimization problem is equivalent to minimizing 
\begin{eqnarray*}
\widetilde{\mathbb{M}}_n(\theta) =:   \int_{\RR}  \left \{ \mathbb{H}_n(y)  - n^{-1}  \sum_{i=1}^n \Phi_\epsilon(y- \theta_i) \right\}^2 dy
\end{eqnarray*}
on the compact finite dimensional subset  
\begin{eqnarray*}
\mathcal{S}_K =:  \left\{(\theta_1, \ldots, \theta_n)^T \in \RR^n:  -K \le \theta_1 \le \ldots \le \theta_n \le K \right \}.\end{eqnarray*}
Now, the function $ \widetilde{\mathbb{M}}_n$ is continuous on $\mathcal{S}_K$ since  for any sequence $(\theta_p)_{p \ge 0}$ in $\mathcal{S}^{\NN}_K$  converging (in any distance) to $\theta \in \mathcal{S}_K$, the sequence of functions
$$
y \mapsto \left(\mathbb{H}_n(y) - n^{-1}  \sum_{i=1}^n \Phi_\epsilon(y - \theta_{p,i}) \right)^2
$$ 
converges pointwise by continuity of $\Phi_\epsilon$ (see Assumption $A2$)  to the limit
\begin{equation}
  \label{eq:4}
  y \mapsto \left(\mathbb{H}_n(y) - n^{-1}  \sum_{i=1}^n \Phi_\epsilon(y - \theta_i) \right)^2.
\end{equation}
Also, for $y \in \RR$,  we have that  \eqref{eq:4} is no larger than

\begin{align*}
  \left \{
  \begin{array}{lll}
    \big(n^{-1} \sum_{i=1}^n \Phi_\epsilon(y+K) \big)^2, & \textrm{for $y < Y_{(1)}$} \\
    4,\   &  \textrm{for $Y_{(1)}  \leq y \le Y_{(n)}$}\\
    \bigg(1 - n^{-1}  \sum_{i=1}^n \Phi_\epsilon(y -K) \bigg)^2,  & \textrm{for $y > Y_{(n)}$} 
  \end{array}
  \right.
\end{align*}
where the function on the right side can be shown to be integrable using similar arguments as above. By the Lebesgue dominated convergence theorem, it follows that 
\begin{eqnarray*}
\lim_{p \to \infty}\widetilde{\mathbb{M}}_n(\theta_p) = \widetilde{\mathbb{M}}_n(\theta). 
\end{eqnarray*}
Thus,  $\widetilde{\mathbb{M}}_n $ admits at least a minimizer in $\mathcal{S}_K$, $\widehat{\theta}_n$ say. We conclude that   $\mathbb{M}_n$ admits at least a minimizer $\widehat m_n$ which is bounded by $K$, and such that $(\widehat m_n(X_{(1)}), \ldots, \widehat m_n(X_{(n)}))^T  =  \widehat{\theta}_n$. The values of the minimizer being given by $\hat\theta_n$ at the observed covariates $X_1,\dots,X_n$, any monotone interpolation of these values gives a solution to \eqref{Minim2}. In particular, there exists a solution $\widehat m_n$ that takes constant values between successive covariates and that is right continuous. 

{\it Proof of Claim 4.}    Without loss of generality, and possibly changing scale, we can assume that $\epsilon$ is supported on $[-1,1]$.  We show below that there exists at least a solution to \eqref{Minim2} that is bounded in sup-norm by $K_0+2$, where $K_0$ is taken from Assumption \ref{assm:A1:m0-sup}.   
  For an arbitrary $m\in\mathcal M$, we define the truncated version  $\bar m$  by 
\begin{equation*}
\bar m(x)=\begin{cases}
K_0+2&\mbox{ if }m(x)\geq K_0+2\\
-K_0-2&\mbox{ if }m(x)\leq -K_0-2\\
m(x)&\mbox{ otherwise}
\end{cases}.
\end{equation*}

In the following, we place ourselves in the event $\Vert \epsilon \Vert_\infty \le 1$ which occurs with probability 1.  Consider $y>K_0+1$. Since $|Y_i|\leq K_0+1$ for all $i$ , we then have $\mathbb{H}_n(y) =1$, and 
$$\Phi_\epsilon(y-m(X_i))=\Phi_\epsilon(y-\bar
 m(X_i))=1$$
 for all $X_i$'s such that $m(X_i)\leq -K_0-2$.
 Also, for all $X_i$'s such that $m(X_i)\geq K_0+2$ we have that 
$$ \Phi_\epsilon(y-m(X_i))\leq \Phi_\epsilon(y-\bar
 m(X_i))\leq 1. $$
This implies that
\begin{eqnarray}\label{eq:int+}\notag
&&\int_{K_0+1}^\infty  \Big \{ \mathbb{H}_n(y)  -  n^{-1}  \sum_{i=1}^n \Phi_\epsilon(y  -  \bar m(X_i)) \Big \}^2   dy \\
&&\qquad \leq \int_{K_0+1}^\infty  \Big \{ \mathbb{H}_n(y)  -  n^{-1}  \sum_{i=1}^n \Phi_\epsilon(y  -  m(X_i)) \Big \}^2   dy.
\end{eqnarray}
Similarly, it can be shown that
\begin{eqnarray}\label{eq:int-}\notag
&&\int_{-\infty}^{-K_0-1}  \Big \{ \mathbb{H}_n(y)  -  n^{-1}  \sum_{i=1}^n \Phi_\epsilon(y  -  \bar m(X_i)) \Big \}^2   dy \\
&&\qquad \leq \int_{-\infty}^{-K_0-1}  \Big \{ \mathbb{H}_n(y)  -  n^{-1}  \sum_{i=1}^n \Phi_\epsilon(y  -  m(X_i)) \Big \}^2   dy.
\end{eqnarray}
Now, consider $y$ such that $|y|\leq K_0+1$. If for some $i$ we have  $y>m(X_i)+1$ (or $y<m(X_i)-1$) then we have that
$$\Phi_\epsilon(y-m(X_i))=\Phi_\epsilon(y-\bar
 m(X_i)).$$
Indeed,  if $y > m(X_i) +1$, then  $m(X_i)  < K_0$. In case $m(X_i) > -K_0-2$ we have $m(X_i)  = \bar m(X_i)$ and  $\Phi_\epsilon(y-m(X_i))=\Phi_\epsilon(y-\bar  m(X_i)) =1$.   If $m(X_i) \le -K_0-2$, then  $\bar{m}(X_i) =  - K_0-2$ and hence $y - \bar{m}(X_i)  \ge 1$ implying again that $\Phi_\epsilon(y-m(X_i))=\Phi_\epsilon(y-\bar  m(X_i)) =1$. Similar arguments can be used in case $y < m(X_i) -1$.

Now, the equality in the above display holds also if $|y-m(X_i)|\leq 1$ since in that case,  $|m(X_i)|\leq K_0+2$, implying that  $\bar m
 (X_i)=m(X_i)$.   Combining this with \eqref{eq:int+} and \eqref{eq:int-} shows that
\begin{eqnarray*}
&&\int_{\RR}  \Big \{ \mathbb{H}_n(y)  -  n^{-1}  \sum_{i=1}^n \Phi_\epsilon(y  -  \bar m(X_i)) \Big \}^2   dy \\
&&\qquad \leq \int_{\RR}  \Big \{ \mathbb{H}_n(y)  -  n^{-1}  \sum_{i=1}^n \Phi_\epsilon(y  -  m(X_i)) \Big \}^2   dy.
\end{eqnarray*}
From Claim 3 in Proposition \ref{prop: exist}, there exists at least a solution to \eqref{Minim2} and from the arguments above, its truncated version also is a solution. Hence, there exists at least a solution that is bounded in the sup-norm by $K_0+2$ with probability $1$. This  completes the proof of the proposition.    \hfill  $\Box$

\subsection{Proofs for Section \ref{sec: rate}}\label{sec: proofs}
We first prove the propositions in Section \ref{sec: rate} and finish with the proof of Theorems \ref{theo: rate}, \ref{rem:1},  \ref{theo: nonoise}, \ref{theo: rateunequalsamplesizes} and \ref{theo: unifcons}.

\medskip
\medskip

\noindent \textbf{Proof of Proposition \ref{prop: DistanceG}.}    Let   $\widehat{\mathbb{H}}_n $ and $\mathbb{H}^0_n$ be the distribution functions defined as
\begin{eqnarray}\label{eq: Hn}
\widehat{\mathbb{H}}_n(y)  =  \frac{1}{n} \sum_{i=1}^n \Phi_\epsilon(y -  \widehat{m}_n(X_i)), \ \textrm{and} \ \ \mathbb{H}^0_n(y)  =  \frac{1}{n} \sum_{i=1}^n \Phi_\epsilon(y -  m_0(X_i)),
\end{eqnarray}
for $y \in \R$.   
Recall the Plancherel's identity  for Fourier transforms: for a function $g  \in L_1(\RR)\cap L_2(\RR) $, where $L_1(\RR)$ and $L_2(\RR)$ denote respectively the set of integrable, and the set of square integrable functions from $\RR$ to $\RR$ with respect to the Lebesgue measure it holds that
\begin{eqnarray*}
\int_{\RR}  g(x)^2  dx =  \frac{1}{2\pi}  \int_{\RR}  \vert \phi_g(x)  \vert^2 dx  
\end{eqnarray*}
where $\phi_g$ is defined in (\ref{phig}).   If $F_1 $ and   $F_2 $ are two distribution functions  with finite expectations, it follows using integration by parts that
\begin{eqnarray*}
\psi_{F_2}(x)  -  \psi_{F_1}(x)  =  -i x  \int_{\RR}   (F_2(t)  - F_1(t))  e^{itx}  dt   
\end{eqnarray*}
implying that 
\begin{eqnarray}\label{phipsi}
\phi_{F_2 - F_1}(x)  = i \  \frac{\psi_{F_2}(x)  -  \psi_{F_1}(x)}{x}
\end{eqnarray}
for $x \ne 0$.    Moreover, if $F_1$ and $F_2$ have finite expectations then
$$\int_{-\infty}^0F_j(x)dx<\infty\mbox{ and }\int_0^\infty (1-F_j(x))dx<\infty,$$
for $j\in\{1,2\}$, implying that $F_1-F_2\in L_1(\RR)\cap L_2(\RR)$. Therefore, the Plancherel identity implies that
\begin{eqnarray*}\label{PlanchDF}
\int_{\RR}  (F_2(x)  -  F_1(x))^2  dx =  \frac{1}{2\pi}  \int_{\RR} \frac{1}{x^2} \vert \psi_{F_2}(x)  - \psi_{F_1}(x)\vert^2 dx.  
\end{eqnarray*}
We apply below this identity with $F_1$ and $F_2$ replaced respectively by $\widehat{\mathbb{L}}_n$ and  $\mathbb L_n$, defined in (\ref{DefLnhatLn}).  Note that the two corresponding distributions have finite expectations since they are supported on a finite set. Hence,
\begin{eqnarray*}
\int_{\R}  \left(\widehat{\mathbb{L}}_n(w)   -  \mathbb L_n(w)  \right)^2  dw  &=&  \frac{1}{2\pi} \int_{\R} \frac{1}{t^2}\left\vert \psi_{\widehat{\mathbb{L}}_n}(t)  -  \psi_{\mathbb{L}_n}(t)  \right\vert^2  dt.\\
\end{eqnarray*}
By Assumption \ref{assm:A5}, we can find $T^* > 0$ such that  $\vert \phi_{f_\epsilon}(t) \vert   \ge\vert \phi_{f_\epsilon}(T) \vert>0$ for all $T>T^*$ and $\vert t\vert \le T$.  Using that $\vert \psi_{F} \vert \le 1$ for any distribution function $F$, it follows from the previous display that for all $T> T^*$ we have
\begin{eqnarray*}
\int_{\R}  \left(\widehat{\mathbb{L}}_n(w)   -  \mathbb L_n(w)  \right)^2  dw  
& \le  & \frac{1}{2 \pi \vert \phi_{f_\epsilon}(T) \vert^2} \int_{-T}^{T} \frac{\vert \phi_{f_\epsilon}(t) \vert^{2}}{t^2}  \left\vert \psi_{\widehat{\mathbb{L}}_n}(t)  -  \psi_{\mathbb{L}_n}(t)  \right\vert^2  dt       +  \frac{4}{\pi T}\\
& \le  & \frac{1}{2 \pi \vert \phi_{f_\epsilon}(T) \vert^2} \int_{\RR} \frac{\vert \phi_{f_\epsilon}(t) \vert^{2}}{t^2}  \left\vert \psi_{\widehat{\mathbb{L}}_n}(t)  -  \psi_{\mathbb{L}_n}(t)  \right\vert^2  dt       +  \frac{4}{\pi T}.
\end{eqnarray*}
Now, using again Plancherel's identity we have 
\begin{eqnarray*}
 \int_{\R} \frac{\vert \phi_{f_\epsilon}(t) \vert^2}{t^2}  \left\vert \psi_{\widehat{\mathbb{L}}_n}(t)  -  \psi_{\mathbb{L}_n}(t)  \right\vert^2  dt  & = & \int_{\R}  \vert \phi_{f_\epsilon}(t) \vert^2  \vert \phi_{\widehat{\mathbb{L}}_n - \mathbb{L}_n}(t) \vert^2 dt \\
& = & \int_{\R} \vert \phi_{f_\epsilon \star (\widehat{\mathbb{L}}_n - \mathbb{L}_n)}(t) \vert^2 dt\\
& = &   2\pi\int_{\R}  \left( \widehat{\mathbb{H}}_n(y)  -  \mathbb{H}^0_n(y) \right)^2  dy
\end{eqnarray*}
since $\widehat{\mathbb{H}}_n  =  f_\epsilon \star \widehat{\mathbb{L}}_n$ and $\mathbb{H}^0_n  = f_\epsilon \star  \mathbb{L}_n$.
Here, $(f \star g)(y) := \int_{\RR} f(z) g(y-z) dz.$
Hence, it follows from Assumption \ref{assm:A4:Phi-smoothness} that for sufficiently large $T$,
\begin{equation}
  \label{eq: Lnhat-Ln}
  \int_{\R}  \left(\widehat{\mathbb{L}}_n(w)   -  \mathbb L_n(w)  \right)^2  dw  
   \le   \frac{T^{2\beta}}{ d_0^2} \int_{\R}  \left( \widehat{\mathbb{H}}_n(y)  -  \mathbb{H}^0_n(y) \right)^2  dy      +  \frac{4}{\pi T}.   
\end{equation}
Assuming that we have  
\begin{eqnarray}\label{ErrorH}
\int_{\R}  E\left( \widehat{\mathbb{H}}_n(y)  -  \mathbb{H}^0_n(y) \right)^2  dy  = O(n^{-1}),
\end{eqnarray}
it will follow that for all sufficiently large $T$,
\begin{eqnarray*}
\int_{\R} E \left(\widehat{\mathbb{L}}_n(w)   -  \mathbb L_n(w)  \right)^2  dw \le   O(T^{2\beta}    n^{-1})+  \frac{4}{\pi T}.
\end{eqnarray*}
For  $T = T_n  \sim n^{1/(2\beta +1)}$ we get
\begin{eqnarray*}
\int_{\R}  E\left(\widehat{\mathbb{L}}_n(w)   -  \mathbb L_n(w)  \right)^2  dw \le   O\left( \frac{1}{n^{1/(2\beta+1)}}  \right),
\end{eqnarray*}
which proves Proposition \ref{prop: DistanceG}.

Now, we will show (\ref{ErrorH}).   From the inequality  $(a + b)^2 \le 2 (a^2 + b^2)$, which holds for any $a$ and $b$ in $\RR$, and the definition of $\widehat{m}_n$ it follows that 
\begin{eqnarray*}
 \int_{\R}  \left( \widehat{\mathbb{H}}_n(y)  -  \mathbb{H}^0_n(y) \right)^2  dy &\le  & 2 \int_{\R}  \left( \widehat{\mathbb{H}}_n(y)  -  \mathbb{H}_n(y) \right)^2   dy +  2 \int_{\R}  \left(\mathbb{H}^0_n(y)  -  \mathbb{H}_n(y)  \right)^2   dy \\
& \le & 4 \int_{\R}  \left( \mathbb{H}^0_n(y)  -  \mathbb{H}_n(y) \right)^2  dy \\
& \le &  8   \int_{\R}  \left(\mathbb{H}^0_n(y)  - H_0(y) \right)^2 dy +  8    \int_{\R}\left( \mathbb{H}_n(y)- H_0(y) \right)^2  dy 
\end{eqnarray*}
where
\begin{eqnarray*}
E[\left( \mathbb{H}_n(y)- H_0(y) \right)^2]  = n^{-1}  H_0(y)  (1-H_0(y)),
\end{eqnarray*}
and
\begin{eqnarray*}
  E[\left(\mathbb{H}^0_n(y)  - H_0(y) \right)^2]  & =  &   \frac{1}{n} \Var \Phi_\epsilon( y - m_0(X))\\
  &=  &  \inv{n} \int_{\RR} \lp \Phi_\epsilon(y - m_0(x)) - H_0(y) \rp^2 dF_0(x).
\end{eqnarray*}
Both the integrals
\begin{eqnarray}\label{I1}
I_1=  \int_{\R}  H_0(y)  (1-H_0(y)) dy
\end{eqnarray}
and 
\begin{eqnarray}\label{I2}
I_2=  \int_{\R}  \int_{\RR} \lp \Phi_\epsilon(y - m_0(x)) - H_0(y) \rp^2 dF_0(x)   dy
\end{eqnarray}
 are finite,  see Appendix \ref{appendix I1}.  This yields  the result.     \hfill $\Box$

\medskip
\medskip

\noindent \textbf{Proof of Proposition \ref{prop: E1andE3andE2}.}  \  In the sequel we denote by $\mathbb P^X_n$  and $P^X$  the empirical probability measure associated with $X_1, \ldots, X_n$ and the true corresponding probability measure. Then, the two integrals in Proposition \ref{prop: E1andE3andE2} are the integrated square of the empirical processes  
\begin{eqnarray*}
\widehat{\mathbb{L}}_n(w)    -\widehat{L}^0_n(w)  = (\mathbb P^X_n - P^X)  \one_{\{ \widehat m_n(\cdot) \le w \}}
\end{eqnarray*}
and 
\begin{eqnarray*}
\mathbb{L}_n(w) - L_0(w)  = (\mathbb P^X_n - P^X)   \one_{\{ m_0(\cdot) \le w\}}.
\end{eqnarray*}
In Appendix~\ref{sec:some-rudiments-EPT}, we recall some of the basic tools of empirical processes that we need in this proof. In what follows, the notation $\lesssim$ means smaller or equal modulo a universal positive multiplicative constant.  For all fixed $w\in\RR$ and $m \in \mathcal{M}$, let $k_{w,m}$ be the function defined by $k_{w,m}(x)=\one_{m(x) \le w}$ for all $x\in[0,1]$. Consider the set of functions
\begin{eqnarray*}
\mathcal{I} : =  \Big \{  k_{w,m},  \ \textrm{with}  \   m \in \mathcal{M} \ \textrm{and }  w \in [A,B]  \Big \}.   
\end{eqnarray*}
Using the same notation as in \cite{aadbook}  (for completeness, we provide definitions in Appendix~\ref{sec:some-rudiments-EPT}), let us write  $ \mathbb{G}_n k = \sqrt n (\mathbb P^X_n - P^X)  k$ for $k \in \mathcal{I}$.  Since $\mathcal{I}$  is a subset of the class of monotone non-increasing functions $f: \RR \mapsto [0, 1]$,   it follows from \citet[Theorem 2.7.5]{aadbook} that there exists a universal constant  $C > 0$, such that  for any $\delta > 0$ and any probability measure $Q$,
\begin{eqnarray*}
\log N_B\Big(\delta, \mathcal{I}, L_2(Q)\Big) \le \frac{C}{\delta}
\end{eqnarray*}
(where $N_B(\cdot, \cdot, \cdot)$ is defined in Appendix~\ref{sec:some-rudiments-EPT}).
Since $\mathcal{I}$ admits $F(t) =1$ as an envelope, this and the inequality in (\ref{RelNumb}) imply that
\begin{eqnarray*}
J(1, \mathcal{I}) & \le  &  \sup_{Q} \int_0^1 \sqrt{1 + \log N_B(2\delta, \mathcal{I}, L_2(Q))} d\delta  \\
                          & \le &  \int_0^1 \sqrt{1 +  \frac{C}{2 \delta}} d\delta \le 1 + \sqrt{2C}  < \infty,
  \end{eqnarray*} 
where $J(\delta, \mathcal{F})$ is defined in  (\ref{Jdef}). Since $X_1,\dots,X_n$ are i.i.d. it follows now from \citet[Theorem 2.14.1]{aadbook}  that 
\begin{eqnarray}\label{BoundGn}
\left(E\left[  \Vert \mathbb G_n  \Vert^2_{\mathcal{I}}\right]\right)^{1/2}  \lesssim J(1, \mathcal{I}). 
\end{eqnarray}
Let us denote  
\begin{eqnarray*}
M_n &=  \max \left(\int_{A}^B  \left(\widehat{\mathbb{L}}_n(w)    -\widehat{L}^0_n(w)   \right)^2  dw,  \int_{A}^B  \Big(\mathbb{L}_n(w) - L_0(w)\Big)^2  dw  \right).
\end{eqnarray*}
Then,
\begin{eqnarray*}
0 \le M_n \le  \frac{B-A}{n}   \Vert \mathbb{G}_n \Vert_{\mathcal{I}}^2.  
\end{eqnarray*}
The first two claims in the proposition now follow from (\ref{BoundGn}) combined to   the Markov's inequality. 

Now,
using the inequality $(a + b +c)^2 \le 3 (a^2 + b^2  + c^2)$ for any real numbers $a, b $ and $c$, we have
\begin{eqnarray}\label{Decomp}
\int_{A}^B  \left(\widehat{L}^0_n(w)   - L_0(w)  \right)^2  dw  & \le   & 3  \int_{A}^B  \left(\widehat{\mathbb{L}}_n(w)    -\widehat{L}^0_n(w)   \right)^2  dw   \notag \\
&&  \ + \  3 \int_{A}^B  \left(\widehat{\mathbb{L}}_n(w) - \mathbb{L}_n(w) \right)^2  dw \notag \\
&&  \ + \  3 \int_{A}^B  \left(\mathbb{L}_n(w) - L_0(w)\right)^2  dw. \notag
\end{eqnarray}
It thus follows from Proposition \ref{prop: DistanceG} that \eqref{E2} holds
provided that $B-A=O_P(n^{2\beta/(2\beta+1)}).$ \hfill $\Box$

\medskip
\medskip

\par \noindent \textbf{Proof of Proposition \ref{prop: DistanceG2}.}  \ The first equality in Proposition \ref{prop: DistanceG2} follows from Lemma \ref{lem: inv}  above combined with the definition of $\widehat L^0_n$ and $L^0_n$, while the second equality follows from  Proposition \ref{prop: E1andE3andE2}.   \hfill $\Box$

\medskip
\medskip

\par \noindent \textbf{Proof of Proposition \ref{prop: Conversion}.}  \ It follows from Lemma \ref{L1inv} below  that
\begin{eqnarray*}
\int_a^b \left \vert \widehat m_n(x)  -  m_0(x) \right \vert dF_0(x)  \le 
\int_{A_n}^{B_n}\vert F_0\circ \widehat m_n^{-1}(x)-F_0\circ m_0^{-1}(x)\vert dx.
\end{eqnarray*}
The proposition then follows from applying the Cauchy-Schwarz inequality.   \hfill $\Box$

\medskip
\medskip

\begin{lemma}\label{L1inv} 
  Let $f:[0,1]\to\RR$ and $g:[0,1]\to\RR$ be right-continuous non-decreasing functions. Let $f^{-1}$ and $g^{-1}$ be the corresponding generalized inverses, see \eqref{eq: inverse}   where the infimum of an empty set is defined to be one.
  Let $H:[0,1]\to[0,1]$ be a continuous non-decreasing function. Then, for all $a<b$ in $[0,1]$ we have
  $$\int_a^b\vert f(t)-g(t)\vert dH(t)\leq \int_{I(a)}^{S(b)}\vert H\circ g^{-1}(x)-H\circ f^{-1}(x)\vert dx $$
  where
  $$I(a)=f(a)\wedge g(a)\ ;\ S(b)=f(b)\vee g(b).$$
\end{lemma} 

\medskip

\begin{proof}[Proof of Lemma \ref{L1inv}.]

  For all real numbers $u$, let $u_+=\max(u,0)$. We then have
  \begin{eqnarray}\label{L1invI1I2}
    \int_a^b\vert f(t)-g(t)\vert dH(t) = I_1+I_2
  \end{eqnarray}
  where
  $$I_1=\int_a^b( f(t)-g(t))_+ dH(t) \mbox{ and } I_2=\int_a^b(g(t)-f(t))_+ dH(t) .$$
  Let us deal first with $I_1$. We have
  \begin{align*}
    I_1 = \int_a^b\int_0^\infty\mathbb{I}_{\{x\le  f(t)-g(t)\}}dxdH(t)
    & = \int_a^b\int_{g(t)}^\infty\mathbb{I}_{\{x\le f(t)\}}dxdH(t) \\
    & = \int_a^b\int_{g(t)}^{f(b)}\mathbb{I}_{\{x\le f(t)\}}dxdH(t),
  \end{align*}
  where we use a change of variable for the second equality and the monotonicity of $f$  for the third one. Similar to \eqref{eq: switch}, the equivalence
  $$t\ge f^{-1}(x)\Longleftrightarrow f(t)\ge x$$
  holds for all $t\in[0,1)$ and $x\in\RR$. Combining this with the Fubini theorem, we arrive at
  \begin{eqnarray*}
    I_1
    &=& \int_a^b\int_{g(t)}^{f(b)}\mathbb{I}_{\{t\geq f^{-1}(x)\}}dxdH(t)\\
    &=& \int_{g(a)}^{f(b)}\int_a^b\mathbb{I}_{\{t\ge f^{-1}(x)\}}\mathbb{I}_{\{t< g^{-1}(x)\}}dH(t)dx.
  \end{eqnarray*}
Hence, it follows from the continuity of $H$ that
\begin{eqnarray*}
I_1
&=& \int_{g(a)}^{f(b)}\left(H(g^{-1}(x)\wedge b)-H( f^{-1}(x)\vee a)\right)_+dx\\
&\le& \int_{g(a)}^{f(b)}\left(H(g^{-1}(x))-H( f^{-1}(x))\right)_+dx,
\end{eqnarray*}
since $H$ is non-decreasing. Since $I(a)\le g(a)$ and $S(b)\ge f(b)$, this implies that
\begin{eqnarray*}
I_1
&\le& \int_{I(a)}^{S(b)}\left(H\circ g^{-1}(x)-H\circ f^{-1}(x)\right)_+dx.
\end{eqnarray*} 
Interchanging the roles of $f$ and $g$, we obtain
\begin{eqnarray*}
I_2
&\leq& \int_{I(a)}^{S(b)}\left(H\circ f^{-1}(x)-H\circ g^{-1}(x)\right)_+dx
\end{eqnarray*}
and therefore,
\begin{eqnarray*}
I_1+I_2
&\leq& \int_{I(a)}^{S(b)}\left\vert H\circ f^{-1}(x)-H\circ g^{-1}(x)\right\vert dx.
\end{eqnarray*}
Lemma \ref{L1inv} then follows from \eqref{L1invI1I2}.
\end{proof}

\medskip
\medskip

\par \noindent \textbf{Proof of Theorem \ref{theo: rate}.}   \textit{ Proof of the first claim.} If we have \eqref{eq: AnBn}, then $A_n$ and $B_n$ from Proposition \ref{prop: Conversion} are both of the order $O_P(1)$. Hence, the first claim in Theorem \ref{theo: rate} is an immediate consequence of Proposition \ref{prop: Conversion} combined with Proposition \ref{prop: DistanceG2}. 

{\it Proof of the second claim.} Now, we show that  \eqref{eq: AnBn} holds  true for all $a$ and $b$ such that $0<F_0(a)\leq F_0(b)<1$.
 It follows from the definition of $\widehat{\mathbb{L}}_n$ and  $\mathbb L_n$  together with the H\"older inequality and Proposition \ref{prop: DistanceG} that
\begin{eqnarray*}
\int_{\Vert m_0\Vert_\infty}^{2\Vert m_0\Vert_\infty}  \left(1- \frac{1}{n} \sum_{i=1}^n
   \one_{ \{\widehat m_n(X_i)  \le w \}}\right)dw
&=&\int_{\Vert m_0\Vert_\infty}^{2\Vert m_0\Vert_\infty}  \left\vert\widehat{\mathbb{L}}_n(w)   -  \mathbb L_n(w)  \right\vert  dw\\
&\leq&\left(\Vert m_0\Vert_\infty \int_{\R}  \left(\widehat{\mathbb{L}}_n(w)   -  \mathbb L_n(w)  \right)^2  dw\right)^{1/2}\\
&=& o_P(1).
\end{eqnarray*}
On the other hand,
\begin{eqnarray*}
\int_{\Vert m_0\Vert_\infty}^{2\Vert m_0\Vert_\infty}  \left(1- \frac{1}{n} \sum_{i=1}^n
   \one_{ \{\widehat m_n(X_i)  \le w \}}\right)dw
&\geq&\Vert m_0\Vert_\infty n^{-1} \sum_{i=1}^n
   \one_{ \{\widehat m_n(X_i)  >2 \Vert m_0\Vert_\infty \}}
\end{eqnarray*}
and therefore,
\begin{eqnarray*}
 \sum_{i=1}^n
   \one_{ \{\widehat m_n(X_i)  >2 \Vert m_0\Vert_\infty \}}=o_P(n).
\end{eqnarray*}
By monotonicity of $ \widehat m_n$ this implies that
\begin{eqnarray*}
 \one_{\{ \widehat m_n(b)>2\Vert m_0\Vert_\infty \}}\sum_{i=1}^n
   \one_{ \{X_i> b \}}
&\leq& \sum_{i=1}^n
   \one_{ \{\widehat m_n(X_i)  >2 \Vert m_0\Vert_\infty \}}=o_P(n).
\end{eqnarray*}
By the law of large numbers, $n^{-1}\sum_{i=1}^n
   \one_{ \{X_i> b \}}$ converges in probablity to $1-F_0(b)>0$ and therefore, it follows from the previous display that
$$ \one_{\{ \widehat m_n(b)>2\Vert m_0\Vert_\infty \}}=o_P(1).$$
This implies that
$$\lim_{n\to\infty}P( \widehat m_n(b)>2\Vert m_0\Vert_\infty )=0.$$
One can prove similarly that  
$$\lim_{n\to\infty}P( \widehat m_n(a)<-2\Vert m_0\Vert_\infty )=0.$$
This implies  \eqref{eq: AnBn} by monotonicity of $\widehat m_n$,  which completes the proof of the second claim in Theorem \ref{theo: rate}.  \hfill $\Box$

\medskip
\medskip

\par \noindent \textbf{Proof of Theorem~\ref{rem:1}}   \   As the arguments are very similar to those used in the proof of Theorem~\ref{theo: rate},  we focus here on how the converge rate is obtained in the supersmooth case.  Under Assumptions  \ref{assm:A0:basic:m0-Xi}--\ref{assm:A3:f0-bounded},  Assumption~\ref{assm:A4'} and Assumption  \ref{assm:A5}, we can show that this rate of convergence is driven by  
  \begin{eqnarray*}
    O\left(\exp\left( \frac{2 T^{\beta}}{\gamma}\right) T^{-2\alpha}  n^{-1}\right)
    + \frac{4}{\pi T}
  \end{eqnarray*}
  for $T = T_n \to \infty$ as $n \to \infty$ which should be determined so that the above expression is smallest. This means that the first term should converge to $0$ or equivalently that  there exists a sequence $(K_n)_n$ such that $\log (T_n) = K_n \to \infty$ and $K_n \le \log n$  such that
  \begin{eqnarray*}
    \frac{2 T^{\beta}_n}{\gamma}    = \log n  + ( 2 \alpha -1) K_n
  \end{eqnarray*}
  or equivalently
  \begin{eqnarray*}
    T_n  = c \Big(\log n  + (2 \alpha -1)  K_n\Big)^{1/\beta}, \  \ \textrm{for $c =  (\gamma/2)^{1/\beta}$}.
  \end{eqnarray*}
  It is not difficult to see that the optimal choice of the sequence $(K_n)_n$ is $K_n =  (1-a) \log n $ for some $a \in (0,1)$  (the case $a=0$ is impossible because otherwise we would have $T_n = c (2\alpha)^{1/\beta} (\log(T_n))^{1/\beta}$).  This in turn yields
  \begin{eqnarray*}
    T_n = c \Big(a \log n + 2\alpha \log(T_n) \Big)^{1/\beta},
  \end{eqnarray*}  
  implying that  $T_n \sim c a^{1/\beta} (\log n)^{1/\beta}$ and that rate of convergence is  $(\log n)^{-1/(2\beta)}.$  \hfill $\Box$

\medskip
\medskip

\par \noindent \textbf{Proof of Theorem \ref{theo: nonoise}.}  \ 
%
%
%
Without loss of generality, we can assume that the support points of $\epsilon$ are all in $[-1,1]$.  From Proposition \ref{prop: exist}  we know that with probability $1$ there exists a solution to \eqref{Minim2} which is bounded in the sup-norm by $K_0+2$. In the sequel, we denote by $\widehat m_n$ such a solution. 

Recall $\widehat{\mathbb{H}}_n$ and $ \mathbb{H}^0_n$ from \eqref{eq: Hn}. Using the Cauchy-Schwarz inequality it follows that
\begin{eqnarray*}
\int_{\R}  \left| \widehat{\mathbb{H}}_n(y)  -  \mathbb{H}^0_n(y) \right|  dy  &=&
\int_{-K_0-3}^{K_0+3} \left| \widehat{\mathbb{H}}_n(y)  -  \mathbb{H}^0_n(y) \right|  dy\\
&\leq&\left((2K_0+6)\int_{\R}  \left| \widehat{\mathbb{H}}_n(y)  -  \mathbb{H}^0_n(y) \right|^2  dy\right)^{1/2}.
\end{eqnarray*}
Since the equality in \eqref{ErrorH} holds under Assumptions \ref{assm:A0:basic:m0-Xi} to 
\ref{assm:A2:Phi-cts}  it follows that
\begin{eqnarray*}
\int_{\R}  \left| \widehat{\mathbb{H}}_n(y)  -  \mathbb{H}^0_n(y) \right|  dy  &=&O_P(n^{-1/2}).
\end{eqnarray*}
Now note that $\widehat{\mathbb{H}}_n$ and  $\mathbb{H}^0_n$ are distribution functions with bounded support, and hence they admit a finite first moment. Therefore, denoting by $W_1(F,G)$ the Wasserstein-distance of first order between two probability
distributions with respective distribution functions $F$ and $G$, it follows from 
Proposition~\ref{prop:1} in Appendix~\ref{sec:wass-dist-lemm} 
that  
\begin{eqnarray*}
W_1\left( \widehat{\mathbb{H}}_n,\mathbb{H}^0_n \right)  &=&O_P(n^{-1/2}).
\end{eqnarray*}
With $ \widehat{\mathbb{L}}_n$ and $\mathbb{L}_n$ taken from \eqref{DefLnhatLn}, it follows from
Proposition~\ref{prop:2} in Appendix~\ref{sec:wass-dist-lemm} 
that there exists some constant $C>0$ that depends only on the distribution of $\epsilon$ and $K_0$  such that
\begin{eqnarray*}
W_1\left( \widehat{\mathbb{L}}_n,\mathbb{L}_n \right) \leq C \ W_1\left( \widehat{\mathbb{H}}_n,\mathbb{H}^0_n \right).
\end{eqnarray*}
Thus, 
\begin{eqnarray*}
W_1\left( \widehat{\mathbb{L}}_n,\mathbb{L}_n \right) &=&O_P(n^{-1/2}).
\end{eqnarray*}
Using again %
Proposition~\ref{prop:1} in Appendix~\ref{sec:wass-dist-lemm} 
the latter rate yields
\begin{eqnarray*}
\int_{\R}  \left| \widehat{\mathbb{L}}_n(y)  -  \mathbb{L}_n(y) \right|  dy  &=&O_P(n^{-1/2}).
\end{eqnarray*}
Combining this with Proposition \ref{prop: E1andE3andE2} and the Cauchy-Schwarz inequality, we obtain 
\begin{eqnarray*}
\int_{\R}  \left| \widehat{L}^0_n(y)  - L_0(y) \right|  dy  &=&\int_{-K_0-2}^{K_0+2}  \left| \widehat{L}^0_n(y)  - L_0(y) \right|  dy \\
&=&O_P(n^{-1/2}).
\end{eqnarray*}
Using the result of Lemma \ref{L1inv}, it follows that
\begin{eqnarray*}
\int_0^1|\widehat m_n(x)-m_0(x)|dF_0(x) & \le  &   \int_{m_0(0) \wedge \widehat m_n(0)}^{m_0(1) \vee \widehat m_n(1)} |F_0\circ \widehat m_n^{-1}(y)-F_0\circ m_0^{-1}(y)|dy \\
& = & \int_{m_0(0) \wedge \widehat m_n(0)}^{m_0(1) \vee \widehat m_n(1)}   \left| \widehat{L}^0_n(y)  - L_0(y) \right|  dy  \\
&\le  &   \int_{-K_0-2}^{K_0+2} \left| \widehat{L}^0_n(y)  - L_0(y) \right|  dy 
\end{eqnarray*}
implying that$\int_0^1|\widehat m_n(x)-m_0(x)|dF_0(x)  = O_P(n^{-1/2})$.   \hfill $\Box$

\medskip
\medskip

\medskip
\par \noindent \textbf{Proof of Theorem \ref{theo: rateunequalsamplesizes}.}  \   
{\it Proof of Claim 1.}
We start with the case where the noise has an absolutely continuous density.  We restrict attention to the ordinary smooth case because the arguments are very similar in the supersmooth one.  Now, since the proof in the ordinary smooth case follows the same lines of the proof of Theorem \ref{theo: rate},  details are omitted and the reader is referred to the latter proof.  Here, we only point out the main existing differences.    Similar to \eqref{DefLnhatLn} and \eqref{eq: Hn},we define
$$
  \widehat{\mathbb{L}}_{\nx, \ny}(w) : =  \frac{1}{\nx} \sum_{i=1}^{\nx}
   \one_{ \{\widehat m_{\nx, \ny}(X_i)  \le w \}}
  \ \textrm{and} \ \
  \mathbb L_{\nx}(w)  :=   \frac{1}{\nx} \sum_{i=1}^{\nx}
  \one_{\{m_0(X_i)  \le w\}}  
$$
for all $w \in \R$, and
$$\widehat{\mathbb{H}}_{\nx, \ny}(y)  =  \frac{1}{\nx} \sum_{i=1}^{\nx} \Phi_\epsilon(y -  \widehat{m}_{\nx, \ny}(X_i)), \ \textrm{and} \ \ \mathbb{H}^0_{\nx}(y)  =  \frac{1}{\nx} \sum_{i=1}^{\nx} \Phi_\epsilon(y -  m_0(X_i))$$
for all $y\in\RR$.
With similar arguments as for the proof of \eqref{eq: Lnhat-Ln} we obtain that for all sufficiently large $T$,
\begin{equation}\label{eq: Lnhat-Lnx}
\int_{\R}  \left(\widehat{\mathbb{L}}_{\nx, \ny}(w)   -  \mathbb L_{\nx}(w)  \right)^2  dw   \le   \frac{T^{2\beta}}{ d_0^2} \int_{\R}  \left( \widehat{\mathbb{H}}_{\nx, \ny}(y)  -  \mathbb{H}^0_{\nx}(y) \right)^2  dy      +  \frac{4}{\pi T}.
\end{equation}
Moreover, it follows from the definition of $\hat m_{n_x, n_y}$ that
\begin{eqnarray*}
  \int_{\R}  \left( \widehat{\mathbb{H}}_{\nx, \ny}(y)  -  \mathbb{H}^0_{\nx}(y) \right)^2  dy
  &\le  & 2 \int_{\R}  \left( \widehat{\mathbb{H}}_{\nx, \ny}(y)  -  \mathbb{H}_{\ny}(y) \right)^2   dy  \\
  & & + \,2 \int_{\R}  \left(\mathbb{H}^0_{\nx}(y)  -  \mathbb{H}_{\ny}(y)  \right)^2   dy 
\end{eqnarray*}
which is less than or equal to
\begin{align*}
  & 4 \int_{\R}  \left( \mathbb{H}^0_{\nx}(y)  -  \mathbb{H}_{\ny}(y) \right)^2  dy
    \le  8   \int_{\R}  \left(\mathbb{H}^0_{\nx}(y)  - H_0(y) \right)^2 dy +  8    \int_{\R}\left( \mathbb{H}_{\ny}(y)- H_0(y) \right)^2  dy 
\end{align*}
where
\begin{eqnarray*}
E[\left( \mathbb{H}_{\ny}(y)- H_0(y) \right)^2]  = \ny^{-1}  H_0(y)  (1-H_0(y)),
\end{eqnarray*}
and
\begin{eqnarray*}
  E[\left(\mathbb{H}^0_{\nx}(y)  - H_0(y) \right)^2]  & =  &   \frac{1}{\nx} \Var \Phi_\epsilon( y - m_0(X))\\
  &=  &  \inv{\nx} \int_{\RR} \lp \Phi_\epsilon(y - m_0(x)) - H_0(y) \rp^2 dF_0(x).
\end{eqnarray*}
The integrals $I_1$ and $I_2$ defined in \eqref{I1} and \eqref{I2}
 are finite, since it can be shown that $ \int_{-\infty}^0 H_0(y) dy < \infty$ and $\int_0^\infty (1-H_0(y)) dy < \infty$;  see Appendix \ref{appendix I1}. Hence, we obtain that
\begin{eqnarray}\label{rateHnxny}
\int_{\R}  E\left( \widehat{\mathbb{H}}_{\nx, \ny}(y)  -  \mathbb{H}^0_{\nx}(y) \right)^2  dy  = O((\nx\wedge\ny)^{-1}).
\end{eqnarray}
Combining this with \eqref{eq: Lnhat-Lnx} proves that for all sufficiently large $T$,
$$\int_{\R}  E\left(\widehat{\mathbb{L}}_{\nx, \ny}(w)   -  \mathbb L_{\nx}(w)  \right)^2  dw   \le   T^{2\beta}O((\nx\wedge\ny)^{-1})     +  \frac{4}{\pi T}.$$
For  $T  \sim (\nx\wedge\ny)^{1/(2\beta +1)}$ we get
\begin{eqnarray}\label{eq: prop22unequal}
\int_{\R}  E\left(\widehat{\mathbb{L}}_{\nx, \ny}(w)   -  \mathbb L_{\nx}(w)  \right)^2  dw \le   O\left( \frac{1}{(\nx\wedge\ny)^{1/(2\beta+1)}}  \right),
\end{eqnarray}
which proves an analogue of Proposition \ref{prop: DistanceG} in the case of possibly unequal sample sizes.

Next, we consider an analogue of Proposition \ref{prop: E1andE3andE2}. For this task, we denote by $\mathbb P^X_{\nx}$  and $P^X$  the empirical probability measure associated with $X_1, \ldots, X_{\nx}$ and the true corresponding probability measure. We consider the empirical processes  
\begin{eqnarray*}
\widehat{\mathbb{L}}_{\nx, \ny}(w)    -\widehat{L}^0_{\nx, \ny}(w)  = (\mathbb P^X_{\nx} - P^X)  \one_{\{ \widehat m_{\nx, \ny}(\cdot) \le w \}}
\end{eqnarray*}
and 
\begin{eqnarray*}
\mathbb{L}_{\nx}(w) - L_0(w)  = (\mathbb P^X_{\nx} - P^X)   \one_{\{ m_0(\cdot) \le w\}};
\end{eqnarray*}  
where 
$$\widehat L^0_{n_x,n_y}(w)  = \int \one_{\{\widehat m_{n_x,n_y}(x) \le w\}}dF_0(x) $$
and $L_0$ is defined in \eqref{eq: Ln}. Then, with similar arguments as in the proof of Proposition \ref{prop: E1andE3andE2}, we obtain that for all random variables $A < B $ (that may depend on $n$) it holds that 
\begin{eqnarray*}\label{E1}
\int_{A}^B  \left(\widehat{\mathbb{L}}_{\nx, \ny}(w)   - \widehat{L}^0_{\nx, \ny}(w)   \right)^2  dw   \leq (B-A)  O_P( 1/{\nx})\leq (B-A) O_P(1/ {(\nx\wedge\ny)}),
\end{eqnarray*}
\begin{eqnarray*}\label{E3}
 \int_{A}^B  \Big(\mathbb{L}_{\nx}(w) - L_0(w)\Big)^2  dw  \leq (B-A) O_P(1/ {\nx}) \leq (B-A) O_P(1/ {(\nx\wedge\ny)}),
\end{eqnarray*}
where $O_P( 1/{\nx})$ is uniform in $A$ and $B$. Moreover, if $B-A=O_P((\nx\wedge\ny)^{2\beta/(2\beta+1)})$, then
\begin{eqnarray*}\label{E2uneq}
\int_{A}^B  \left(\widehat L^0_{\nx, \ny}(w)  -  L_0(w)  \right)^2  dw =  O_P((\nx\wedge\ny)^{-1/(2\beta+1)}).
\end{eqnarray*}

Next, similar to Proposition \ref{prop: DistanceG2} we obtain that
for all random variables $A < B $ such that $B-A=O_P(1)$ it holds that
\begin{eqnarray*}
\int_{A}^B  \left(  F_0 \circ \widehat m^{-1}_{\nx, \ny}(w)  -   F_0 \circ m^{-1}_0(w)   \right)^2  dw & =&  
\int_{A}^B   \left(\widehat{L}^0_n(w)  -   L_0(w) \right)^2  dw\\
&= &  O_P((\nx\wedge\ny)^{-1/(2\beta +1)} ).
\end{eqnarray*}

Proposition \ref{prop: Conversion} still holds in the case of possibly different sample sizes with $\wh m_n$ replaced by $\hat m_{n_x,n_y}$ If \eqref{eq: AnBn} also holds, then $A_n$ and $B_n$ from Proposition \ref{prop: Conversion} are both of the order $O_P(1)$. In that case, the second assertion in Theorem \ref{theo: rateunequalsamplesizes} is an immediate consequence of Proposition \ref{prop: Conversion} combined with the preceding display. Hence, it remains to prove that \eqref{eq: AnBn} holds.
It follows from the definition of $\widehat{\mathbb{L}}_{\nx, \ny}$ and  $\mathbb L_{\nx}$  together with the Cauchy-Schwarz inequality and \eqref{eq: prop22unequal} that
\begin{eqnarray*}
\int_{\Vert m_0\Vert_\infty}^{2\Vert m_0\Vert_\infty}  \left(1- \frac{1}{\nx} \sum_{i=1}^{\nx}
   \one_{ \{\widehat m_{\nx, \ny}(X_i)  \le w \}}\right)dw
  &=&\int_{\Vert m_0\Vert_\infty}^{2\Vert m_0\Vert_\infty}  \left\vert \widehat{\mathbb{L}}_{\nx, \ny}(w)   -  \mathbb L_{\nx}(w)  \right\vert  dw\\
\end{eqnarray*}
which is bounded above by
\begin{equation*}
  \sqrt{\left(\Vert m_0\Vert_\infty \int_{\R}  \left(\widehat{\mathbb{L}}_{\nx, \ny}(w)   -  \mathbb L_{\nx}(w)  \right)^2  dw\right)}
  =  o_P(1).
\end{equation*}
On the other hand,
\begin{eqnarray*}
\int_{\Vert m_0\Vert_\infty}^{2\Vert m_0\Vert_\infty}  \left(1- \frac{1}{\nx} \sum_{i=1}^{\nx}
   \one_{ \{\widehat m_{\nx, \ny}(X_i)  \le w \}}\right)dw
&\geq&\Vert m_0\Vert_\infty \nx^{-1} \sum_{i=1}^{\nx}
   \one_{ \{\widehat m_{\nx, \ny}(X_i)  >2 \Vert m_0\Vert_\infty \}}
\end{eqnarray*}
and therefore,
\begin{eqnarray*}
 \sum_{i=1}^{\nx}
   \one_{ \{\widehat m_{\nx, \ny}(X_i)  >2 \Vert m_0\Vert_\infty \}}=o_P(\nx).
\end{eqnarray*}
By monotonicity of $ \widehat m_{\nx, \ny}$ this implies that
\begin{eqnarray*}
 \one_{\{ \widehat m_{\nx, \ny}(b)>2\Vert m_0\Vert_\infty \}}\sum_{i=1}^{\nx}
   \one_{ \{X_i> b \}}
&\leq& \sum_{i=1}^{\nx}
   \one_{ \{\widehat m_{\nx, \ny}(X_i)  >2 \Vert m_0\Vert_\infty \}}=o_P(\nx).
\end{eqnarray*}
By the law of large numbers, $\nx^{-1}\sum_{i=1}^{\nx}
   \one_{ \{X_i> b \}}$ converges in probability to $1-F_0(b)>0$ and therefore, it follows from the previous display that
$$ \one_{\{ \widehat m_{\nx, \ny}(b)>2\Vert m_0\Vert_\infty \}}=o_P(1).$$
This implies that
$$\lim_{n\to\infty}P( \widehat m_{\nx, \ny}(b)>2\Vert m_0\Vert_\infty )=0.$$
One can prove similarly that  
$$\lim_{n\to\infty}P( \widehat m_{\nx, \ny}(a)<-2\Vert m_0\Vert_\infty )=0.$$
This implies  \eqref{eq: AnBn} by monotonicity of $\widehat m_n$.  \\

{\it Proof of Claim 2.}
Now, we turn to the case where $\epsilon$ is supported on a finite set of points.  Without loss of generality we assume that $\Vert \epsilon \Vert_\infty \le 1$ with probability 1.  The same proof of Claim 4 in Proposition \ref{prop: exist} can be again used to show existence with probability 1 of an estimator $\wh m_{n_x, n_y}$ which is bounded in the sup-norm by $K_0+2$.  Now, the rate obtained above in (\ref{rateHnxny}) and the Cauchy-Schwarz inequality allow us to write that 
\begin{eqnarray*}
\int_{\mathbb R} \left \vert \widehat{\mathbb{H}}_{n_x, n_y}(y)  -  \mathbb{H}^0_{n_x}(y)  \right \vert  dy  & = &   \int_{-K_0-3}^{K_0+3} \left \vert \widehat{\mathbb{H}}_{n_x, n_y}(y)  -  \mathbb{H}^0_{n_x}(y)  \right \vert  dy  \\
& = &  O_P( (n_x \wedge n_y)^{-1/2}).
\end{eqnarray*}
Using the same arguments as in the proof of Theorem \ref{theo: nonoise} %
(in particular Proposition~\ref{prop:1} in Appendix~\ref{sec:wass-dist-lemm}) 
this implies that $W_1(\widehat{\mathbb{H}}_{n_x, n_y}, \mathbb H^0_{n_x})  =  O_P((n_x \wedge n_y)^{-1/2})$.  Since $\widehat{\mathbb H}_{n_x, n_y}(y)  = \int_{\mathbb R} \widehat{\mathbb{L}}_{n_x, n_y}(y-t)  d\Phi_\epsilon(t)$ and $\mathbb H^0_{n_x}(y)  =  \int_{\mathbb R} \mathbb{L}_{n_x}(y-t)  d\Phi_\epsilon(t)$, we can use
again
Proposition~\ref{prop:2} in Appendix~\ref{sec:wass-dist-lemm}
to find a constant $D > 0$ depending only on the distribution of $\epsilon$ and $K_0$ such that
\begin{eqnarray*}
W_1(\widehat{\mathbb{L}}_{n_x, n_y}, \mathbb{L}_{n_x})  \le D  \  W_1(\widehat{\mathbb{H}}_{n_x, n_y}, \mathbb H^0_{n_x}).
\end{eqnarray*} 
implying that 
\begin{eqnarray*}
\int_{\mathbb R}  \left \vert \widehat{\mathbb{L}}_{n_x, n_y}(w)  - \mathbb{L}_{n_x}(w) \right \vert dw =  \int_{-K_0-2}^{K_0+2}  \left \vert \widehat{\mathbb{L}}_{n_x, n_y}(w)  - \mathbb{L}_{n_x}(w) \right \vert dw  =  O_P( (n_x \wedge n_y)^{-1/2}) 
\end{eqnarray*}
using again Proposition 2 of \cite{Meis}.  Therefore, 
\begin{eqnarray*}
\int_{-K_0-2}^{K_0+2}  \vert \widehat L^0_{n_x, n_y}(w)  -  L_0(w)  \vert dw  &= & \int_{-K_0-2}^{K_0+2}  \vert F_0 \circ \widehat m^{-1}_{n_x, n_y}(w)  -   F_0 \circ m^{-1}_0(w) \vert dw   \\
& = &   O_P( (n_x \wedge n_y)^{-1/2}) 
\end{eqnarray*}
and hence
\begin{align*}
  \int_0^1 \vert \widehat m_{n_x, n_y}(x)  -  m_0(x) \vert dF_0(x)
  &  \le     \int_{\widehat m_{n_x, n_y}(0) \wedge m_0(0)}^{\widehat m_{n_x, n_y}(1) \vee m_0(1)}  \vert F_0 \circ \widehat m^{-1}_{n_x, n_y}(w)  -   F_0 \circ m^{-1}_0(w) \vert dw \\
  & \le   \int_{-K_0-2}^{K_0+2}  \vert F_0 \circ \widehat m^{-1}_{n_x, n_y}(w)  -   F_0 \circ m^{-1}_0(w) \vert dw  \\
  & =   O_P( (n_x \wedge n_y)^{-1/2}),  
\end{align*}
which completes the proof of Theorem \ref{theo: rateunequalsamplesizes}.    \hfill $\Box$

\medskip
\medskip

\noindent \textbf{Proof of Theorem \ref{theo: unifcons}.}  Theorem \ref{theo: rate}, Theorem \ref{rem:1}, or Theorem \ref{theo: nonoise} imply that 
along any subsequence
of $\{ n \}_{n=1}^\infty$   we can find a further subsequence  $\{ n_i \}_{i=1}^\infty$
such that with probability $1$  
\begin{eqnarray*}
\lim_{i  \to \infty}  \int_{a}^{b}  \vert \widehat{m}_{n_i}(x) -  m_0(x)  \vert  dF_0(x)  =0.
\end{eqnarray*}
Call this event $E_{1}$.
For notational ease let $\wh{m}_{n_i} \equiv \wh{m}_i$.
Further, by Corollary~2.3 of
\cite{Stein:2005vn} 
(stated for Lebesgue measure but the proof does not rely on the Lebesgue measure at all and the result holds for a general measure space),
there exists another subsubsequence
(which we call again $\{ n_i \}_{i=1}^\infty$ for convenience)
such that $\wh{m}_i$ converges $F_0$-a.e.\ to $m_0$.

  Recall $C$ is a compact in the interior of $ [a,b] \cap \mc{S}_0$.
  Then since $m_0$ is continuous on $C$, $\wh{m}_i$ converges on a dense subset of
  $[a,b] \cap \mc{S}_0$
  to $m_0$ (for any points $\alpha,\beta \in [a,b] \cap \mc{S}_0$, the $F_0$ measure of $(\alpha,\beta)$ is given by $F_0(\beta) - F_0(\alpha)$, so if $F_0(\beta) - F_0(\alpha) > 0$ then  there must be a point of convergence in $(\alpha,\beta)$, since convergence is $F_0$-a.e.), and both $\wh{m}_i$ and $m_0$ are monotone, it follows that $\wh{m}_i$ converges pointwise on all of $C$ to $m_0$ (one can sandwich any point in $C$, including its boundary points, by sequences of points above and below at which $\wh{m}_i$ converges to $m_0$ and appeal to monotonicity).

  And, we can strengthen the convergence to uniform convergence on $C$, since $m_0$ and $\wh{m}_i$ are monotone and $m_0$ is continuous on $C$,
  again, for any
  $\omega \in E_{1}$.
  The elementary proof is as follows.  Fix $\epsilon > 0$. By uniform continuity of $m_0$ on (the compact) $[a,b]$, there exists $\delta>0$ such that $|m_0(x)-m_0(y)|\leq\epsilon$ for all $x,y$ such that $|x-y|\leq\delta$.  Cover $C$ with the set of open
  (in $C$'s subspace topology)
  sets $A(x,\delta):= \{ y \in C: |x-y| < \delta/2\}$ for all $x \in C$, and extract by compactness a finite subcover of these open  sets, $A(x_i,\delta)$, for $i=1,\ldots, N$.
  Let $x_{i1} := \inf {A}(x_i,\delta)$ and $x_{i2}:= \sup {A}(x_i, \delta)$.  Since $C$ is closed, $x_{ij} \in C$, $j=1,2$.
  By pointwise convergence $\wh{m}_n(x_{ij})$ is within $\epsilon $ of $m_0(x_{ij})$, $j=1,2$, for all $i$ and $n$ large enough.  Now, take any $x \in C$; let $j$ be such that $x_{j1} \le x \le x_{j2}$.  Using monotonicity of $\wh{m}_n$ and of $m_0$, 
  we have for $n$ large enough that $\wh{m}_n(x) \le \wh{m}_n(x_{j2}) \le m_0(x_{j2}) + \epsilon \le m_0(x) + 2 \epsilon$.  Similarly, using $x_{j1}$, we have $\wh{m}_n(x) \ge m_0(x) - 2\epsilon$, which proves the uniform convergence on $C$.

  Hence, for all $\omega\in E_1$, for all subsequences of $\{ n_i \}_{i=1}^\infty$, we can find a further subsequence
  (depending on $\omega$)
  along which $\sup_{x \in C } \vert \widehat{m}_{i}(x) - m_0(x) \vert $ converges to zero. Hence, for all $\omega\in E_1$, this supremum distance converges to zero along the subsequence $\{ n_i \}_{i=1}^\infty$. Therefore,
  \begin{equation*}
    P\left(\lim_{i \to \infty} \sup_{x \in C }  \vert \widehat{m}_{i}(x) -  m_0(x)  \vert = 0\right) =   1.
  \end{equation*}
  Since for any subsequence of $\{n\}_{n=1}^\infty$,
  we have almost sure convergence along a subsubsequence, it follows that along the original sequence $\{ n\}$
\begin{eqnarray*}
\sup_{x \in C }  \vert \widehat{m}_n(x) -  m_0(x)  \vert \to 0, \ \textrm{in probability}
\end{eqnarray*}
which completes the proof.   \hfill $\Box$

\medskip

\subsection{Proofs for Section \ref{sec:morerates}}
We begin with  an auxiliary lemma. Below, we use the same notation $\widehat{\mathbb{H}}_n$ and $ \mathbb{H}^0_n$ as in \eqref{eq: Hn}. We recall that $H_0$ denotes the distribution function of $Y$. We use the same notation $P_n^X$ and $P^X$ as in the proof of Proposition \ref{prop: E1andE3andE2} and denote, moreover, by $E^X$ the expectation corresponding to $P^X$.

\begin{lemma}\label{lem: parametric} 
   Let Assumptions~\ref{assm:A0:basic:m0-Xi} and~\ref{assm:A1:m0-sup} hold. Assume that $\epsilon$ is supported on  $[-1;1]$ and independent of $X$. For any solution $\wh m_n\in\cal M$ to \eqref{Minim2} that is  bounded in the sup-norm by $K_0+2$, we then have
\begin{eqnarray}\label{eq: lemparametric}\notag
\int_\RR  \Big \{ H_0(y)  - E^X[\Phi_\epsilon (y  -  \widehat m_n(X))] \Big \}^2   dy
&\leq& 2\int_\RR \Big \{ \mathbb{H}_n(y)  -  \widehat{\mathbb{H}}_n(y) \Big \}^2   dy
+O_P(n^{-1})\\
&=&O_P(n^{-1}).
\end{eqnarray}
Moreover, 
\begin{eqnarray}\label{eq: parametric2}
\left\{\int\int _\RR  \left( \Phi_\epsilon(y-m_0(x))  - \Phi_\epsilon(y  -  \wh m_n(x))\right)dydF_0(x)   \right\}^2=O_P(n^{-1}).
\end{eqnarray}
\end{lemma} 

\medskip

\begin{proof}[Proof of Lemma \ref{lem: parametric}.]
We can write
\begin{eqnarray*}
n^{-1} \sum_{i=1}^n \Phi_\epsilon(y-\wh m_n(X)) - E^X\left[\Phi_\epsilon(y  -  \wh m_n(X)) \right] =  \left(\mathbb P^X_n   - P^X \right) \Phi_\epsilon(y-  \wh m_n(\cdot)).
\end{eqnarray*}
For a fixed $y$, consider the class of non-decreasing functions
\begin{eqnarray*}
\Big \{x \mapsto -\Phi(y  -  m(x)), x \in [0,1], m \ \textrm{monotone non-decreasing and $\Vert  m \Vert_\infty \le K_0+2$}\Big\}. 
\end{eqnarray*}
Using entropy arguments as in the proof of Proposition \ref{prop: E1andE3andE2} by replacing the class $\mathcal I$ with the one defined above we can show that for all random variables $A<B$ such that $B-A = O_P(1)$, it holds that
\begin{eqnarray*}
\int_A^B  \Big \{ n^{-1}  \sum_{i=1}^n \Phi_\epsilon(y  -  \wh m_n(X_i))   -  E^X\left[\Phi_\epsilon(y  -  \wh m_n(X)) \right]  \Big \}^2   dy=O_P(n^{-1}).
\end{eqnarray*}
Since $\epsilon$ is supported on $[-1,1]$, and $\wh m_n$ is assumed to be bounded in the sup-norm by $K_0+2$, the above integral over $[A,B]$ with $A=-K_0-3$ and $B=K_0+3$, is equal to the same integral over the whole real line $\mathbb{R}$. Hence, we get
\begin{eqnarray}\label{eq: parametric}
\int_\RR  \Big \{ n^{-1}  \sum_{i=1}^n \Phi_\epsilon(y  -  \wh m_n(X_i))   - E^X\left[\Phi_\epsilon(y  -  \wh m_n(X)) \right]  \Big \}^2   dy=O_P(n^{-1}).
\end{eqnarray}
On the other hand, with $H_0$ the distribution function of $Y$ we have
\begin{eqnarray*}
E\int_{\mathbb R}(\mathbb{H}_n(y)-H_0(y))^2 dy&=&\int_{\mathbb R}E(\mathbb{H}_n(y)-H_0(y))^2 dy\\
&=&n^{-1}\int_{\mathbb R}H_0(y)(1-H_0(y)) dy
\end{eqnarray*}
since $n\mathbb{H}_n(y)$ is a binomial random variable with parameter $n$ and probability of success $H_0(y)$. The integral on the right hand side is finite since $Y$ has bounded support (included in $[-K_0-1,K_0+1]$) and therefore,
\begin{eqnarray*}
\int_{\mathbb R}(\mathbb{H}_n(y)-H_0(y))^2 dy&=&O_P(n^{-1}).
\end{eqnarray*}
Combining this with \eqref{eq: parametric} together with the fact that for all real numbers $a$ and $b$, we have $(a+b)^2\leq 2a^2+2b^2$, we conclude that 
\begin{eqnarray*}
&&\int_\RR  \Big \{ H_0(y)  - E^X\left[\Phi_\epsilon(y  -  \wh m_n(X)) \right] \Big \}^2   dy\\
&&\qquad
\leq 2\int_\RR \Big \{ \mathbb{H}_n(y)  -  n^{-1}  \sum_{i=1}^n \Phi_\epsilon(y  -  \widehat m_n(X_i)) \Big \}^2   dy
+O_P(n^{-1}).
\end{eqnarray*}
The first inequality follows by definition of $\widehat{\mathbb{H}}_n$. 

Now, it follows from the definition of $\widehat m_n$ that 
\begin{eqnarray*}
 \int_{\mathbb{R}}  \Big \{ \mathbb{H}_n(y)  -  \widehat{\mathbb{H}}_n(y ) \Big \}^2   dy&\leq&  \int_{\mathbb{R}}  \Big \{ \mathbb{H}_n(y)  -  {\mathbb{H}}_n^0(y )
 \Big \}^2   dy\\
&\leq& 2 \int_{\mathbb{R}}  \Big \{ \mathbb{H}_n(y)  -   H_0(y) \Big \}^2   dy +2 \int_{\mathbb{R}}  \Big \{ \mathbb{H}_n^0(y)  -   H_0(y) \Big \}^2 dy. 
\end{eqnarray*}
Since $\epsilon$ is independent of $X$,  both $n \mathbb{H}_n(y)$ and $n {\mathbb{H}}_n^0(y )$ are the average of $n$  i.i.d. bounded random variables with mean $H_0(y)$ and therefore,
\begin{equation*}
E \int_{\mathbb{R}}  \Big \{ \mathbb{H}_n(y)  -  \widehat{\mathbb{H}}_n(y ) \Big \}^2   dy \leq  2n^{-1} \int_{\mathbb{R}} Var(1_{\{Y\leq y\}}) dy +2n^{-1} \int_{\mathbb{R}} Var(\Phi_\epsilon(y-m_0(X)) dy. 
\end{equation*}
Both integrals on the right-hand side are finite since the integrands are bounded and equal to zero for all $y\leq -K_0-1$ and all $y\geq K_0+1$. Hence,
\begin{eqnarray*}
\int_{\mathbb{R}}  \Big \{ \mathbb{H}_n(y)  -  \widehat{\mathbb{H}}_n(y ) \Big \}^2   dy= O_P(n^{-1}).
\end{eqnarray*}
This completes the proof of \eqref{eq: lemparametric}.

Now, with $F_0$ the distribution function of $X$, it follows from the assumption that $\widehat m_n$ (as well as $m_0$) is bounded in sup-norm by $K_0+2$ that
\begin{eqnarray*}
&&\int_\RR  \Big \{ H_0(y)  - E^X\left[\Phi_\epsilon(y  -  \wh m_n(X)) \right] \Big \}^2   dy\\
&&\qquad= \int_\RR  \Big \{ \int \left(\Phi_\epsilon(y-m_0(x))  - \Phi_\epsilon(y  -  \widehat m_n(x))\right) dF_0(x)\Big \}^2   dy\\
&&\qquad= \int_{-K_0-3}^{K_0+3}  \Big \{ \int \left(\Phi_\epsilon(y-m_0(x))  - \Phi_\epsilon(y  -  \widehat m_n(x))\right) dF_0(x)\Big \}^2   dy
\end{eqnarray*}
so it follows from the Jensen inequality and the Fubini theorem that
\begin{eqnarray*}
&&\int_\RR  \Big \{ H_0(y)  - E^X\left[\Phi_\epsilon(y  -  \wh m_n(X)) \right]  \Big \}^2   dy\\
&&\qquad\geq \frac{1}{2K_0+6}\left\{\int_\RR \int  \left( \Phi_\epsilon(y-m_0(x))  - \Phi_\epsilon(y  -  \widehat m_n(x))\right)dF_0(x)   dy\right\}^2\\
&&\qquad= \frac{1}{2K_0+6}\left\{\int\int _\RR  \left( \Phi_\epsilon(y-m_0(x))  - \Phi_\epsilon(y  -  \widehat m_n(x))\right)dydF_0(x)   \right\}^2.
\end{eqnarray*}
Combining this with \eqref{eq: lemparametric}  yields \eqref{eq: parametric2} and completes the proof of the lemma.
\end{proof}

\medskip

\par \noindent \textbf{Proof of  Theorem \ref{theo: parametricrate}.}  \

  {\it Case of error with finite support:} Suppose that $\epsilon$ is supported on a finite set such that all the points of the support belong to $[-1,1]$ and let $k \ge 2$. For any solution $\wh m_n$ to (\ref{Minim2}) which is bounded by $K_0+2$ (which exists with probability $1$ in view of Proposition \ref{prop: exist}), it follows  from Theorem \ref{theo: nonoise} and the identity $a^k -  b^k = (a-b) (a^{k-1} + a^{k-2} b + \ldots +  b^{k-1})$  that
\begin{eqnarray*}
\left \vert \int_0^1 \left(\widehat m^k_n(x)  - m^k_0(x) \right) dF_0(x)  \right\vert & \le  & k (K_0 +2)^{k-1}\left(\int_0^1 \vert \wh m_n(x)  - m_0(x)  \vert dF_0(x)\right) \\
& = &   O_P(1/\sqrt n).
\end{eqnarray*}
To show \eqref{eq:moments}, it is enough to show that 
\begin{eqnarray}\label{emppmoment}
\int_0^1 \widehat m_n^k(x) d(\mathbb F_n(x)  - F_0(x))  \equiv   \frac{1}{\sqrt n} \ \mathbb G_n \widehat m^k_n 
\end{eqnarray}
is $O_P(1/\sqrt n)$. Define
\begin{eqnarray*}
\mathcal M_{c}  =   \{ m:  m \ \textrm{non-decreasing on $[0,1]$ and} \ \Vert m \Vert_\infty \le c\},
\end{eqnarray*}
and 
\begin{eqnarray*}
\mathcal G_{k, c}  =   \{ m^k:  m   \in \mathcal{M}_c \},
\end{eqnarray*}
for $c> 0$.   If $k=1$, then  $\widehat m_n \in \mathcal{M}_{K_0+2}$ and it follows from similar arguments as in the proof of Proposition \ref{prop: E1andE3andE2} that there exists $M > 0$ depending only on $K_0$ such that $\Vert \mathbb G_n \Vert_{\mathcal M_{K_0+2}}  = O_P(1) $ which implies from (\ref{emppmoment}) that  \eqref{eq:moments} holds true.  Now, suppose that $k \ge 2$. Using the decomposition $m^k  =  m^k \mathds{1}_{m \ge 0}  +   m^k \mathds{1}_{m < 0}$  we see that for any $m \in \mathcal{M}_{K_0+2}$,   $m^k$ is either the sum or the difference (depending on whether $k$ is odd or even) of two functions in $ \mathcal{M}_{(K_0+2)^k}$.
Using Proposition \ref{prop: preserv-2}  with $(\lambda_1, \lambda_2) = (1,-1)$ or $(\lambda_1, \lambda_2)  = (1,1)$, it follows that for any discrete measure $Q$ and $\delta > 0$
\begin{eqnarray*}
\log N(\delta, \mathcal{G}_{k, K_0+2}, L_2(Q))   \le  2 \log N(\delta/2, \mathcal{M}_{(K_0+2)^k}, L_2(Q)).
\end{eqnarray*}
Using \eqref{RelNumb} and similar arguments as in the proof of Proposition \ref{prop: E1andE3andE2}, we conclude again that  $\Vert \mathbb G_n \Vert_{\mathcal G_{k, K_0+2}}  = O_P(1)$.
Together with (\ref{emppmoment}), it follows that  \eqref{eq:moments} holds true.

\medskip

 {\it Case of uniform error:}
Now, we turn to the case where $\epsilon \sim \mathcal{U}[-1,1]$.  To compute the integral on the left-hand side of \eqref{eq: parametric2}, we distinguish between several cases. We recall that, because $\Phi_\epsilon$ is the distribution function of a uniformly distributed random variable over $[-1,1]$, $\Phi_\epsilon(t)$ is equal to 0 if $t<-1$, to 1 if $t>1$, and to $(t+1)/2$ otherwise. For all $x$ such that $m_0(x)\leq \widehat m_n(x)\leq m_0(x)+2$ we have
\begin{eqnarray*}
&&\Phi_\epsilon(y-m_0(x))  - \Phi_\epsilon(y  -  \wh m_n(x))\\
&&\qquad =\begin{cases}
(y-m_0(x)+1)/2&\mbox{ if }m_0(x)-1\leq y\leq\widehat m_n(x)-1\\
(\widehat m_n(x)-m_0(x))/2&\mbox{ if }\widehat m_n(x)-1\leq y\leq m_0(x)+1\\
1-(y-\widehat m_n(x)+1)/2&\mbox{ if }m_0(x)+1\leq y\leq\widehat m_n(x)+1\\
0&\mbox{ otherwise }.\\
\end{cases}
\end{eqnarray*}
Hence, for all $x$ such that $m_0(x)\leq \widehat m_n(x)\leq m_0(x)+2$ we have
\begin{eqnarray*}
&&\int_\RR\left(\Phi_\epsilon(y-m_0(x))  - \Phi_\epsilon(y  -  \wh m_n(x))\right)dy\\
&&\qquad =\frac12(\widehat m_n(x)-m_0(x))^2+\frac 12(\widehat m_n(x)-m_0(x))(m_0(x)-\widehat m_n(x)+2)\\
&&\qquad =(\widehat m_n(x)-m_0(x)).
\end{eqnarray*}
Similarly, for all $x$ such that $m_0(x)+2< \widehat m_n(x)$ we have
\begin{eqnarray*}
&&\Phi_\epsilon(y-m_0(x))  - \Phi_\epsilon(y  -  \wh m_n(x))\\
&&\qquad =\begin{cases}
(y-m_0(x)+1)/2&\mbox{ if }m_0(x)-1\leq y\leq m_0(x)+1\\
1&\mbox{ if }m_0(x)+1\leq y\leq\widehat m_n(x)-1\\
1-(y-\widehat m_n(x)+1)/2&\mbox{ if }\widehat m_n(x)-1\leq y\leq\widehat m_n(x)+1\\
0&\mbox{ otherwise,}\\
\end{cases}
\end{eqnarray*}
which implies that
\begin{eqnarray*}
\int_\RR\left(\Phi_\epsilon(y-m_0(x))  - \Phi_\epsilon(y  -  \wh m_n(x))\right)dy =(\widehat m_n(x)-m_0(x)).
\end{eqnarray*}
Hence,
\begin{eqnarray*}
&&\int1_{\{m_0(x)\leq \widehat m_n(x)\}}\int _\RR  \left( \Phi_\epsilon(y-m_0(x))  - \Phi_\epsilon(y  -  \widehat m_n(x))\right)dydF_0(x) \\
&&\qquad=\int1_{\{m_0(x)\leq \widehat m_n(x)\}}\left(\widehat m_n(x)-m_0(x)\right)dF_0(x)
\end{eqnarray*}
Similarly,
\begin{eqnarray*}
&&\int1_{\{\widehat m_n(x)\leq m_0(x)\}}\int _\RR  \left( \Phi_\epsilon(y-m_0(x))  - \Phi_\epsilon(y  -  \widehat m_n(x))\right)dydF_0(x) \\
&&\qquad=\int1_{\{ \widehat m_n(x)\leq m_0(x)\}}\left(\widehat m_n(x)-m_0(x)\right)dF_0(x)
\end{eqnarray*}
Combining the two previous displays yields
\begin{eqnarray*}
&&\int\int _\RR  \left( \Phi_\epsilon(y-m_0(x))  - \Phi_\epsilon(y  -  \widehat m_n(x))\right)dydF_0(x) \\
&&\qquad=\int\left(\widehat m_n(x)-m_0(x)\right)dF_0(x).
\end{eqnarray*}
Now, from \eqref{eq: parametric2} it follows that
\begin{eqnarray*}
\left \vert \int_0^1 \left(\widehat m_n(x)-m_0(x)\right)dF_0(x) \right \vert = O_P(1/\sqrt n).
\end{eqnarray*}
As we already know that $\int_0^1 \widehat m_n(x)  d(\mathbb F_n(x) - F_0(x))  = O_P(1/\sqrt n)$, the second claim of the proposition now follows. \hfill $\Box$

\section{Proof  of Proposition \ref{Fenchel}}

To make the notation less cumbersome, we write in the following $Z_i$ for the $i$-th order statistic $X_{(i)}$.  Suppose that $\widehat m_n$ takes at least two distinct values and let $ 1\le j < j' \le n$  be such that $\widehat m_n$ is constant on $[Z_j , Z_{j'})$, where  $Z_j  < Z_{j'}$ are  two successive jump points of $\widehat m_n$. Consider the function $m_\delta$ which is right-continuous, constant between the order statistics $Z_1, \ldots, Z_n$, and  
\begin{eqnarray*}
m_\delta(Z_i) = \left \{
\begin{array}{ll} 
 \widehat  m_n(Z_i)   +\delta, \  \  i \in \{j, \ldots, j'-1 \}  \\
 \widehat  m_n(Z_i),\  \  \  \  \  \  \   \textrm{otherwise}.
\end{array}
\right.
\end{eqnarray*}
Then, the function $m_\delta$ as defined above belongs to $\mathcal{M}$, provided that $\vert \delta \vert $ is small enough. It follows from the definition \eqref{Minim2}  that $\mathbb{M}_n(m_\delta) \ge \mathbb{M}_n(\widehat m_n)$. Using Taylor expansion of $\Phi_\epsilon$  with the integral remainder term we can write that for $i \in \{j, \ldots, j'-1\}$
\begin{eqnarray*}
  \Phi_\epsilon(y - \widehat m_n(Z_i) -\delta)  =   \Phi_\epsilon(y - \widehat m_n(Z_i))  - \delta  f_\epsilon(y - \widehat m_n(Z_i)  ) +  R_{\delta, i}(y)
\end{eqnarray*}
where  the remainder term $R_{\delta, i}$  is given below. Hence, 
\begin{eqnarray}\label{eq: signMn}
  0  & \leq & \mathbb{M}_n(m_\delta) -  \mathbb{M}_n(\widehat m_n)    \\ \notag
     & = &   \int_{\RR}  \Big \{  \HH_n(y)  -\frac{1}{n} \sum_{i \notin \{j, \ldots, j'-1\}}  \Phi_\epsilon(y - \widehat m_n(Z_i)) - \frac{1}{n} \sum_{i=j}^{j'-1}  \Phi_\epsilon(y - \widehat m_n(Z_i) -\delta)   \Big \}^2  dy \\ \notag
     &&   - \int_{\RR}  \Big \{  \HH_n(y)  -\frac{1}{n} \sum_{i=1}^n \Phi_\epsilon(y - \widehat m_n(X_i) )  \Big \}^2 dy \\ \notag
     & = &   \int_{\RR}  \Big \{  \HH_n(y)  -\frac{1}{n} \sum_{i \notin \{j, \ldots, j'-1\}}  \Phi_\epsilon(y - \widehat m_n(Z_i)) \\ \notag
     &&  \hspace{1.4cm} - \frac{1}{n} \sum_{i=j}^{j'-1}  \Phi_\epsilon(y - \widehat m_n(Z_i) )  + \delta  \frac{1}{n}\sum_{i=j}^{j'-1}f_\epsilon(y - \widehat m_n(Z_i))  - \frac{1}{n} \sum_{i=j}^{j'-1} R_{\delta, i} (y) \Big \}^2  dy \\ \notag
     &&   - \int_{\RR}  \Big \{  \HH_n(y)  -\frac{1}{n} \sum_{i=1}^n \Phi_\epsilon(y - \widehat m_n(X_i) )  \Big \}^2 dy
\end{eqnarray}
which equals
\begin{align*}
  \MoveEqLeft \int_{\RR}  \Big \{  \HH_n(y)  -\frac{1}{n} \sum_{i =1}^n  \Phi_\epsilon(y - \widehat m_n(X_i))  + \delta  \frac{1}{n}\sum_{i=j}^{j'-1}f_\epsilon(y - \widehat m_n(Z_i))  - \frac{1}{n} \sum_{i=j}^{j'-1} R_{\delta, i} (y) \Big \}^2  dy \\
  & \quad - \int_{\RR}  \Big \{  \HH_n(y)  -\frac{1}{n} \sum_{i=1}^n \Phi_\epsilon(y - \widehat m_n(X_i) )  \Big \}^2 dy
\end{align*}
which equals
\begin{align*}
  \MoveEqLeft \frac{2}{n} \int_{\RR}  \Big (  \HH_n(y)  -\frac{1}{n} \sum_{i=1}^n \Phi_\epsilon(y - \widehat m_n(X_i) ) \Big ) \Big(\delta \sum_{i=j}^{j'-1}  f_\epsilon(y - \widehat m_n(Z_i))  - \sum_{i=j}^{j'-1}  R_{\delta, i}(y)  \Big)   dy \\
  & \quad +  \frac{1}{n^2}\int_{\RR} \Big(\delta \sum_{i=j}^{j'-1}  f_\epsilon(y - \widehat m_n(Z_i))  - \sum_{i=j}^{j'-1}  R_{\delta, i}(y)  \Big)^2  dy
\end{align*}
which equals
\begin{align*}
\MoveEqLeft  \frac{2}{n}  \int_{\RR}  \Big (  \HH_n(y)  -\frac{1}{n} \sum_{i=1}^n \Phi_\epsilon(y - \widehat m_n(X_i) ) \Big ) \Big(\delta (j'-j) f_\epsilon(y - \widehat m_n(Z_j))  - \sum_{i=j}^{j'-1}  R_{\delta, i}(y)  \Big)   dy\\
  &\quad   +  \frac{1}{n^2}\int_{\RR} \Big( \delta (j'-j) f_\epsilon(y - \widehat m_n(Z_j))  - \sum_{i=j}^{j'-1}  R_{\delta, i}(y)  \Big)^2  dy  
\end{align*}
where for $i =j, \ldots, j'-1$
\begin{eqnarray*}
R_{\delta, i}(y) & = &  \int_{ y - \widehat m_n(Z_i)}^{ y - \widehat m_n(Z_i) - \delta}  f'_{\epsilon}(t) \cdot ( y - \widehat m_n(Z_i)  - \delta - t)  dt \\
& = & - \int_{0}^{\delta}  f'_{\epsilon}(y - \widehat m_n(Z_i) - u) \cdot (u - \delta) \ du, \  \text{letting $u = y - \widehat m_n(Z_i) -t $}   \\
                 &  = & - \delta^2  \int_{0}^1  f'_{\epsilon}( y - \widehat m_n(Z_i) -\delta v) \  (v-1) \ dv , \ \text{letting $v = u/\delta$}.
\end{eqnarray*}
Thus,
$(\mathbb{M}_n(m_\delta) -  \mathbb{M}_n(\widehat m_n))/\delta $ equals
\begin{align*}
  \MoveEqLeft
  \frac{2}{n}  (j'-j) \int_{\RR}  \Big (  \HH_n(y)  -\frac{1}{n} \sum_{i=1}^n \Phi_\epsilon(y - \widehat m_n(X_i) ) \Big )    f_\epsilon(y - \widehat m_n(Z_j)) dy \\   
  & \ -  \ \frac{2}{n}  \int_{\RR}  \Big (  \HH_n(y)  -\frac{1}{n} \sum_{i=1}^n \Phi_\epsilon(y - \widehat m_n(X_i) ) \Big )   \frac{1}{\delta} \sum_{i=j}^{j'-1}  R_{\delta, i}(y)   dy \\
  &  \ + \  \frac{1}{n^2}\int_\RR \bigg\{ \delta (j'-j)^2 f_\epsilon(y - \widehat m_n(Z_j))^2  - 2 (j'-j) f_\epsilon(y - \widehat m_n(Z_j) \sum_{i=j}^{j'-1}  R_{\delta, i}(y)  \\
  & \qquad \qquad \qquad   +    \frac{\Big(\sum_{i=j}^{j'-1}  R_{\delta, i}(y)\Big)^2}{\delta}  \bigg\}  dy 
\end{align*}
which equals
\begin{align*}
  \MoveEqLeft 
  \frac{2}{n}  (j'-j)  \int_{\RR}  \Big (  \HH_n(y)  -\frac{1}{n} \sum_{i=1}^n \Phi_\epsilon(y - \widehat m_n(X_i) ) \Big )   f_\epsilon(y - \widehat m_n(Z_j)) dy   \\
  & \  \  +  \ \frac{2 \delta}{n}  \int_{\RR}  \Big (  \HH_n(y)  -\frac{1}{n} \sum_{i=1}^n \Phi_\epsilon(y - \widehat m_n(X_i) ) \Big ) \times \\
  & \qquad \qquad \Big(  \sum_{i=j}^{j'-1}  \int_0^1   f'_{\epsilon}( y - \widehat m_n(Z_i) - \delta v)  \  (v-1) \   dv  \Big)  dy \\   
  &  \  +  \  \frac{\delta (j'-j)^2}{n^2}   \int_{\RR}  f_\epsilon(y - \widehat m_n(Z_j))^2 dy  \\
  & \ +  \ \frac{2 \delta  (j'-j)}{n^2}  \int_\RR f_\epsilon(y - \widehat m_n(Z_j) ) \sum_{i=j}^{j'-1} \int_{0}^1    f'_{\epsilon}( y - \widehat m_n(Z_i) -\delta v) \  (v-1) \ dv  dy  \\
  & \ +   \ \frac{\delta^3}{n^2}  \int_{\RR} \bigg(  \int_0^1 \sum_{i=j}^{j'-1}   f'_{\epsilon}( y - \widehat m_n(Z_i) -\delta v) \  (v-1) \ dv \bigg)^2 dy.
\end{align*}
We show below that each term on the right hand side that depends on $\delta$ takes the form of $\delta^i, i =1,2,3$  times a finite integral, so that it tends to zero as $\delta\to 0$.
From Assumption \ref{assm:a6}, it follows that
\begin{eqnarray*}
&&\left \vert  \Big (  \HH_n(y)  -\frac{1}{n} \sum_{i=1}^n \Phi_\epsilon(y - \widehat m_n(X_i) ) \Big ) \cdot  \Big(  \sum_{i=j}^{j'-1}  \int_0^1   f'_{\epsilon}( y - \widehat m_n(Z_i) - \delta v)  \  (v-1) \   dv  \Big)  \right \vert \\
&&  \le   D  (j'-j) \ \Big \vert \HH_n(y)  -\frac{1}{n} \sum_{i=1}^n \Phi_\epsilon(y - \widehat m_n(X_i) ) \Big \vert   
\end{eqnarray*}
which can be shown to be integrable  on $\RR$ using the property of $\Phi_\epsilon$ in  (\ref{Integrable}).   Also, Assumption \ref{assm:a6} implies that there exists $D' > 0$ such that $\sup_{t \in \RR}  f_\epsilon(t)  \le D'$. Then,
\begin{eqnarray*}
\int_{\RR}  f_\epsilon(y - \widehat m_n(Z_j))^2 dy  \le D' \int_{\RR}  f_\epsilon(y - \widehat m_n(Z_j)) dy = D',
\end{eqnarray*}
and by Fubini's Theorem  
\begin{eqnarray*}
&& \int_{\RR} \sum_{i=j}^{j'-1} \int_{0}^1   f_\epsilon(y - \widehat m_n(Z_j) )  \vert f'_{\epsilon}( y - \widehat m_n(Z_i) -\delta v)\vert \  (v-1) \ dv  dy \\
 && =   \sum_{i=j}^{j'-1} \int_{0}^1  \left( \int_\RR f_\epsilon(y - \widehat m_n(Z_j) )  \vert f'_{\epsilon}( y - \widehat m_n(Z_i) -\delta v)\vert  dy  \right)  (v-1)  dv \\
& &\le D \sum_{i=j}^{j'-1} \int_{0}^1 (v-1)  dv  =  \frac{D  (j'-j)}{2}
\end{eqnarray*}
using Assumption \ref{assm:a6} and the fact that $f_\epsilon$ is a density. Finally, using again Assumption \ref{assm:a6} and Fubini's Theorem we have 
\begin{eqnarray*}
  && \int_{\RR}  \bigg(  \int_0^1 \sum_{i=j}^{j'-1}   f'_{\epsilon}( y - \widehat m_n(Z_i) -\delta v)  \  (v-1) \ dv \bigg)^2 dy \\
  && \le   \frac{D(j'- j)}{2} \int_{\RR}   \int_0^1 \sum_{i=j}^{j'-1}  \vert f'_{\epsilon}( y - \widehat m_n(Z_i) -\delta v)\vert  \  (1-v) \ dv  dy
\end{eqnarray*}
which equals
\begin{align*}
   \MoveEqLeft
  \frac{D(j'- j)}{2} \int_0^1   \left(\sum_{i=j}^{j'-1}  \int_{\RR} \vert f'_{\epsilon}( y - \widehat m_n(Z_i) -\delta v) \vert dy \right) (1-v) dv 
   \\
   &
      =  \frac{D(j'- j)^2}{4}  \int_{\RR} \vert f'_{\epsilon}( t) \vert dt < \infty,
\end{align*}
by Assumption \ref{assm:a6}.
By using \eqref{eq: signMn} and distinguishing between the cases of positive and negative  values of $\delta$ it follows that  
\begin{eqnarray*}
  0&=&\lim_{\delta \to 0} \frac{\mathbb{M}_n(m_\delta) -  \mathbb{M}_n(\widehat m_n)}{\delta} \\
   &=&   \frac 2n (j'-j) \int_{\RR}  \Big (  \HH_n(y)  -\frac{1}{n} \sum_{i=1}^n \Phi_\epsilon(y - \widehat m_n(X_i) ) \Big )   f_\epsilon(y - \widehat m_n(Z_j)) dy 
\end{eqnarray*}
and therefore,
\begin{eqnarray*}
0= \int_{\RR}  \Big( \HH_n(y)   - n^{-1}  \sum_{i=1}^n  \  \Phi_\epsilon(y - \widehat m_n(X_i))\Big)    f_\epsilon(y - \widehat m_k)   dy 
\end{eqnarray*}
where $\widehat m_k  =  \widehat m_n(Z_j) =  \ldots =  \widehat m_n(Z_{j'-1})$. This is precisely the condition given in (\ref{Cond1}). 

In the case $\widehat m_n$ takes a unique value, a similar reasoning give the same result, characterizing $\widehat m_k$ for $k=1$. 




%
  
Now,  the alternative expression in (\ref{AltCond2}) follows from the fact that for any $a \in \RR$
\begin{equation*}
  \int_{\RR}  \HH_n(y)  f_\epsilon(y -  a) dy =    \frac{1}{n} \sum_{i=1}^n \int^{\infty}_{Y_i} f_\epsilon(y-a) dy
   =  1- \frac{1}{n} \sum_{i=1}^n   \Phi_\epsilon(Y_i-a)
\end{equation*}
which completes the proof. \hfill $\Box$

\begin{figure}[h]
  \centering
  \includegraphics[scale=.5]{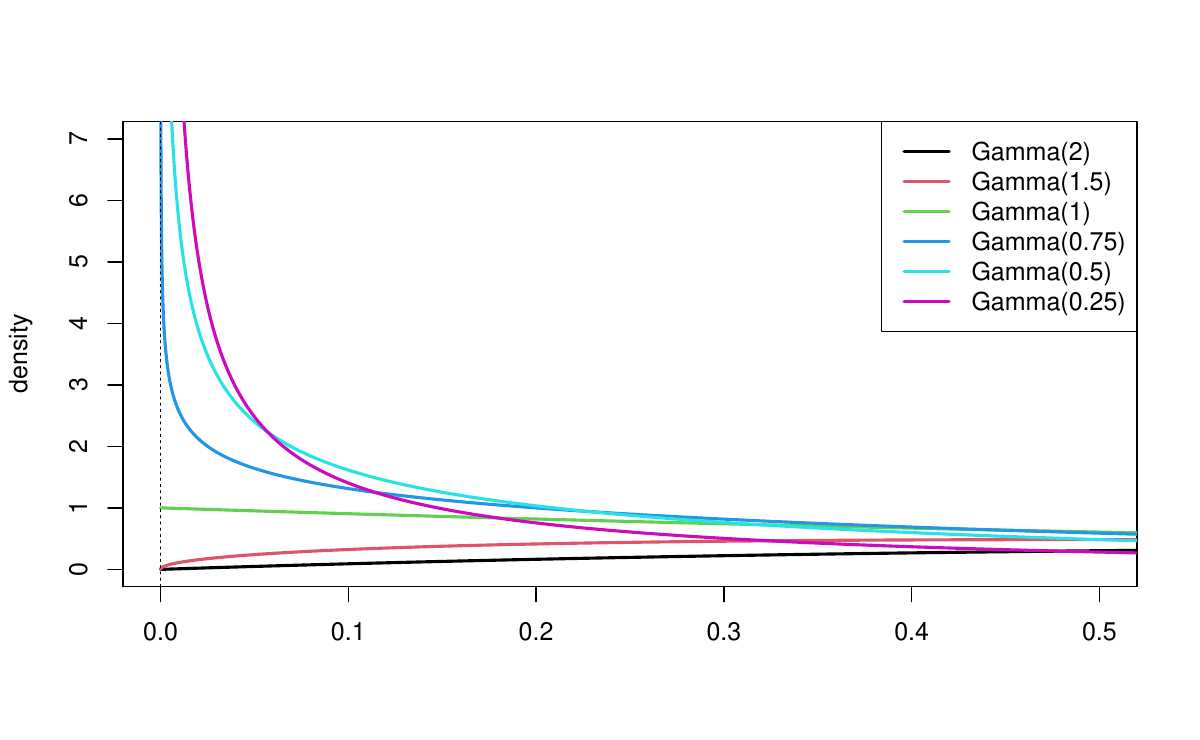}   
  \caption[Gamma density plots.]{Plots of gamma densities with varying shape parameter.}
  \label{fig:gamma-density}
\end{figure}

\bibliography{shuffreg2}

\end{document}